\documentclass[aps,superscriptaddress,nofootinbib,showpacs,notitlepage, twocolumn]{revtex4-1}

\usepackage{amsfonts}
\usepackage{amsmath}
\usepackage{amsthm}
\usepackage{braket}
\usepackage{graphicx}
\usepackage{makecell}
\usepackage{blkarray}
\usepackage[dvipsnames]{xcolor}
\usepackage{amssymb}
\usepackage{dsfont}
\usepackage{verbatim}
\usepackage{amsmath,amscd}
\usepackage[all,cmtip]{xy}
\usepackage{multirow}

\usepackage{blkarray}
\usepackage{booktabs}

\AtBeginDocument{
\lightrulewidth=.05em
\cmidrulewidth=.03em
\belowrulesep=.65ex
\belowbottomsep=0pt
\aboverulesep=.4ex
\abovetopsep=0pt
\cmidrulesep=\doublerulesep
\cmidrulekern=.5em
\defaultaddspace=.5em
}

\DeclareMathOperator{\Ima}{Im}

\definecolor{darkblue}{RGB}{0,0,127} % choose colors
\definecolor{darkgreen}{RGB}{0,130,80}
\definecolor{darkred}{RGB}{150,10,10}

\usepackage{afterpage}

\usepackage{calc}
\graphicspath{{./Figures/}}
\usepackage{tikz,pgfplots}\pgfplotsset{compat=newest}
\usepackage{xtab}
\usetikzlibrary{external,calc,decorations.pathreplacing,decorations.markings,decorations.pathmorphing,arrows.meta,shapes.geometric,cd}
\newlength\figureheight
\newlength\figurewidth

\newcommand{\Z}{\ensuremath{\mathbb{Z}}}

\newcommand{\R}[1]{{Ref.~\onlinecite{#1}}}
\newcommand{\strop}{\ensuremath{S}}

\newtheorem{lemma}{Lemma}

\newtheorem{definition}{Definition}

\newcommand{\drawgenerator}[8]{%
\xymatrix@!0{%
& #8 \ar@{-}[ld]\ar@{.}[dd] \ar@{-}[rr] & & #7 \ar@{-}[ld]  \\%
#1 \ar@{-}[rr] \ar@{-}[dd] &  & #2 \ar@{-}[dd] &            \\%
& #6 \ar@{.}[ld] &  & #5 \ar@{-}[uu] \ar@{.}[ll]       \\%
#3 \ar@{-}[rr] &  & #4 \ar@{-}[ru]                       %
}%
}

\usepackage{array,floatrow}
\usepackage{graphicx}
\usepackage[caption=false,label font={bf,normalsize}]{subfig}
\usepackage{hyperref}
\hypersetup{colorlinks,
	linkcolor=darkblue,
	citecolor=darkgreen,
	filecolor=red,
	urlcolor=blue,
	pdftitle={},
	pdfauthor={}
}
\begin{document}

\title{Sorting topological stabilizer models in three dimensions}
\author{Arpit Dua}
\affiliation{Department of Physics, Yale University, New Haven, CT 06520-8120, USA}
\affiliation{Yale Quantum Institute, Yale University, New Haven, CT 06520, USA}
\author{Isaac H. Kim}
\affiliation{Stanford Institute for Theoretical Physics, Stanford University, Stanford CA 94305 USA} 
\author{Meng Cheng}
\affiliation{Department of Physics, Yale University, New Haven, CT 06520-8120, USA}
\author{Dominic~J. Williamson}
\affiliation{Department of Physics, Yale University, New Haven, CT 06520-8120, USA}

\begin{abstract}
% Topological order for translation invariant Pauli stabilizer models in two dimensions is well-understood due to a structure theorem by Haah, which states that all such models are equivalent to copies of toric code. In two dimensions, there is an S-matrix invariant that encodes the topological data about commutation between string logical operators and equivalently, braiding between topological excitations. 
The S-matrix invariant is known to be complete for translation invariant topological stabilizer models in two spatial dimensions, as such models are phase equivalent to some number of copies of toric code. 
%as it indicates the number of copies of toric code that lies in the same topological phase of matter as the model.
%with topological order in two spatial dimensions are fully classified into phases of matter equivalent to copies of the toric code by their S-matrix invariant. 
In three dimensions, much less is understood about translation invariant topological stabilizer models due to the existence of fracton topological order. Here we introduce bulk commutation quantities inspired by the 2D S-matrix invariant that can be employed to coarsely sort 3D topological stabilizer models into qualitatively distinct types of phases: topological quantum field theories, foliated or fractal type-I models with rigid string operators, or type-II models with no string operators. 
\end{abstract}
  
\maketitle
In recent years the study of topological phases of matter has moved to the forefront of theoretical condensed matter physics. This has been fueled in part by the theorized existence of exotic phases of matter that can serve as topological quantum memories and computers~\cite{qdouble}. The classification of topological phases of matter in two spatial dimensions in terms of anyon theories, or modular tensor categories, and chiral central charges forms the cornerstone of the subject~\cite{Moore1988,kitaev2006anyons}. For the special case of 2D translation invariant topological stabilizer models this classification was established rigorously at the level of lattice Hamiltonians~\cite{Haah2018a,haah2016algebraic}. 
In three spatial dimensions, recent steps have been taken towards a similar goal, with some success for topological phases that admit a topological quantum field theory (TQFT) description~\cite{lan2017classification,Lan2019,Zhu2018}. However, the burgeoning field of \textit{fracton} topological phases~\cite{chamon2005quantum,PhysRevB.81.184303,bravyi2011topological,doi:10.1080/14786435.2011.609152,haah2011local,PhysRevLett.107.150504,kim20123d,yoshida2013exotic,bravyi2013quantum,Haah2013,haah2013commuting,haah2014bifurcation,PhysRevLett.116.027202,new_TO_vijay,vijay2016fracton,PhysRevB,PhysRevB.95.245126,vijay2017isotropic,vijay2017generalization,PhysRevB.96.165106,PhysRevB.96.195139,HHB_models,hsieh_halasz_partons,PhysRevB.97.155111,PhysRevB.97.041110,PhysRevB.97.165106,Bulmash2018,PhysRevB.97.125102,PhysRevB.97.125101,PhysRevB.97.144106,finite_temp_Xcube,you2018symmetric,prem2018cage,hao_twisted,Brown2019,Dua2019,Schmitz2019,You2019,You2019b,Weinstein2019,Prem2019,Bulmash2019} present a new and challenging facet of the classification problem, as practically all the familiar tools from the study of TQFTs no longer apply. 
These phases are characterized by the restricted mobility of their topological superselection sectors. In the most extreme case of type-II fracton models~\cite{haah2011local, PhysRevLett.107.150504,bravyi2013quantum,Haah2013,haah2013commuting,
haah2014bifurcation,kim20123d,yoshida2013exotic}, there exist no string operators capable of moving any nontrivial topological excitation. More precisely, there is a logarithmic energy barrier as a function of the distance a particle is to be moved. 

Even in the relatively simple setting of translation invariant topological stabilizer models in 3D there are a plethora of known models with very little organizing structure. An exception to this is the growing understanding of \textit{foliated} fracton phases~\cite{shirley2017fracton,shirley2018FoliatedFracton,shirley2018Fractional,shirley2018Foliated,shirley2018universal,Slagle2018foliated,Shirley2019,Wang2019}, which has proved quite successful due to many tools from the study of 2D topological phases being applicable there. 

In this work we set our sights beyond any particular subclass of fracton models to consider techniques that apply to arbitrary translation invariant topological stabilizer models in 3D. Our aim is to provide diagnostics that can be applied to an unknown model to sort it into one of a few coarsely defined classes of topological phases that share similar qualitative characteristics. To achieve this we focus on the presence and deformability of string and membrane operators, providing several generalizations of the familiar S-matrix from two dimensional anyon models that capture qualitative differences between distinct classes of fracton topological order. 
Our approach is based solely on bulk properties of each model and so is not sensitive to boundary conditions, unlike previously used quantities such as the ground space degeneracy. 
Furthermore, the operator based quantities employed do not suffer from spurious contributions~\cite{Williamson2018}, unlike entanglement entropy based quantities~\cite{PhysRevB.97.125102,PhysRevB.97.125101,PhysRevB.97.144106}. 

Our main results are presented in several tables: the outcome of applying the sorting procedure to a number of fracton models, including Haah's 18 cubic codes~\cite{haah2011local,Haah2013}, is summarized in table~\ref{classes_results}. 
The characteristic behaviour of different classes of fracton order under our diagnostics, which determines the sorting, is summarized in table~\ref{procedure}. Table~\ref{table_invariants} contains detailed results of the diagnostics for each model we have considered. 
% Different classes of topological order are linked to the the corresponding zoo of examples in the appendix. 

The paper is laid out as follows: 
In section~\ref{top_classes} we describe the qualitatively distinct kinds of topological order that have been observed in 3D stabilizer models. 
%, based on mobilities of excitations which is connected to the existence as well as rigidity of string operators. 
%The classes include the class of conventional gapped TQFT stabilizer models, Type-I topological order with rigid string operators and Type-II topological order with no string operators. 
In section~\ref{tools} we describe the tools we use to identify which kind of topological order a 3D stabilizer model supports. 
%This includes commutation quantities which we have divided into two kinds- flat-rod configurations and membrane-membrane anti-commuting operators. Other tool that would be helpful for characterization of 3D stabilizer model but not necessary for identification of its type is intersection of stabilizer relations. Together we call these quantities as Operator Data associated with the 3D stabilizer model. 
In section~\ref{algo} we describe how to apply these tools to a given 3D stabilizer model to identify the class of topological order it supports. 
In section~\ref{sec:conclusions} we conclude with a discussion of the results. 

In addition to the main text there are several appendices containing examples and useful reference information: In appendix~\ref{sec:SCdeformability} we frame the deformation of pair-creation operators more formally in terms of a  sufficient condition for deformability. 
In appendix~\ref{sec:deformabilityexamples} we present examples of the deformation procedure for particle pair-creation operators. 
In appendix~\ref{sec:memmemscaling} we list numerical results that demonstrate how the number of anti-commuting membrane-membrane operator pairs scales with their size. 
In appendix~\ref{zoo} we attempt to provide an exhaustive list of 3D translation invariant topological stabilizer models that have previously appeared in the literature, along with their known properties. 

\begin{table*}[t]
    \centering
    \renewcommand{\arraystretch}{1.5}{
    \begin{tabular}{c|c|c|c|c}
         &  TQFT~[\ref{tqftzoo}] & Foliated type-I~[\ref{foliatedType1zoo}] & Fractal type-I~[\ref{fractalType1zoo}] & Type-II~[\ref{Type2zoo}]\tabularnewline
         \hline
    Mobilities & 3 & 0,1,2 & 0,1,2 & 0\tabularnewline
    of particles & & & & \tabularnewline
\hline
    Scaling of & constant & sub-extensive & sub-extensive & fluctuating \tabularnewline
    number of qubits & & & + fluctuations & with sub-extensive envelope \tabularnewline
\hline    
Examples & 3D toric code  & Checkerboard Model & Sierpinski FSL model & Cubic codes 1-4,7,8,10\tabularnewline
    & (with bosonic or & X-Cube Model & Cubic codes 0,5,6,9,11-17 & Hsieh-Halász-II model\footnote{Our results are consistent with fractal type-I or type-II.} \tabularnewline
    & fermionic charge) & Chamon's Model & & 
    \end{tabular}}
\caption{Classes of topological order. The possible mobilities of particles and scaling of the number of qubits with system size are characteristic of each class. Links are provided to a zoo of examples in the appendix.}
\label{classes_results}
\end{table*}

\section{Topological Order in Stabilizer models}
\label{top_classes}
We focus on translation invariant stabilizer Hamiltonians with topological order. These Hamiltonians are specified by a choice of commuting local Pauli stabilizer generators $h^{(i)}$, which become the interaction terms in a Hamiltonian, 
\begin{align}
    H = \sum_{i,\vec{v}} (\openone -  h^{(i)}_{\vec{v}})  \, ,
\end{align}
where $\vec{v}$ are lattice vectors. 
In the above equation, $h^{(i)}_{\vec{v}}$ indicates a local generator $h^{(i)}$ after translation by a lattice vector $\vec{v}$. 
Each local generator $h^{(i)}$ is a tensor product of Pauli matrices acting on a spatially local set of qudits, tensor producted with the identity on all other qudits. If each local generator consists of exclusively $X$ or $Z$ operators, we refer to the Hamiltonian as CSS~\cite{PhysRevA.54.1098,Steane2551}. 

In the recent literature on fractons, topologically ordered stabilizer models in three dimensions have been coarsely classified into three categories~\cite{vijay2016fracton}: TQFT order, type-I topological order and type-II topological order. We further divide type-I topological order into \emph{foliated} type-I topological order and \emph{fractal} type-I topological order, the latter of which has been referred to as type-I.5 in some works. We present a summary of the properties characterizing the different classes below. 

\subsection{TQFT order}
TQFT order, sometimes referred to as conventional topological order, is defined by the presence of particle excitations in nontrivial superselection sectors that can be moved in all directions by deformable string operators. It is characterized by a constant topological ground space degeneracy (constant w.r.t. the system size after sufficient coarse graining) and deformable logical operator segments. In two dimensions, all topological stabilizer models are essentially built by stacking 2D toric codes, and hence are TQFTs. This is due to a structure theorem~\cite{Haah2018a,haah2016algebraic} which states that under a locality-preserving unitary, any 2D translation invariant topologically ordered Pauli stabilizer model can be mapped to copies of the toric code and some ancillary qubits in a trivial product state. 
We conjecture that a similar structure theorem holds for TQFT stabilizer models in three dimensions i.e. a translation invariant TQFT stabilizer model in 3D is equivalent to copies of 3D toric codes and some trivial ancillas. There is a slight subtlety compared to the 2D case as there are two in-equivalent versions of the 3D toric code, one has a bosonic point charge and the other has a fermionic point charge~\cite{Levin_wen_fermion}.  
Hence, properties of logical operators for TQFT stabilizer models can be extrapolated from properties of logical operators in the toric code. For 2D toric code, the logical operator pairs are anti-commuting logical string operators. For 3D toric code, the logical operator pairs are composed of deformable string operators with corresponding anti-commuting logical operators as deformable membrane operators. Therefore, we expect TQFT stabilizer models to be characterized by a constant ground space degeneracy and logical string operators that are fully deformable. 

For example, the stabilizer generators of the 3D toric code with bosonic point particle are given by 
\begin{align}
\begin{array}{c}
\drawgenerator{IXI}{XXX}{}{IIX}{}{}{XII}{}
\quad
\drawgenerator{}{}{}{}{IIZ}{IZZ}{}{IZI}
\quad
\\
\drawgenerator{}{}{IIZ}{}{}{ZIZ}{}{ZII}
\quad
\drawgenerator{}{}{IZI}{}{ZII}{ZZI}{}{}
\end{array}
\, ,
\end{align}	
where $X,Y,Z$ denote the Pauli matrices and  ${IXI=I \otimes X \otimes I}$. 
Here, we have performed a coarse graining such that qubits sitting on the adjacent edges in the $\hat{x},\hat{y},\hat{z}$ directions, respectively, in the original model, are merged onto a single vertex. Representative logical operators on a torus are generated by three pairs of anti-commuting membrane and string operators $\bar{X}_{\hat{j}},\bar{Z}_{\hat{j}}$, respectively, where 
\begin{align}
\bar{X}_{\hat{x}}& = \prod_{y,z} (XII)_{0,y,z}\, ,
&
\bar{Z}_{\hat{x}}& =  \prod_{x} (ZII)_{x,0,0}\, . 
\end{align}
The subscript $\hat{x}$ in $\bar{X}_{\hat{x}}$ indicates the direction perpendicular to the plane on which the logical operator acts. 
The subscript $x,y,z$ indicates the vertex coordinate on which each three qubit operator acts. For $\bar{Z}_{\hat{x}}$, the subscript $\hat{x}$ indicates the direction of the string operator. 
One can construct membrane and string operators along $\hat{y}$ and $\hat{z}$ analogously. 
%For example, $X_{\hat{y}} = \bigotimes_{x,z} (IXI)_{\hat{x},\hat{y},\hat{z}}$ and so on. 
In section~\ref{tools}, we present an invariant that selectively detects the presence and number of such anti-commuting operator pairs that are deformable in all three dimensions. This should serve to identify the number of copies of 3D toric code in a given stabilizer model. 

\subsection{Type-I topological order}
Type-I topological order is defined by the presence of particles in nontrivial superselection sectors that have restricted mobility. It is characterized by a sub-extensive ground space degeneracy and rigid string logical operators. We call the number of dimensions in which a particle can be mobile as its mobility dimension.  The mobility dimension of the most restricted particle is sometimes specified in the nomenclature i.e. a planon, lineon, or fracton topological order contains particles with a minimum mobility dimension of 2, 1, or 0 respectively. Type-I topological orders can be further divided into two broad categories, foliated and fractal as described below. 

\subsubsection{Foliated type-I topological order}
A foliated topological stabilizer model is defined by a foliation structure~\cite{shirley2017fracton} which implies that the model can be grown by stacking with a 2D topological state and applying a local unitary. 
Due to the 2D classification result, one only needs to consider 2D toric code states for foliated topological stabilizer models. 
The most studied example of this type is the $X$-cube model~\cite{vijay2016fracton} which can be grown by stacking with a 2D toric code as shown in Fig.~\ref{foliated_TO}. More precisely, under entanglement renormalization group flow~\cite{shirley2018FoliatedFracton} copies of the 2D toric code can be extracted from the $X$-Cube model. 
This leads to a formal definition of an equivalence class of foliated topological orders~\cite{shirley2018universal} which states that two Hamiltonians are in the same equivalence class if they are connected by stacking with layers of 2D gapped Hamiltonians and gap preserving adiabatic evolution. 
We remark that these equivalence classes are much coarser than the usual definition of gapped phase~\cite{Pai2019}, and were designed to include decoupled stacks of 2D topological orders in the trivial equivalence class, thereby modding them out from the nontrivial equivalence classes. Foliated stabilizer models can support particles of different mobilities such as fractons, lineons and planons. Due to their foliation structure in terms of 2D toric codes, they always support some planons. 
For example, the $X$-Cube model supports planons in all three lattice directions, leading to a foliation structure with stacks of toric codes along the three lattice directions. 
Due to the underlying foliation structure of this class of models, the ground space degeneracy scales with the linear size of the system.  

\begin{figure}
\centering
\includegraphics[scale=0.22]{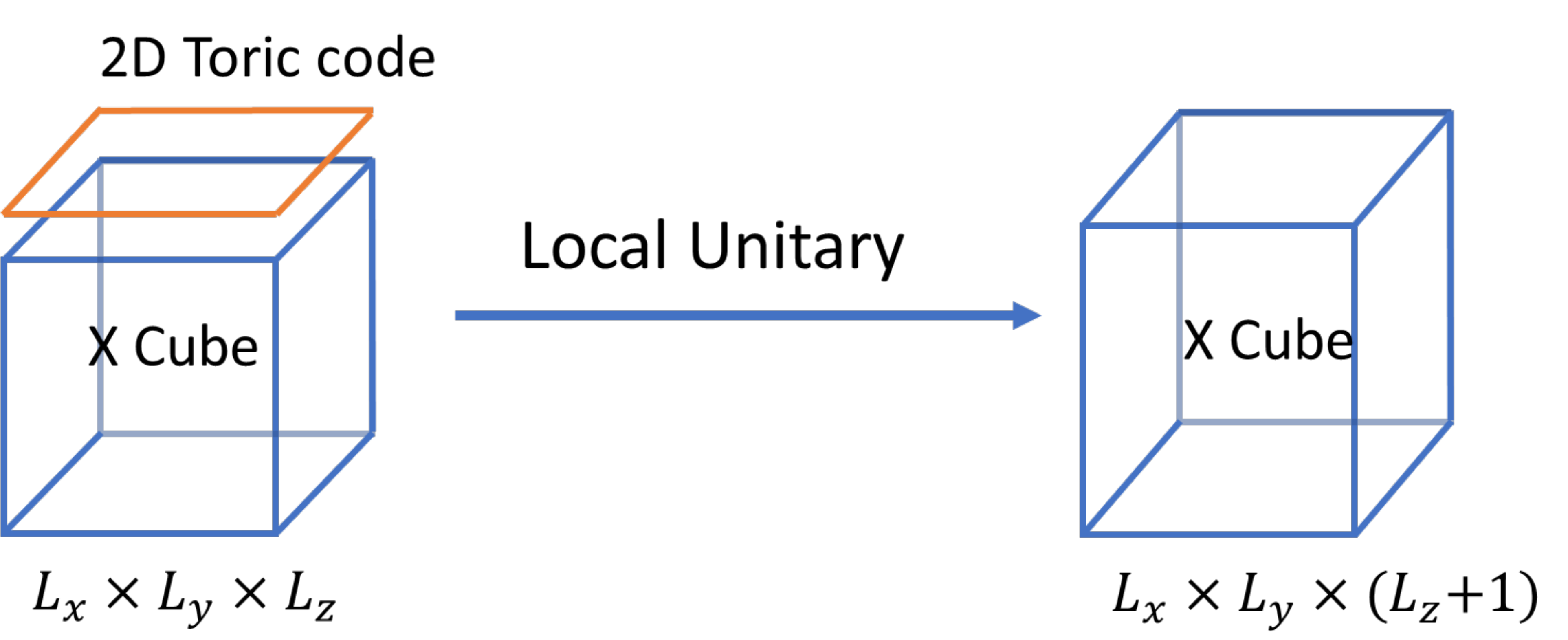}\\
\caption{The foliation structure for X-cube.}
\label{foliated_TO}
\end{figure}

Let us look at the $X$-cube model in more detail. The stabilizer generators are given by
\begin{align}
\begin{array}{c}
\drawgenerator{IXI}{}{IXX}{IIX}{XIX}{XXX}{XII}{XXI}
\quad
\drawgenerator{}{}{}{}{IIZ}{}{IZZ}{IZI}
\quad
\\
\drawgenerator{ZIZ}{}{IIZ}{}{}{}{}{ZII}
\quad
\drawgenerator{}{}{IZI}{ZZI}{ZII}{}{}{}
\end{array}
\end{align}
and translations. Here again, we have performed a coarse graining such that 3 edge qubits in the original $X$-cube model~\cite{vijay2016fracton} are merged onto a single vertex in the same manner as we did for the 3D toric code. 
We consider the model on an $L_x \times L_y \times L_z$ cuboid with periodic boundary conditions. We have logical operators $\bar{X}^{\hat{i}}_{\hat{k},\ell},\bar{Z}^{\hat{j}}_{\hat{k},\ell}$ on pairs of non-contractible loops, where $i\neq j \neq k$ run over $\{x,y,z\}$ and $\ell = 0,\dots,L_k-1$. They are defined as follows 
\begin{align}
\bar{X}^{\hat{x}}_{\hat{z},\ell} = \prod_x (XII)_{x,0,\ell} \, ,
&&
\bar{Z}^{\hat{y}}_{\hat{z},\ell} = \prod_y (ZII)_{0,y,\ell}\, ,
\end{align}
and in a similar fashion for other permutations of $x,y,z$. These string operators are not independent due to three relations 
\begin{align}
    \prod_{\ell} \bar{X}^{\hat{i}}_{\hat{k},\ell} = \prod_{\ell} \bar{X}^{\hat{k}}_{\hat{i},\ell} 
    \, ,
    &&
    \bar{Z}^{\hat{i}}_{\hat{j},0} = \bar{Z}^{\hat{i}}_{\hat{k},0}
    \, .
\end{align}
Thus, overall, there are $2(L_x+L_y+L_z)-3$ logical operator pairs. These string operators are rigid in nature as is characteristic of type-I models. The rigidity of the string operators directly corresponds to the restricted mobility of excitations. For example, particles that are pair-created by a completely rigid undeformable string operator are restricted to move in one-dimension and are therefore lineons. From now on, we will refer to the segments of string operators that create excitations in pairs as pair-creation operators. 

% \begin{align}
% \bar{Z}_{i,\hat{j}} = \prod_j Z_{i,\bar{k},j}, \hspace{2mm}
% \bar{X}_{i,\hat{j}} = \prod_k X_{\bar{j},i,k}
% \end{align}
% where $i\neq j \neq k$ run over $\{x,y,z\}$. The notation can be explained through an example- $Z^L_{x,\hat{z}}$ is a string operator i.e. direct product of $Z$'s along the $\hat{z}$ direction and we have one for each $x$ value. All of the $Z$ string operators denoted by $Z^L_{x,\hat{z}}$ i.e. for different $x$'s are chosen to occur at an arbitrary chosen $y=\bar{y}$. For $X$ string operators, notation is a bit different. $X^L_{x,\hat{z}}$ denotes the string operator conjugate to $Z^L_{x,\hat{z}}$ and is a direct product of X's along the $\hat{y}$ direction for each $x$. The set of $X^L_{x,\hat{z}}$ operators i.e. for all $x's$ are at an arbitrary chosen $z=\bar{z}$. 
% The string operators are not all independent due to three relations as
% \begin{align}
% Z^{L}_{i=\bar{i},\hat{k}}=Z^{L}_{j=\bar{j},\hat{k}}, \hspace{2mm} \prod_{i} X^{L}_{i,\hat{k}} =\prod_{j} X^{L}_{j,\hat{k}}
% \end{align}
% where $i\neq j \neq k$ run over $\{x,y,z\}$. 

\subsubsection{Fractal type-I topological order}

\begin{figure}[t]
\centering
\sidesubfloat[]{\includegraphics[scale=0.22]{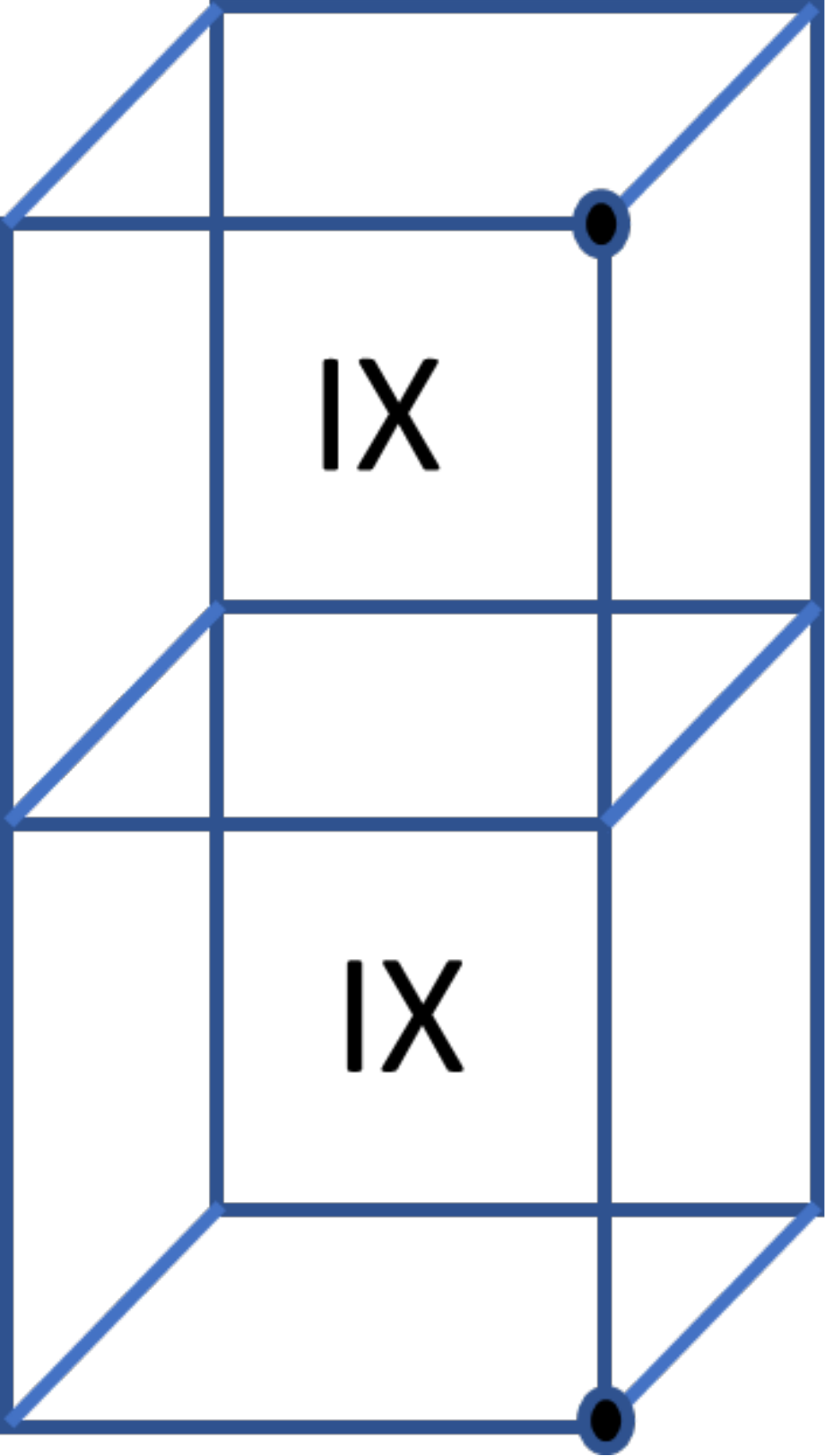}}\\
\sidesubfloat[]{\includegraphics[scale=0.22]{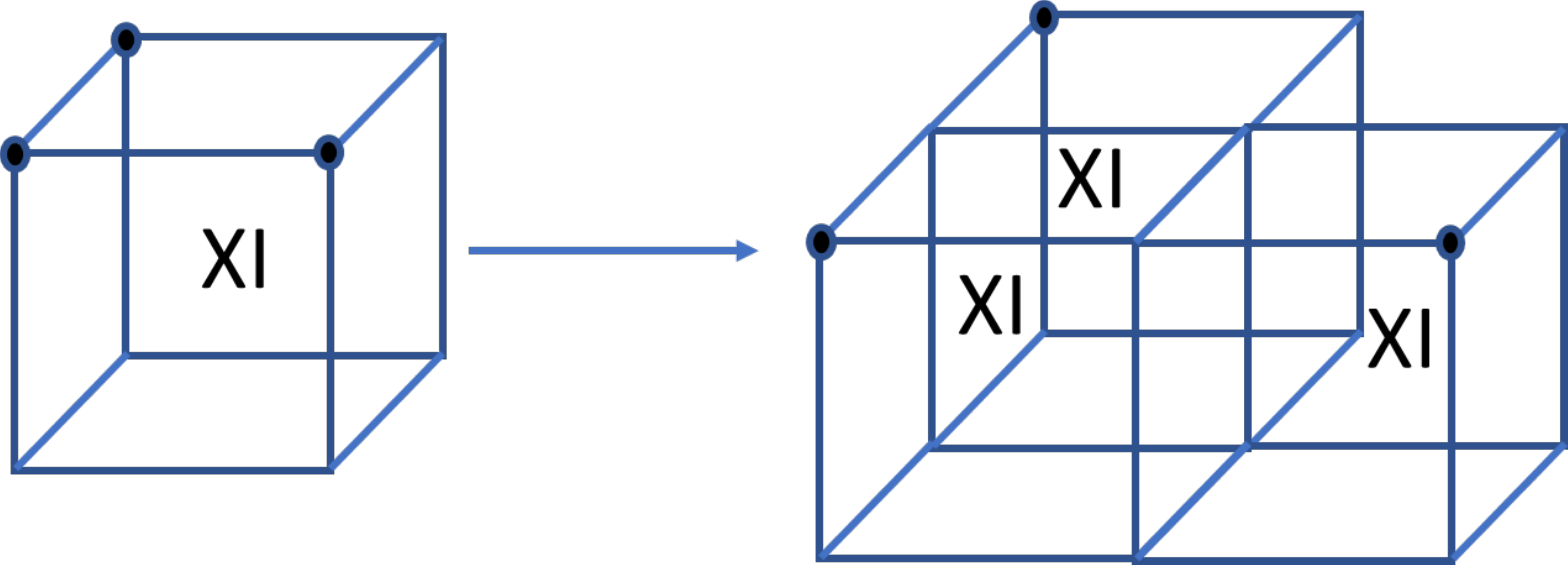}}\\

\caption{Logical operator segments in the SFSL model~\cite{doi:10.1080/14786435.2011.609152,yoshida2013exotic}. a) String operator along the $\hat{z}$ direction. A product of $I X$ operators on adjacent sites along $\hat{z}$ creates the excitation pattern shown in the dual lattice where the stabilizers live on vertices. b) A fractal operator moves three excitations apart in the $xy$-plane. Repeating this at longer length scales results in a Sierpinski triangle shaped operator. 
}
\label{Yoshida_operators}
\end{figure}

Fractal type-I topological orders are defined to be the type-I topological orders that do not admit a foliation structure. Fractal type-I topological order is characterized by the presence of operators supported on a fractal shape that move isolated topological excitations arbitrarily far apart. This subclass of orders have been referred to as type-1.5 (and even incorrectly as type-II) in some of the recent literature, although technically they fall within the original definition of type-I~\cite{vijay2016fracton}.  
Due to the presence of fractal operators, these models indeed can not support a foliation structure. 
In fact, these models need not support planons which again is inconsistent with the existence of a foliation structure. 
Also, due to the \emph{materialized} fractal symmetries of this class of models~\cite{qdouble,PhysRevB,PhysRevB.97.134426,Schmitz2018,Brown2019}, the ground space degeneracy is observed to fluctuate with the system size.  
The simplest example of a fractal type-I lineon topological order is the Sierpinski fractal spin liquid (SFSL) model due to Castelnovo, Chamon and Yoshida~\cite{doi:10.1080/14786435.2011.609152,yoshida2013exotic} specified by the following stabilizer generators
\begin{align}
\begin{array}{c}
\drawgenerator{IX}{XX}{}{XI}{}{}{}{IX}
\quad
\drawgenerator{}{}{}{ZI}{ZI}{ZZ}{}{IZ}
\end{array}
\, .
\end{align}
This model supports rigid string operators (corresponding to one-dimensional particles or lineons) in the $\hat{z}$ direction and a Sierpinski triangle fractal operator that moves topological excitations apart in 2D as shown in Fig.~\ref{Yoshida_operators}. 
Hence this model provides an example with no planons. 
We have also studied fractal type-I models that have fractal operators embedded in 3D that are not embedded in any 2D plane, similar to some type-II models, but that still support lineons and planons along with fractons and hence are not themselves type-II. 
Many of the cubic codes discovered by Haah~\cite{haah2011local,Haah2013} are of this type.

\subsection{Type-II topological order}

Type-II topological order is defined by the absence of any string operators for all topological excitations. Hence type-II orders are always fracton topological orders and furthermore no topologically nontrivial excitations, even composites, are mobile in any direction. They are characterized by discrete fractal logical operators and a sub-extensive ground space degeneracy that fluctuates with system size. 
The most well-studied type-II topological order is Haah's cubic code 1, see Ref.~\onlinecite{haah2011local}, whose stabilizer generators are given by
\begin{align}
\begin{array}{c}
\drawgenerator{XI}{II}{IX}{XI}{IX}{XX}{XI}{IX}
\quad
\drawgenerator{ZI}{ZZ}{IZ}{ZI}{IZ}{II}{ZI}{IZ}
\end{array}
\end{align}
and their translations. This model does not have any string operators and hence \emph{all} topologically non-trivial excitations are immobile. However, single excitations can still be created in isolation at the corners of a fractal operator given by a Sierpinski tetrahedron. Fig.~\ref{CC1_operators}~(a) shows the excitation pattern generated by $X I$ on the dual lattice. Fig.~\ref{CC1_operators}~(b) shows the first iteration of a fractal operator that moves these excitations apart in three dimensions. Such fractal operators are characteristic of type-II fracton orders and are tied to the absence of string operators. However, it should be noted that the presence of fractal logical operators does not imply the absence of string operators, as demonstrated by the fractal type-I models.

\begin{figure}[t]
\centering
\sidesubfloat[]{\includegraphics[scale=0.21]{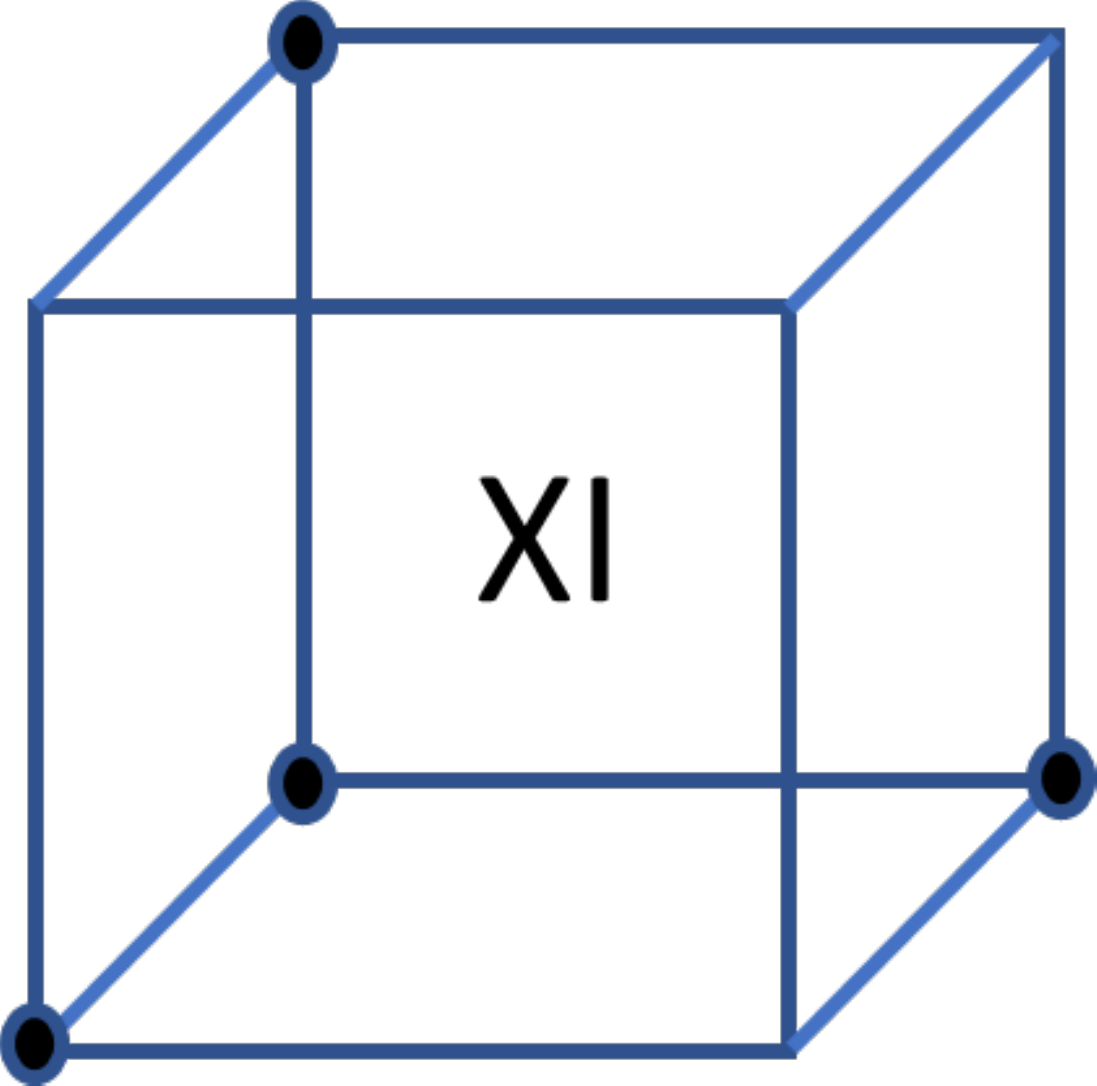}} \hspace{2mm}
\sidesubfloat[]{\includegraphics[scale=0.21]{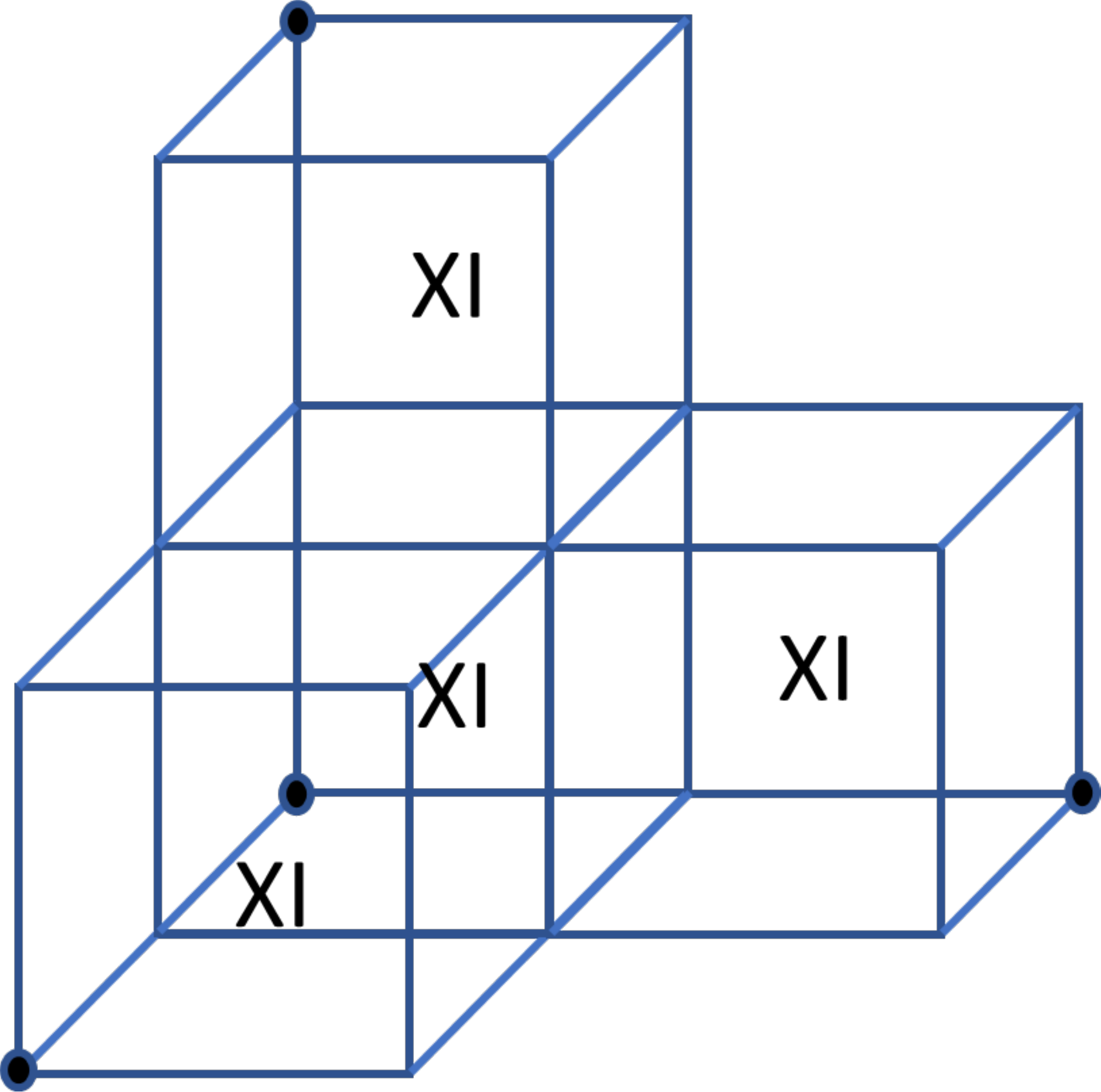}}\\

\caption{Fractal excitation patterns in the cubic code 1 model~\cite{haah2011local}. a) The excitation pattern generated by the $XI$ operator on a single site. b) A small fractal operator that moves excitations apart.
}
\label{CC1_operators}
\end{figure}

\begin{table*}[t]
    \centering
    \renewcommand{\arraystretch}{1.5}{
    \begin{tabular}{c|c|c|c|c}
         Tools & TQFT~[\ref{tqftzoo}] & Foliated type-I~[\ref{foliatedType1zoo}] & Fractal type-I~[\ref{fractalType1zoo}] & Type-II~[\ref{Type2zoo}]\\
         \hline
         Deformable pair-creation operators & all & generically not all\footnote{all pair-creation operators are deformable if all the rigid string operators are along lattice directions\label{xf}} & generically not all\textsuperscript{\ref{xf}} & all\\
         \hline
         String-membrane configurations with & all & not all (possibly none)\footnote{this is not necessarily needed to identify the type as foliated type-I or fractal type-I.\label{xx}} & not all (possible none)\textsuperscript{\ref{xx}} & none \\
          nonzero  commutation rank & & & & \\
         \hline
         Scaling of membrane- & 0 & generically linear & linear and/or & constant and/or\\
         membrane operator pairs & & & fluctuating corrections & fluctuating corrections \\
    \end{tabular}}
    \caption{Tools to identify the type of a topological order. 
    The second row indicates whether or not all pair-creation operators, as shown in Fig.~\ref{excitation_config}, are deformable into flat-rods. 
    The third row displays the number of different string-membrane (or flat-rod) configurations, as shown in Fig.~\ref{flat_rods1}, that lead to a commutation matrix of non-zero rank. 
    The fourth row contains the scaling of the membrane-membrane commutation matrix rank with size.}
    \label{procedure}
\end{table*}

\section{Tools to identify classes of topological order}
\label{tools}
This section describes diagnostic tools that are sufficient to identify the class of topological order that a given 3D topological stabilizer model falls into. 
Our primary motivation is to generalize the string-string commutation matrix invariant~\cite{haah,ribbons,Dua2019} from two to three dimensions. 
A straightforward  string-membrane commutation matrix generalization of the invariant only works for TQFT stabilizer models. 
We want to accommodate the possibility that some string operators may be fully rigid or deformable only in two dimensions. 
Hence, we consider anti-commuting logical operator pairs that can be supported on regions of certain shapes cut out from the 3D stabilizer model in order to gain some information about their rigidity. 
We refer to the shapes of regions that support anti-commuting logical operator pairs as configurations. 
We focus on two types of configurations in particular, string-membrane configurations and membrane-membrane configurations. 

Our design of the string-membrane configurations is informed by complimentary information about deformability of string operators that can be derived from a stabilizer Hamiltonian. The deformation process indicates whether a string operator can be deformed into an equivalent operator supported on the union of flat-rods along the lattice directions. For this reason we refer to the string-membrane configurations as flat-rod configurations. The deformability and generalized commutation matrix data suffice to sort a given topological stabilizer model into one of the four classes of topological order introduced in the previous section. For type-I models we further employ the intersection of generalized Gauss's laws~\cite{Brown2019} to find the minimal mobility dimension, $d$. This determines whether a type-I model has a planon, lineon, or fracton topological order.

%Note there is an obvious S-matrix for TQFTs, but this won't be LU invariant for rigid string-mem comm. We introduce a variant of this that should be LU invariant by flexing the string to ensure 3D mobility.

\subsection{String-membrane configurations}
We consider several configurations involving pairs of anti-commuting operators, one supported on a string segment that creates excitations at the endpoints, and the other on a membrane patch that creates excitations along the loop-like boundary. The distinct string-membrane configurations, shown in Fig.~\ref{flat_rods1}, are inspired by the deformation of string operators to contiguous flat-rods along lattice directions in the proof of the no-strings condition for Haah's cubic code~\cite{haah2011local} and Kim's qupit code~\cite{kim20123d}. 
Hence only a small number of flat-rod configurations, aligned with the cubic lattice directions, are considered. 
There are 48 configurations with flat-rods along all three lattice directions, an example is shown in Fig.~\ref{flat_rods1}~(a). Similarly there are 12 configurations, 4 in each lattice plane, with flat-rods along two lattice directions, for example see Fig.~\ref{flat_rods1}~(b), and 3 configurations with a flat-rod along a single lattice direction, see Fig.~\ref{flat_rods1}~(c). We don't need to check all of these configurations though since we assume that if there exists a string operator that is deformable in 3D for example, then it can be supported on any of the 48 3D flat-rod configurations. Similarly, for each lattice plane, we need to consider only one flat-rod configuration. 

% The string-membrane configurations involve a pair of anti-commuting operators, one supported on a string segment and the other on a membrane segment. 
% We consider several different string-membrane configurations are shown in Fig.~\ref{flat_rods1}. As mentioned previously, the regions that support the string operator is made of tube-like regions along the lattice directions in the cubic lattice. 
% We call these tube-like regions as flat-rods as shown in Fig.~\ref{flat_rods1} and the structure formed by contiguously joining them as flat-rod structure. 
% This configuration is inspired from the proof of no-string condition~\cite{haah2011local} obeyed by the Cubic code or equivalenty from the construction of qupit code~\cite{kim20123d} without string operators because of the deformability of logical operators to contiguously joined flat-rods.

Due to the existence of particles with limited mobility, the geometry of the flat-rod configuration has a profound impact on the properties of the associated string-membrane commutation matrix. The string operators detected by each flat-rod configuration depend upon the deformability of the strings. 
The three dimensional flat-rod configurations detect string operators for particles with three dimensional mobility, the two dimensional flat-rod configurations are additionally sensitive to string operators for planon particles in the same plane, and the one dimensional flat-rod configurations are further sensitive to lineon particles with mobility along the flat-rod. 

Owing to the rigid nature of lineon and planon string operators, the commutation matrix ranks associated with the one and two dimensional flat-rod configurations depend sensitively upon the width of the flat-rod. This demonstrates that even the naive generalization of the TQFT commutation matrix to capture anti-commuting deformable string and membrane operator pairs must be taken with some care. Employing a three dimensional flat-rod configuration ensures that one does in fact obtain a local unitary invariant. We speculate that this invariant counts the number of copies of 3D toric code, with either bosonic or fermionic point particle, that can be disentangled from a stabilizer Hamiltonian.

% We present the results obtained from string-membrane configurations in table~\ref{table_invariants}. 
We tabulate the deformability and commutation matrix data we have obtained for cubic codes 0 to 17 from Ref.~\cite{haah2011local} and several other fracton models in table~\ref{table_invariants}. 
We have used the Hamiltonians for cubic codes 1 to 10 given in Ref.~\cite{Haah2013}; these are equivalent to the Hamiltonians given in the earlier Ref.~\cite{haah2011local} under symplectic transformations~\cite{private_haah}. In the table, full deformability of a model is indicated by a $\checkmark$, while an obstruction to deformability is indicated by $\times$. In one case a ? indicates that we were not able to resolve whether the model was deformable or not. The numerical value of the commutation rank is displayed for the 3D flat-rod configurations, while for the 2D and 1D flat-rod configurations we only record whether the result is zero or nonzero as the numerical value depends upon the width of the flat-rods.

\begin{figure}
\centering
\sidesubfloat[]{\includegraphics[scale=0.3]{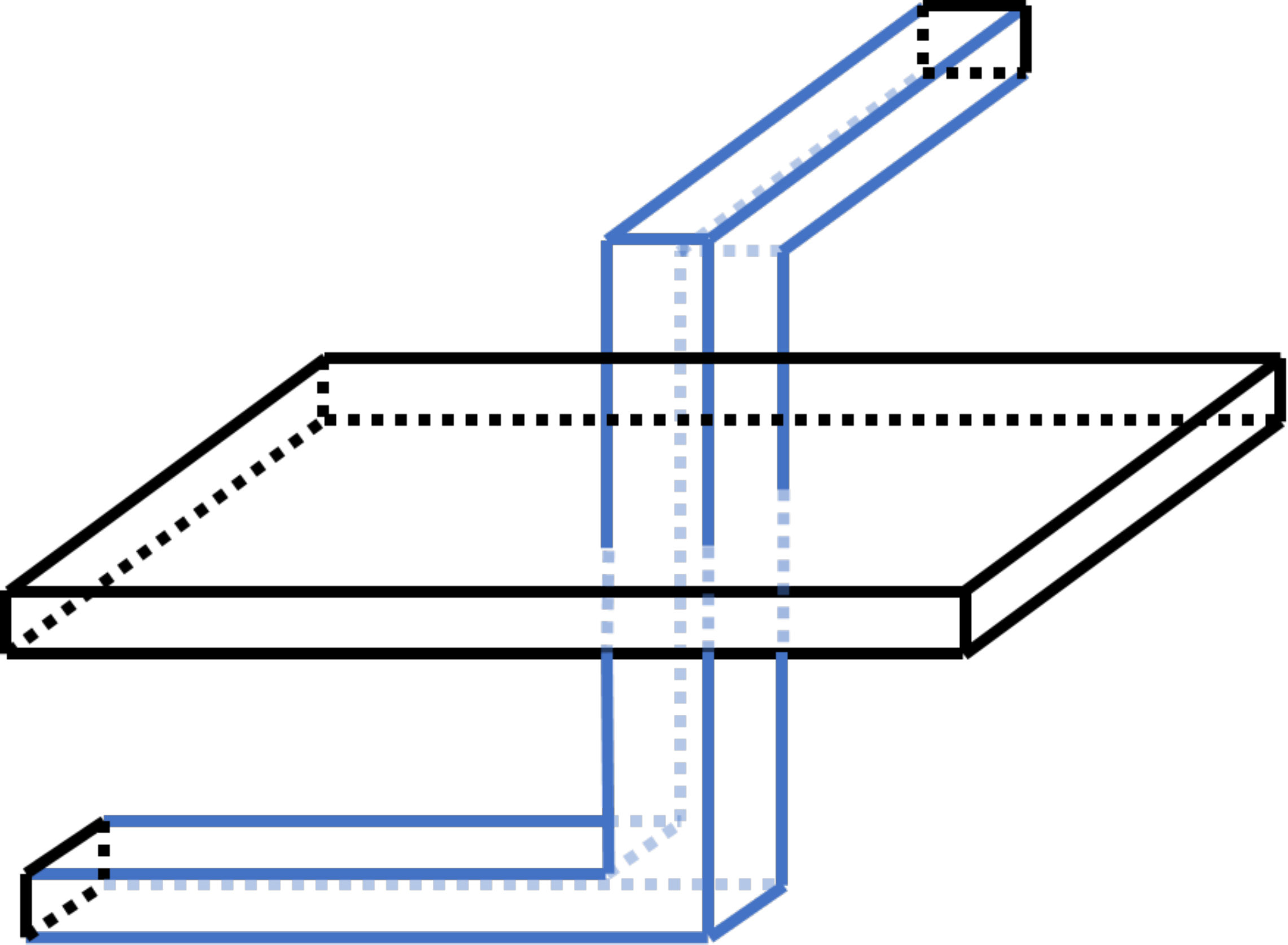}}\\
\sidesubfloat[]{\includegraphics[scale=0.3]{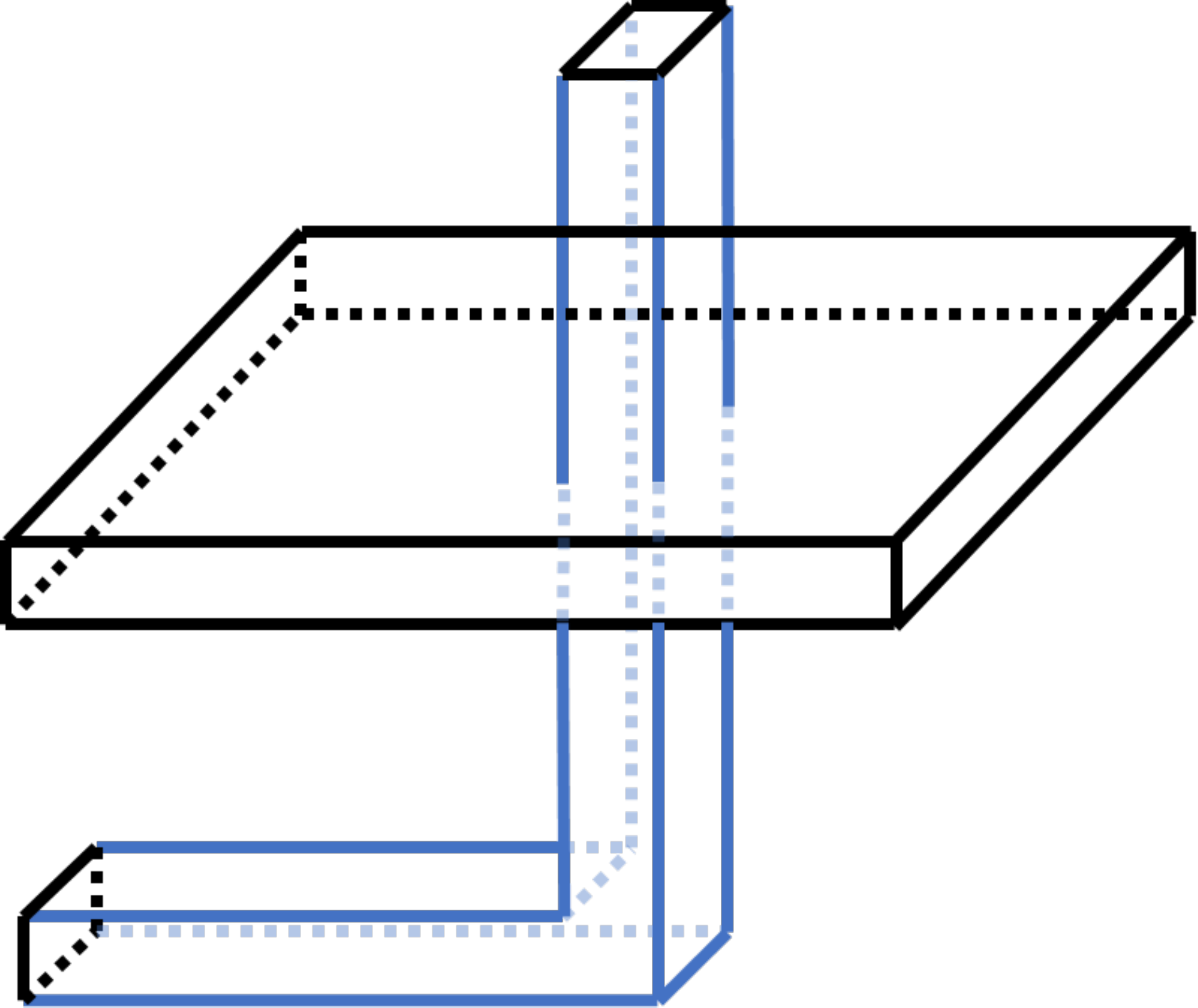}}\\
\sidesubfloat[]{\includegraphics[scale=0.3]{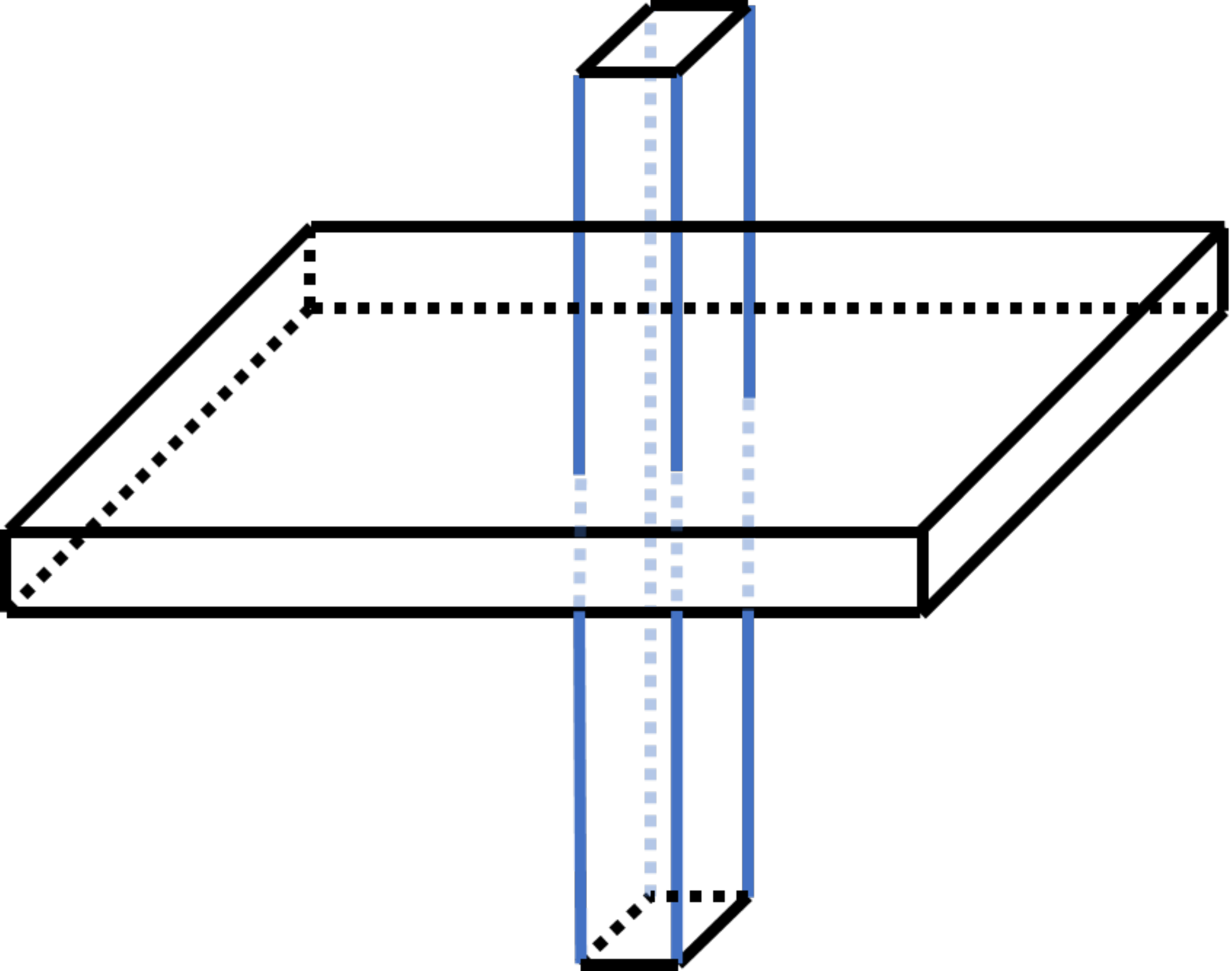}}\\
\caption{Flat-rod configurations. Open boundary conditions where excitations may be created are marked with black edges. Closed boundary conditions where no excitations are created are marked with blue edges}
\label{flat_rods1}
\end{figure}

\subsubsection{Flat-rod commutation matrix}
%The commutation matrix of these string operators, defined below, determines the $S$ matrix invariant of the relevant topological phase. It furthermore does not suffer from spurious contributions, due to subsystem symmetries, which may afflict attempts to extract the number of copies of toric code from a topological entanglement entropy calculation. We search for Pauli string operators in these sub regions that may only create excitations at their end points. These are included in the kernel of the excitation map with closed boundary conditions along the length of the strips, and open boundary conditions at the ends. Where closed boundary conditions correspond to including stabilizer generators that cross the boundary, and open boundary conditions do not. 

As discussed, the flat-rod configurations we have considered are of the string-membrane kind where the string operator is supported on a contiguous configuration of flat-rods while the membrane operator is supported on a planar patch. The commutation matrix $C$ associated with such a configuration is defined via
\begin{align}
    C_{i,j} = 
    \begin{cases} 
   0 & \text{if }  [\strop^{r}_i,\strop^{m}_j]=0 \, , \\
   1       & \text{if }  \{\strop^{r}_i,\strop^{m}_j\} = 0 \, .
%   ,[\strop^{r}_i,\strop^{m}_j] \neq 0  
  \end{cases}
\end{align}
where $\strop^{r}$ are a basis of Pauli operators supported on the flat-rods that commute with the Hamiltonian except at the endpoints and $\strop^{m}$ are a similar basis of Pauli operators supported on the membrane. The $\Z_2$-rank of $C$ is the number of independent pairs of anti-commuting operators supported on these regions. For these operators to anti-commute, the excitations they create must be topologically nontrivial. Similarly, a commutation matrix can be defined for any configuration supporting anti-commuting operator pairs that create well-separated excitations, in particular for the membrane-membrane configuration discussed below. 

\subsubsection{Deformability of pair-creation operators to flat-rod configurations}
To interpret the results of the flat-rod configurations we supplement them with a test of deformability for pair-creation operators in each model. 
Specifically, we consider operators that create a pair of well-separated excitations, each of which may have some spatial extent that is small compared to the separation. We do not assume the excitations are topologically nontrivial, although they may be. 
The pairs of excitations fall into several distinct cases depending upon the vector between their positions. For excitations that share a lattice plane there are two distinct cases per plane, while for excitations that do not, there are four distinct cases, see Fig.~\ref{excitation_config}. 

In the deformation process we first examine the constraints placed on the form of an arbitrary pair-creation operator by the fact that it must commute with the Hamiltonian terms away from the excitations. Next we proceed to check whether the support of such an  operator can be deformed to a flat-rod configuration through multiplication with local stabilizer generators. 
See appendix~\ref{sec:SCdeformability} for a detailed discussion of deformability for pair-creation operators. 

The flat-rod commutation matrices count the number of distinct nontrivial pair-creation operators that can be supported on each region. Hence if a model has full deformability i.e. all pair-creation operators are deformable to flat-rods, keeping the positions of excitations fixed, then the flat-rod commutation matrices count the number of distinct topological charges that can be pair created at their endpoints by operators without any shape restriction. 
In particular, a model with full deformability and with all flat-rod commutation matrix ranks equal to zero supports no string operators and hence is type-II. Similarly, for a model with full deformability, if the 3D, 2D and 1D flat-rod commutation matrix ranks are equal and non-zero, the model must be either TQFT or a mix of TQFT and type-II. The converse is also true, hence full deformability and equality of commutation matrix ranks associated with 3D, 2D and 1D flat-rod configurations is an 'if and only if' condition for a model to be TQFT or a mix of TQFT and type-II. In practice the deformability and the values of flat-rod configurations can only be checked up to a finite length scale. In some cases this leads to an inconclusive result, consistent with both fractal type-I and type-II, where a model is not fully deformable but no string operators are found. 

The 3D and 2D flat-rod configurations are designed to count the number of 3D particles and planons supported on a given region, respectively. 
Hence we expect the value of all 3D flat-rod configurations to be equal, and similarly we expect 2D flat rod configurations of the same width and in the same plane to be equal. Consequently, if some pair-creation operator is deformable to one flat-rod configuration but not an equivalent one we expect the relevant commutation matrix rank to be zero. In the extreme case that a pair-creation operator is not deformable to any flat-rod configuration, it implies the existence of a nontrivial rigid string operator. In particular, if certain pair-creation operators can be cleaned onto rigid lines that do not run along an axis we can check those directions separately for the presence of nontrivial rigid string operators.

\begin{figure}
\vspace{6mm}
\centering
\includegraphics[scale=0.28]{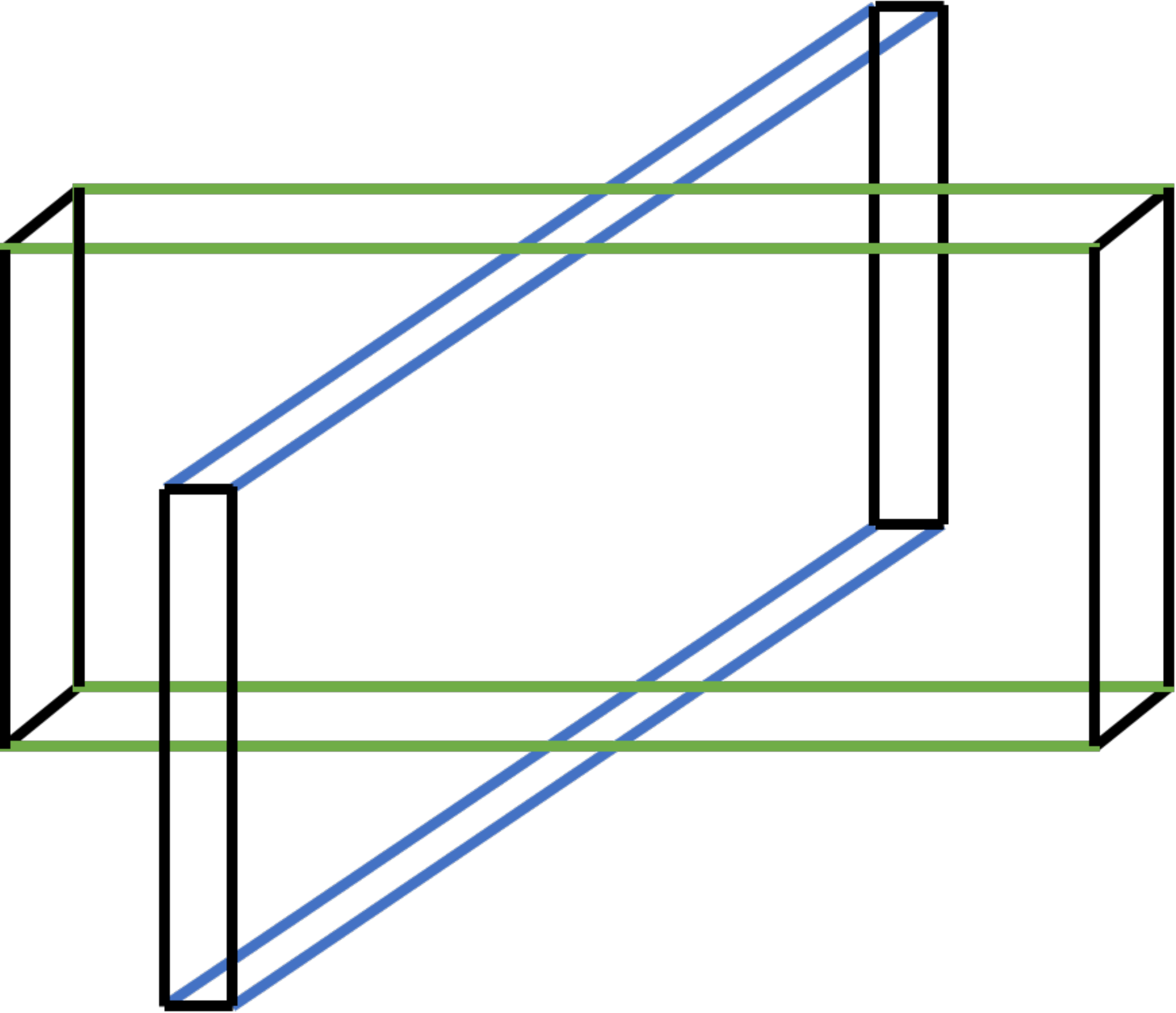}
\caption{The membrane-membrane configuration. Open boundary conditions where excitations may be created are indicated by black edges. Closed boundary conditions where no excitations are created are indicated in green and blue for the two membranes respectively}
\label{planon_planon}
\end{figure}

\setlength{\tabcolsep}{9pt}
\begin{table*}[!p]
\renewcommand{\arraystretch}{1.4}{
\begin{tabular}{c|ccccccc}
 
Model & Deformability & $n_{\text{rods}}^{3D}$ & $n_{\text{rods}}^{2D}\left(zx,xy,yz\right)$ &
$n_{\text{rods}}^{1D}\left(x,y,z\right)$ & $n_{m}\left(zx,xy,yz\right)$ & $d$ & Type\tabularnewline
\hline

CC0 & $\times$ & 0 & $\left(0,0,0\right)$ & $\left(0,0,0\right)$ & $\left(c,c,c\right)$ & 0 & fractal type-I \tabularnewline

CC1 & $\checkmark$ & 0 & $\left(0,0,0\right)$ & $\left(0,0,0\right)$ & $\left(c,c,c\right)$ & 0 & type-II\tabularnewline
 
CC2 & $\checkmark$ & 0 & $\left(0,0,0\right)$ & $\left(0,0,0\right)$ & $\left(c,c,c\right)$ & 0 & type-II\tabularnewline
 
CC3 & $\checkmark$ & 0 & $\left(0,0,0\right)$ & $\left(0,0,0\right)$ & $\left(c,c,c\right)$ & 0 & type-II\tabularnewline
 
CC4 & $\checkmark$ & 0 & $\left(0,0,0\right)$ & $\left(0,0,0\right)$ & $\left(c,c,c\right)$ & 0 & type-II\tabularnewline
 
CC5 & $\times$ & 0 & $\left(0,0,0\right)$ & $\left(0,0,0\right)$ & $\left(c,0,c\right)$ & 1 & fractal type-I\tabularnewline
 
CC6 & $\times$ & 0 & $\left(0,0,0\right)$ & $\left(0,0,0\right)$ & $\left(c,0,c\right)$ & 1 & fractal type-I\tabularnewline
 
CC7 & $\checkmark$ & 0 & $\left(0,0,0\right)$ & $\left(0,0,0\right)$ & $\left(c,c,c\right)$ & 0 &   type-II\tabularnewline
 
CC8 & $\checkmark$ & 0 & $\left(0,0,0\right)$ & $\left(0,0,0\right)$ & $\left(c,c,c\right)$ & 0 & type-II\tabularnewline
 
CC9 & $\times$ & 0 & $\left(0,0,0\right)$ & $\left(0,0,0\right)$ & $\left(c,c,c\right)$ & 1 & fractal type-I\tabularnewline
 
CC10 & $\checkmark$ & 0 & $\left(0,0,0\right)$ & $\left(0,0,0\right)$ & $\left(c,c,c\right)$ & 0 & type-II\tabularnewline
 
CC11 & $\checkmark$ & 0 & $\left(0,0,n_1\right)$ & $\left(n_2,n_1,n_1\right)$ & $\left(c,c,\ell\right)$ & 0 & fractal type-I\tabularnewline
 
CC12 & $\checkmark$ & 0 & $\left(n_1,0,0\right)$ & $\left(n_1,0,n_1\right)$ & $\left(\ell,c,c\right)$ & 0 & fractal type-I\tabularnewline
 
CC13 & $\checkmark$ & 0 & $\left(n_1,0,0\right)$ & $\left(n_1,0,n_1\right)$ & $\left(\ell,c,0\right)$ & 0 & fractal type-I\tabularnewline

CC14 & $\checkmark$ & 0 & $\left(n_1,0,0\right)$ & $\left(n_1,0,n_1\right)$ & $\left(\ell,c,c\right)$ & 0 & fractal type-I\tabularnewline
 
CC15 & $\checkmark$ & 0 & $\left(0,0,n_1\right)$ & $\left(0,n_1,n_1\right)$ & $\left(c,c,\ell\right)$ & 0 & fractal type-I\tabularnewline
 
CC16 & $\times$ & 0 & $\left(0,0,0\right)$ & $\left(0,0,0\right)$ & $\left(c,c,c\right)$ & 0 & fractal type-I\tabularnewline
 
CC17 & $\checkmark$ & 0 & $\left(0,n_1,0\right)$ & $\left(n_1,n_1,n_2\right)$ & $\left(c,\ell,c\right)$ & 0 & fractal type-I\tabularnewline
 
3DTC & $\checkmark$ & 1 & $\left(1,1,1\right)$ & $\left(1,1,1\right)$ & $\left(0,0,0\right)$ & 3 & TQFT\tabularnewline

2DTC$_{xy}$ & $\checkmark$ & 0 & $\left(0,n_1,0\right)$ & $\left(n_1,n_1,0\right)$ & $\left(0,\ell,0\right)$ & 2 & foliated type-I \tabularnewline
 
2DTC$_{yz}$ & $\checkmark$ & 0 & $\left(0,0,n_1\right)$ & $\left(0,n_1,n_1\right)$ & $\left(\ell,0,0\right)$ & 2 & foliated type-I \tabularnewline
 
2DTC$_{xz}$ & $\checkmark$ & 0 & $\left(n_1,0,0\right)$ & $\left(n_1,0,n_1\right)$ & $\left(0,0,\ell\right)$ & 2 & foliated type-I \tabularnewline
 
XC & $\checkmark$ & 0 & $\left(n_1,n_1,n_1\right)$ & $\left(n_2,n_2,n_2\right)$ & $\left(\ell,\ell,\ell\right)$ & 0 & foliated type-I\tabularnewline
 
CB & $\checkmark$ & 0 & $\left(n_1,n_1,n_1\right)$ & $\left(n_2,n_2,n_2\right)$ & $\left(\ell,\ell,\ell\right)$ & 0 & foliated type-I\tabularnewline

Chm & $\checkmark$ & $0$ & $\left(n_1,n_1,n_1\right)$ & $\left(n_2,n_2,n_2\right)$ & $\left(\ell,\ell,\ell\right)$ & 0 & foliated type-I \tabularnewline

SFSL & $\checkmark$ & 0 & $\left(0,0,0\right)$ & $\left(0,0,n_1\right)$ & $\left(0,0,c\right)$ & 1 & fractal type-I  \tabularnewline

HH-I & $\checkmark$ & 0 & $\left(n_1,n_1,n_1\right)$ & $\left(n_2,n_2,n_2\right)$ & $\left(\ell,\ell,\ell\right)$ & 0 & foliated type-I\tabularnewline

HH-II & ? & 0 & $\left(0,0,0\right)$ & $\left(0,0,0\right)$ & $\left(c,c,c\right)$ & 0 & inconclusive\footnote{Our results are consistent with fractal type-I or type-II.}\tabularnewline 
\end{tabular}}
\caption{Sorting data for a range of 3D topological stabilizer models.  We consider the cubic codes 0-17~\cite{Haah2013,haah2011local} labeled CC0-17, the 3D toric code, labeled 3DTC, stacks of 2D toric code parallel to the $ij$ plane, labeled 2DTC$_{ij}$, the X-cube model and checkerboard model~\cite{vijay2016fracton}, labeled XC and CB respectively, Chamon's model~\cite{chamon2005quantum,bravyi2011topological}, labeled Chm, the Sierpinski fractal spin liquid~\cite{yoshida2013exotic} labeled SFSL, and finally the so-called type-I and II spin models in \R{hsieh_halasz_partons}, labeled HH-I and HH-II  respectively. 
The first column indicates whether pair creation operators in a model are fully deformable $\checkmark$, not $\times$, or inconclusive $?$. 
The second column shows the commutation matrix rank of the 3D flat-rod configuration $n_{\text{rods}}^{3D}$. 
The third column shows the commutation matrix rank of the 2D flat-rod configurations, $n_{\text{rods}}^{2D}$, in the $zx$, $xy$ and $yz$ lattice planes, where the membrane is perpendicular to the $x$, $y$ and $z$ direction, respectively. 
The fourth column shows the commutation matrix rank of the 1D flat-rod configurations, $n_{\text{rods}}^{1D}$, with a string along the lattice direction $x$, $y$ and $z$. 
The placeholders $n_1$ and $n_2$ represent some non-zero numbers which depend on the width of the flat rods and may be different for each model. 
The fifth column indicates the scaling of the membrane-membrane commutation matrix rank, $n_{m}$, with the size of the membranes, where $c$ stands for constant and $\ell$ stands for linear scaling, both up to fluctuating corrections. The notation $zx,xy,yz$ indicate the directions of the membranes, where $zx$ refers to a membrane in the $xy$ plane intersecting a membrane in the $yz$ plane and similarly for the others. The sixth column shows the mobility dimension, $d$, of an excited $X$ stabilizer generator. 
The final column displays the class of topological order resulting from the sorting procedure. 
All data except for CC0 (which has width five string operators) and HH-II was taken with excitation configurations within a box of dimensions $L_{x}=L_{y}=L_{z}=L=20$, or with a size $L$ for which the values become stable with respect to increasing the system size further. The flat-rods were taken to be as wide as three stabilizer generators and the membranes as wide as two. The type of topological order for HH-II is listed as inconclusive since the model is not fully deformable using constraints up to third order, but we also did not find a nontrivial string operator. It is possible that the model is fully deformable using higher order constraints and hence is type-II, or that it contains a nontrivial string operator that we did not find and hence is fractal type-I.} 
\label{table_invariants}
\end{table*}

\setlength{\tabcolsep}{6pt}

\subsection{Membrane-membrane configurations}

We compliment the string-membrane configurations by considering membrane-membrane configurations that support rectangular shaped operators which create excitations along only two edges, see Fig.~\ref{planon_planon}, such that the commutation matrix is well defined. 
These commutation quantities are inspired by the possibility of anti-commuting operators in fracton models that cannot appear in TQFTs. In fact the rank of the membrane-membrane commutation matrix is zero for any TQFT, and hence this quantity detects whether a model has fracton topological order. Furthermore, the scaling of the rank with the size of the membranes can be used to distinguish whether a type-I model is an instance of foliated or fractal type-I order.  

The shape of the membrane-membrane configuration determines the type of operators it is sensitive to. 
We only consider configurations where the membranes are aligned with lattice planes. Without loss of generality suppose one membrane has dimensions $L_x \times L_z$ in the $xz$ plane, and the other has dimensions $L_y \times L_z$ in the $yz$ plane. We suppose both are of width $w$, small compared to their lengths $L_x,L_y$. 
For $L_z$ small compared to $L_x,L_y$, the configuration counts anti-commuting string operator pairs along the $\hat{x}$ and $\hat{y}$ axes. This will only detect the presence of lineon and/or planon pairs along these axes. 

We find that configurations where $L_x=L_y=\alpha L_z$, for a constant aspect ratio $\alpha$ of order 1 as $L_z$ is scaled, provide more useful information that allows us to distinguish foliated and fractal type-I orders. 
For foliated type-I orders we expect the scaling to be linear due to the presence of planons. This is because the plane of mobility for an arbitrary planon will generically intersect the plane of a membrane operator along a line, and hence the string operator for that planon can be supported on some translation of the membrane, provided $\alpha$ is sufficiently large. For fractal type-I orders, and type-II orders, the presence of fractal operators that create topological charges at their corners lead to contributions to the membrane-membrane commutation rank that fluctuate with the precise size of the membranes. Depending on the aspect ratio $\alpha$ the contribution of the fractal operators will either be of order constant, or may scale within a linear envelope. 
It is also possible for fractal type-I orders to support planon excitations, in which case we expect linear scaling plus fluctuating corrections. 
% This is unlike the type-II fracton order which always shows only order of constant scaling up to small fluctuations. 

We present results obtained from the membrane-membrane configurations for a large range of topological stabilizer models in table~\ref{table_invariants}. 
The scaling of the membrane-membrane commutation rank with $L_z$ for ${\alpha=1}$ is reported, with $c$ indicating a nonzero result of order constant that may be fluctuating, and $l$ indicating a linear scaling with a possible fluctuating constant correction.

\begin{figure}
\vspace{6mm}
\centering
\includegraphics[scale=0.3]{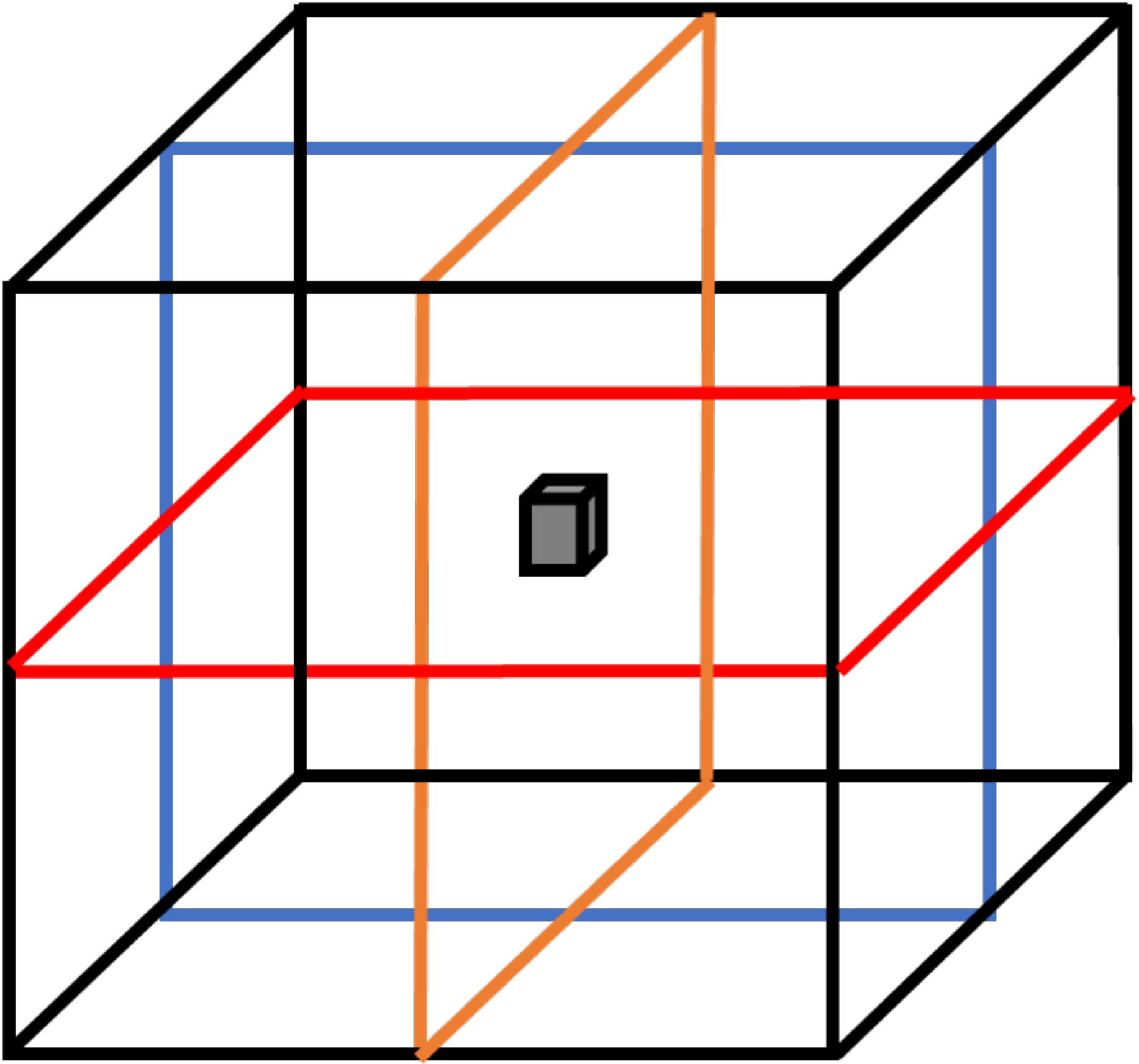}
% \caption{Stabilizer relations in the $X$-Cube model involving a stabilizer as shown at the center of the cube are along the orange, red and blue planes. The intersection of these relations include only the location of the stabilizer and hence the stabilizer excitation is a fracton}
\caption{Generalized Gauss's laws for the fracton sector of the $X$-cube model within a finite region. The set of Gauss's laws containing a given stabilizer share only one common intersection point. }
\label{relationsXC}
\end{figure}

\subsection{Intersection of generalized Gauss's laws}

For type-I models we utilize a further tool to determine the minimal mobility of their topological excitations and hence whether a given model is a fracton, lineon or planon type-I order. We achieve this by investigating the structure of conserved charges due to generalized Gauss's laws in a given stabilizer model. 

On a system with periodic boundary conditions each \emph{relation}, i.e. a set of nontrivial stabilizer generators that multiply to the identity, gives rise to a materialized symmetry of parity conservation for a subset of topological charges~\cite{qdouble,PhysRevB.97.134426,Schmitz2018,Brown2019}. For simplicity we assume there are no local relations, as these are not relevant for the fracton models we have considered.  This leaves open the possibility of global relations, generalizing the familiar global $\Z_2 \times \Z_2$ charge conservation in 2D toric code.
These global relations are sensitive to boundary conditions, dramatically so for fracton models. To avoid such complications we turn our attention to an arbitrary large cube region $\mathcal{R}$ within the bulk. Each relation that overlaps $\mathcal{R}$ leads to a generalized Gauss's law obtained by restricting the product of generators in that relation to include only those that act nontrivially on $\mathcal{R}$. This produces a generalized Wilson operator supported just outside $\mathcal{R}$ that measures the $\Z_2$ charge for that relation within $\mathcal{R}$. 

The intersection of the set of Gauss's laws that contain a stabilizer generator determine the mobility of the particle obtained by exciting it. 
Here the intersection of a set of Gauss's laws refers to the set of stabilizer generators that appear in all the aforementioned Gauss's laws. 
We define the intersection of a stabilizer generator to be the intersection of the set of Gauss's laws that contain it. Exciting a pair of generators that are contained in each other's intersections has neutral charge under all Gauss's laws, and hence can be implemented by a local pair-creation operator~\cite{Haah2013}. 
If a generator's intersection contains only itself, an  excitation of that generator must be a fracton. 
More generally, a generator $g_1$ that supports a fracton could have a nonempty intersection consisting of distinct generators $g_2$ such that $g_1$ is not in the intersection of $g_2$.  
Alternatively, for an excitation of a single generator $g_1$ to be a lineon, the intersection of $g_1$ must contain a distinct generator $g_2$, and the intersection of $g_2$ must contain $g_1$. Similarly for an excitation of a single generator $g_1$ to be a planon, the intersection of $g_1$ must contain two distinct generators $g_2,g_3$ in linearly independent directions that both contain $g_1$ in their intersections. 
A similar condition holds again for a fully mobile excitation, this time involving three distinct generators in linearly independent directions. In the above discussion we have focused on stabilizer Hamiltonians that are given by translates of a single type of $X$ and $Z$ generator for simplicity, a generalization to the case of multiple generator types is straightforward. 

This concept is illustrated for the $X$-cube model in Fig.~\ref{relationsXC} where the Gauss's laws involving a particular $X$ stabilizer lie on the three lattice planes shown and their intersection contains only the single stabilizer. Hence, a single $X$ stabilizer excitation in the $X$-cube model is a fracton~\cite{vijay2017generalization}. 
In our sorting procedure we use the intersection of Gauss's laws to determine the minimum mobility of excitations in each stabilizer model. This is particularly relevant when considering a type-I model, as it may be a fracton, lineon, or planon model, whereas TQFTs always have fully mobile particles and type-II models only contain fractons, by definition. In any case the mobility constraints derived from the Gauss's laws serve as a consistency check. We remark that since the Gauss's law test is applied to a particular choice of generators, it is not particularly useful for determining the type of topological order in a model. For example, to verify that a model is type-II in this manner would require checking every possible charge cluster, which is not feasible. 

The deformability and commutation matrix tests do not suffer from the same difficulty. 
In table~\ref{table_invariants}, we report the dimension of the least mobile particle  in a variety of models, found by using the intersection of Gauss's laws.

\section{Sorting a topological stabilizer model of unknown type}
\label{algo}
As discussed in the previous section, the deformability of pair-creation operators supplemented by string-membrane and membrane-membrane commutation matrix ranks for a given stabilizer model reveal the existence of topological particles along with information about their mobilities. 
% For example, if the three dimensional or two dimensional flat-rod configuration supports a non-trivial logical operator i.e. rank of associated commutation matrix is non-zero, that implies the existence of a 3D particle or a 2D particle respectively. The 2D particle can be proven to be restricted to just two dimensions (and not three) only if the three dimensional flat-rod configuration does not support any nontrivial logical operators because otherwise, the same 3D particle would show up in the 2D configuration. 
We outline a procedure that utilizes these test to sort a given topological stabilizer model of unknown type below. The tools used are summarized in table~\ref{procedure}. 

%\vspace{3mm}
%\noindent
%\box{
%\begin{minipage}{\dimexpr \linewidth-2 %\fboxsep \fboxrule\relax}
%\vspace{2mm}
%  \vspace{5mm}
%\end{minipage}}
%\vspace{1mm}

\subsection{Sorting procedure}

\noindent
\begin{itemize}
\item First check whether all operators that create a pair of excitation patches in any of the configurations shown in Fig.~\ref{excitation_config} are deformable to flat-rods. This proceeds by first attempting to deform an arbitrary pair-creation operator to lie within the minimal box containing the excitations, shown in Fig.~\ref{excitation_config}. Next it is checked whether any such pair-creation operator can be further deformed into a flat-rod configuration. The full deformation procedure is explained in detail in appendix~\ref{section:cleaning_strategy}.   
    \begin{itemize}
    \item If all the pair-creation operators are deformable to flat-rods, check the ranks of the commutation matrices associated with the flat-rod configurations: 
        \begin{itemize}
        \item If the flat-rod commutation matrix ranks are not all equal, then the model is {type-I} with rigid string operators along some lattice directions. 
        \item If the ranks associated with all flat-rod configurations are zero, then the model is type-II. 
        \item If the ranks associated with all flat-rod configurations are equal and non-zero, then the model is either TQFT or a stack of TQFT and type-II up to a local unitary circuit. 
        \item In order to distinguish between TQFT and a combination of TQFT and type-II, check the membrane-membrane commutation matrix, if the result is always 0, the model is a TQFT.   More generally, we conjecture that the value of $n_{\text{rods}}^{3D}$ indicates the number of copies of 3D toric code contained in the model, which also applies to type-I models. 
        \end{itemize}

    \item
    If the pair-creation operators are not all deformable to flat-rods, we conjecture the model must be type-I: 
    \begin{itemize}
        \item Next check the scaling of the membrane-membrane commutation matrix, if it is linear with no correction the model is foliated type-I. Otherwise, in the case there are fluctuating correction terms, the model is fractal type-I.
        \item Finally, check the generalized Gauss's laws within a cube to find the mobility of single stabilizer excitations to determine whether the model has fracton, lineon, or planon topological order, corresponding to the mobility of the most restricted particle. 
    \end{itemize}
    \end{itemize}
\end{itemize}

By following the above procedure, one can sort 3D topological stabilizer models into the following classes: TQFT, type-II, fracton/lineon/planon foliated or fractal type-I. Combinations of the aforementioned types of topological order can also be generated by stacking models of different types. 

The data presented in table~\ref{table_invariants} indicates that cubic codes~\cite{haah2011local,Haah2013} CC1-4, CC7, CC8 and CC10 are type-II while the remaining cubic codes are fractal type-I. 
We remark that the fractal type-I model CC0 and type-II models, such as CC1, differ only in the deformability of their pair-creation operators to flat-rods as the only rigid string operators in CC0 are along non-lattice directions. CC11-15 and CC17 have rigid string operators only along lattice directions and hence the pair-creation operators are deformable to flat-rods along lattice directions. 
Another example of a fractal type-I model is given by the Sierpinski fractal spin liquid (SFSL) in which all elementary excitations are lineons that can move along the $\hat{z}$ direction.  
The X-cube model (XC)~\cite{vijay2016fracton}, checkerboard model (CB)~\cite{vijay2016fracton}, Chamon's model (Chm)~\cite{chamon2005quantum, bravyi2011topological} and the type-I model from Ref.~\cite{hsieh_halasz_partons} (HH-I) are all examples of foliated type-I models. Not all pair-creation operators for the so-called type-II model (HH-II) from Ref.~\cite{hsieh_halasz_partons} could be deformed to flat-rods using the commutation constraints up to third order. 
However, we also did not find a non-trivial string operator. Since it is possible that the pair-creation operators for this model are fully deformable to flat-rods using higher order constraints, we have left the status of this model's type as inconclusive.

\begin{figure*}[t]
\vspace{6mm}
\centering
\sidesubfloat[]{\includegraphics[scale=0.33]{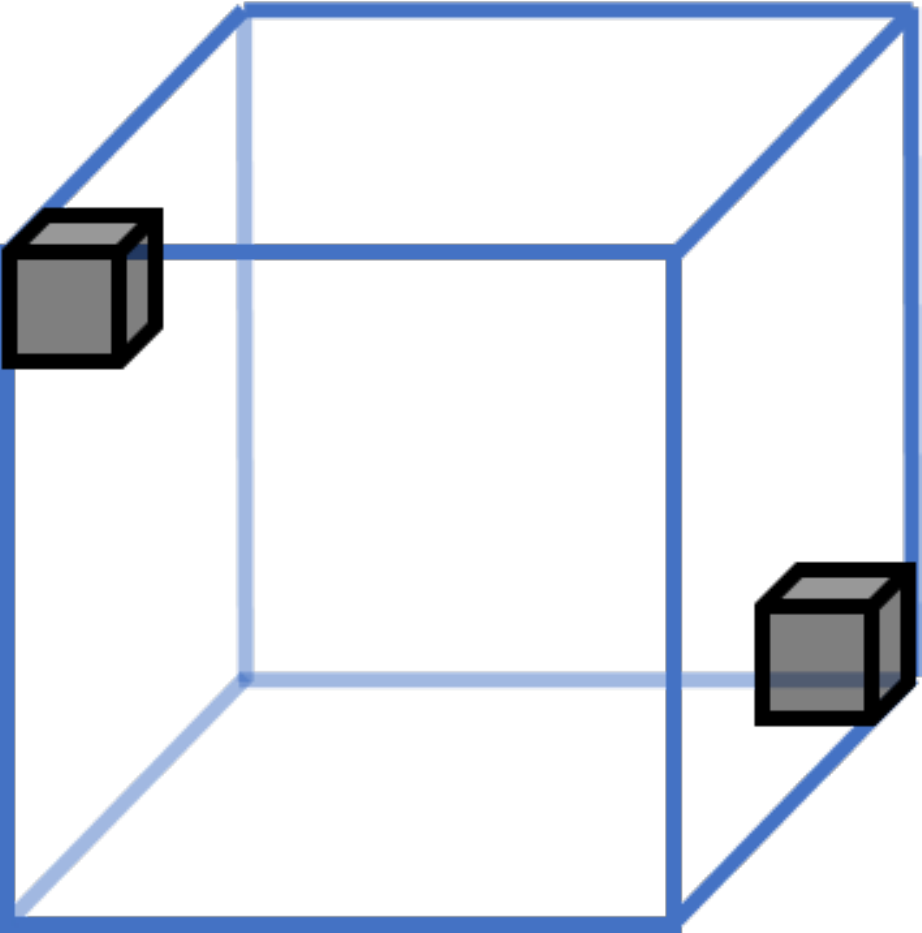}} \hspace{2mm}
\sidesubfloat[]{\includegraphics[scale=0.33]{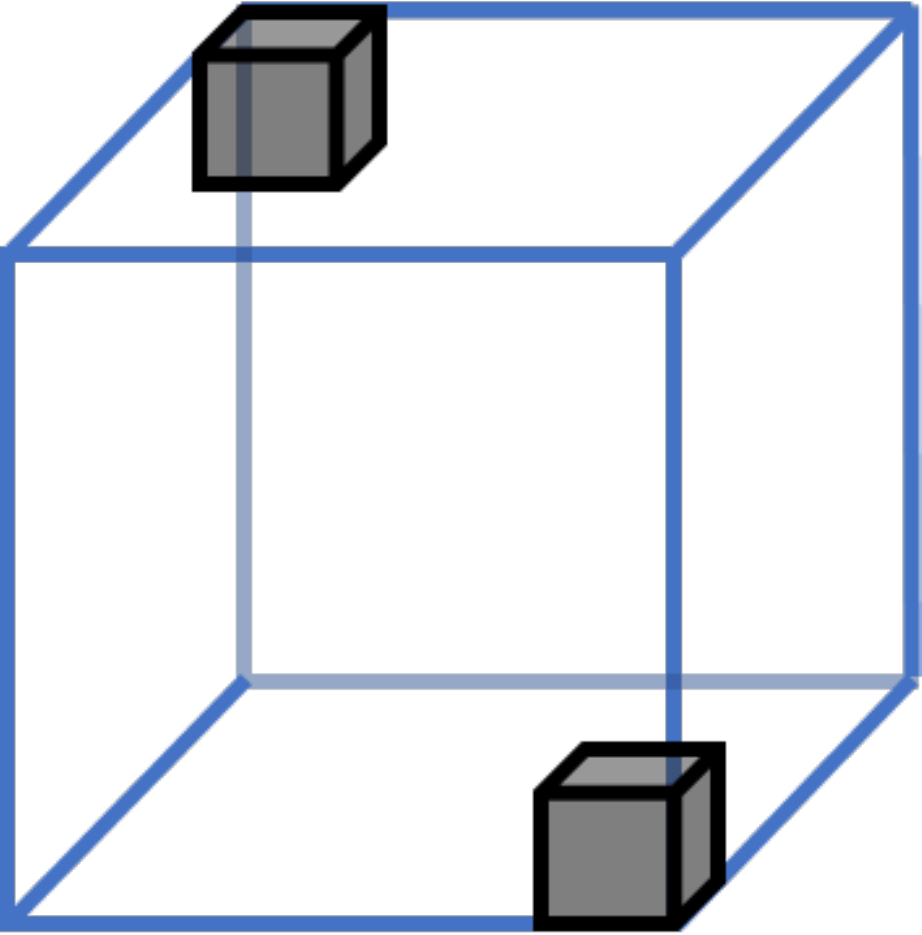}} \hspace{2mm}
\sidesubfloat[]{\includegraphics[scale=0.33]{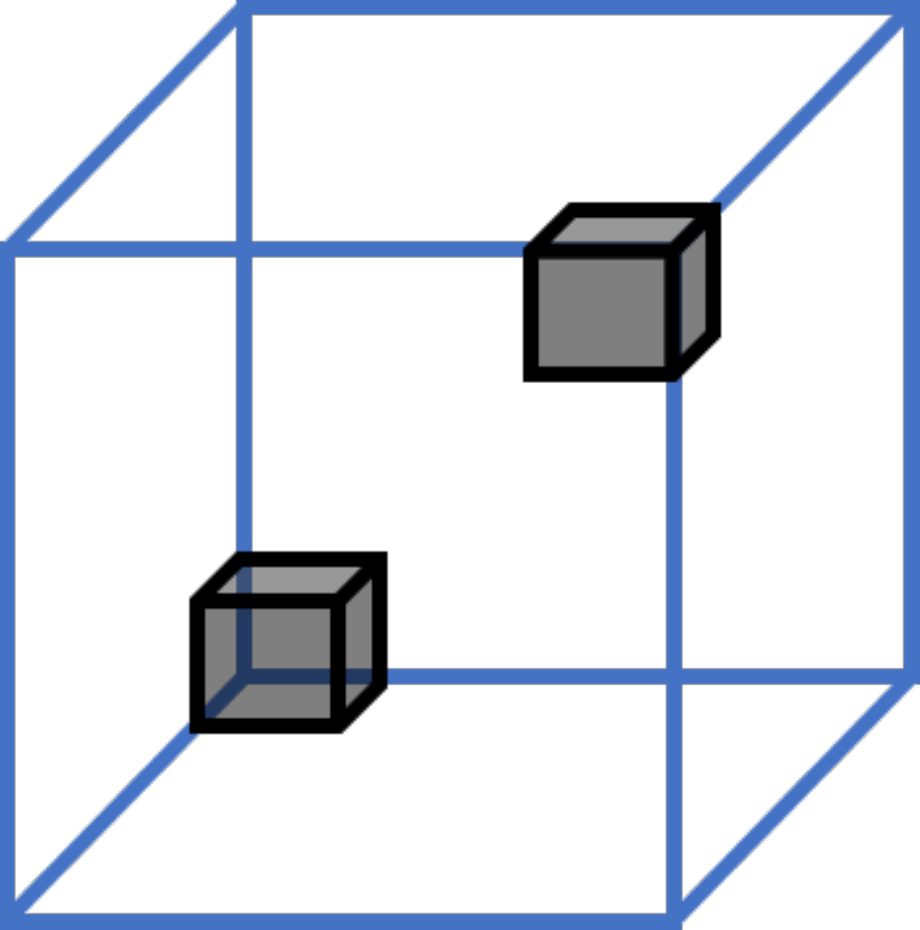}} \hspace{2mm}
\sidesubfloat[]{\includegraphics[scale=0.33]{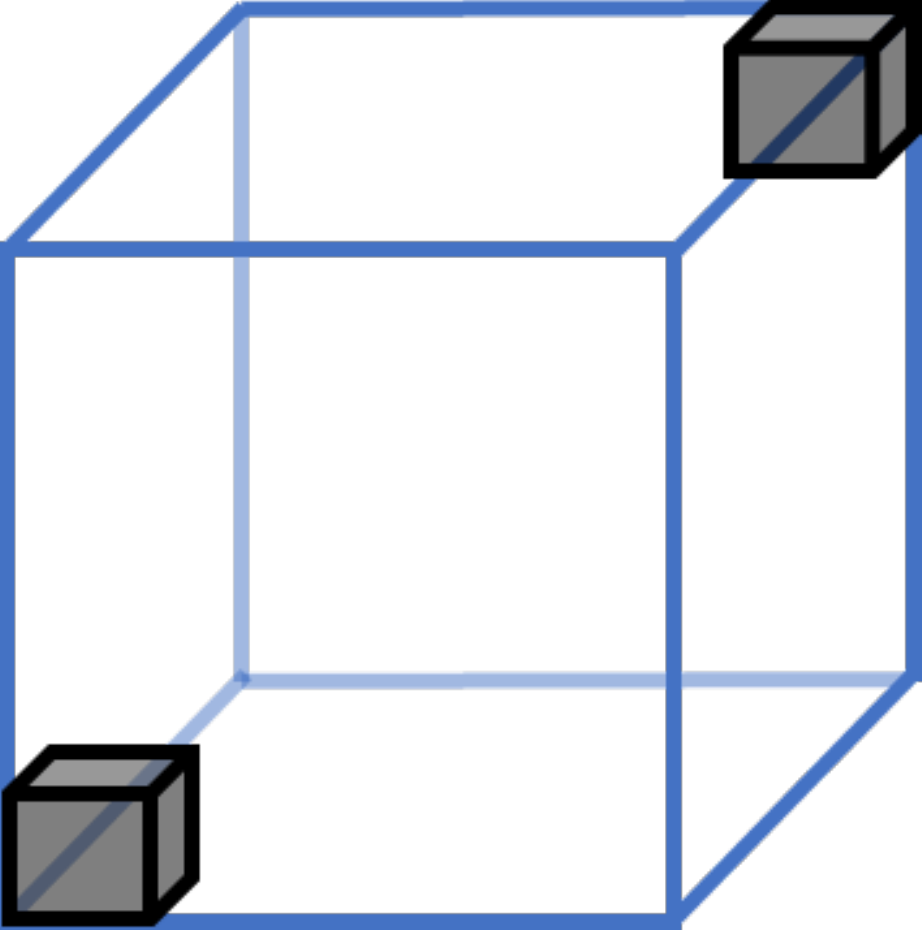}}\\ \vspace{3mm}
\sidesubfloat[]{\includegraphics[scale=0.33]{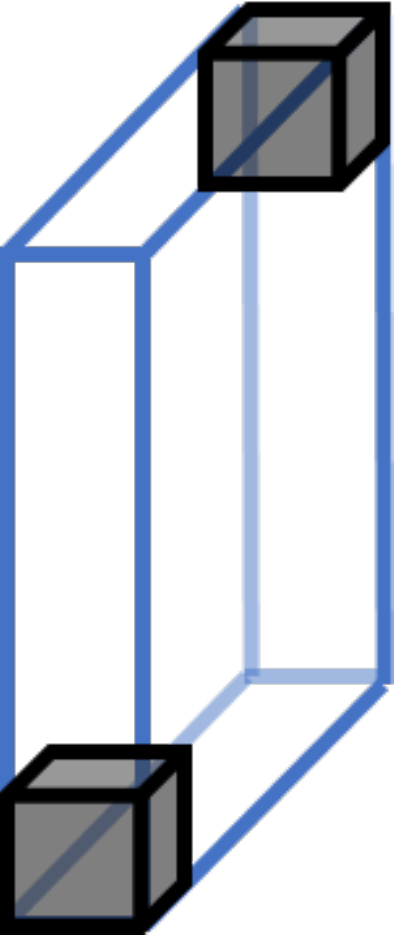}} \hspace{2mm}
\sidesubfloat[]{\includegraphics[scale=0.33]{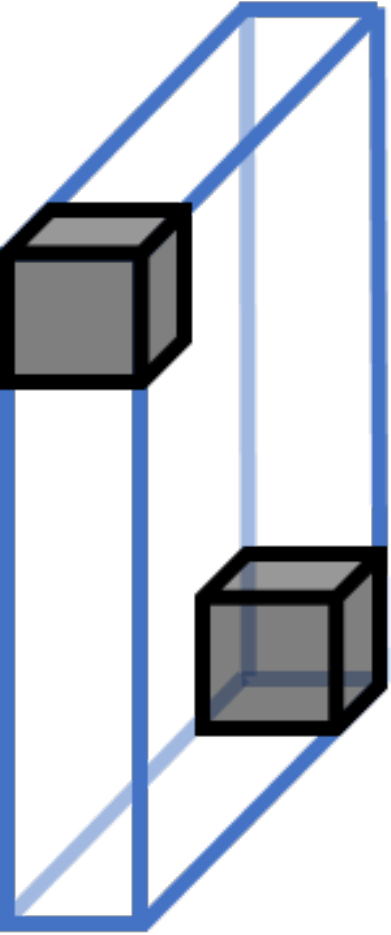}} \hspace{2mm}
\sidesubfloat[]{\includegraphics[scale=0.33]{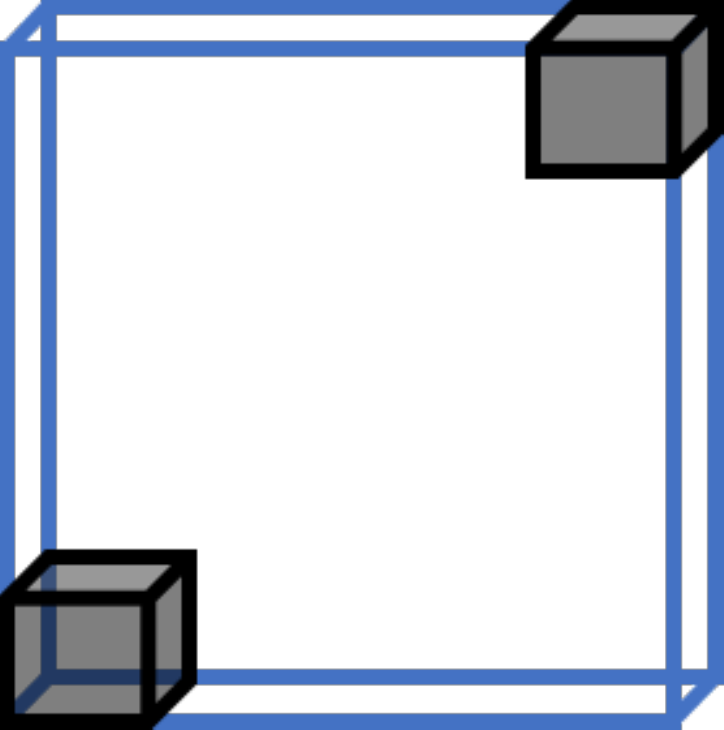}} \hspace{2mm}
\sidesubfloat[]{\includegraphics[scale=0.33]{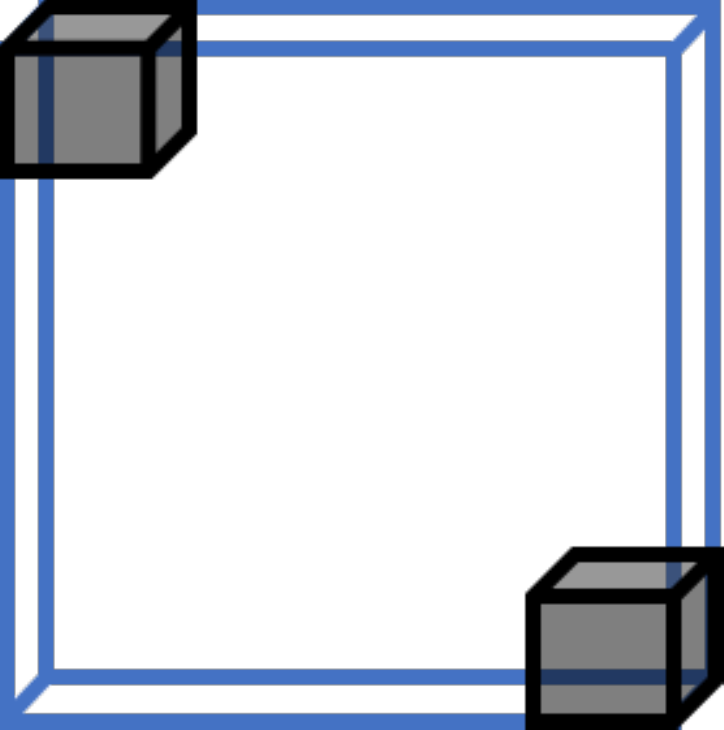}} \hspace{2mm}
\sidesubfloat[]{\includegraphics[scale=0.33]{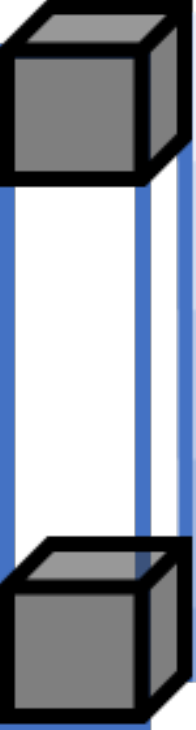}} \\  \vspace{3mm}
\sidesubfloat[]{\includegraphics[scale=0.33]{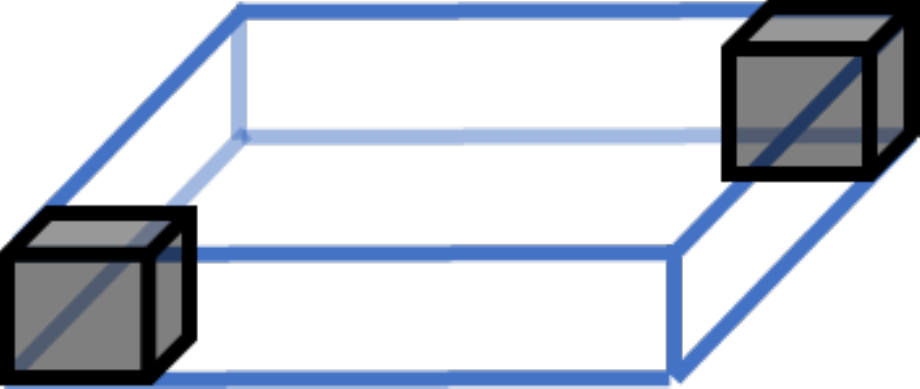}} \hspace{2mm}
\sidesubfloat[]{\includegraphics[scale=0.33]{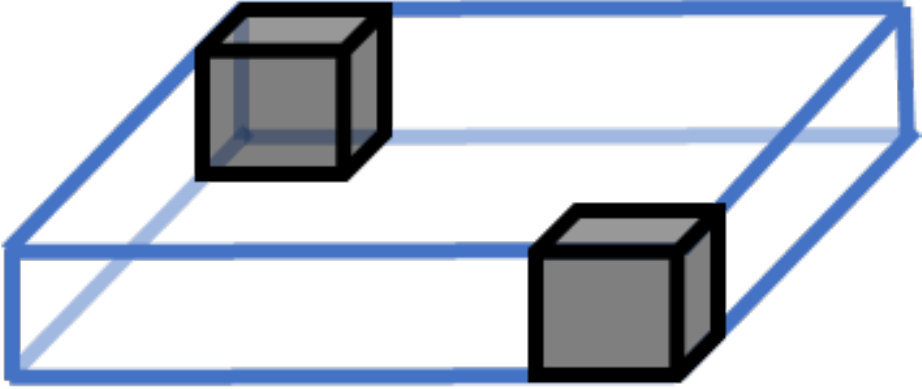}} \hspace{2mm}
\sidesubfloat[]{\includegraphics[scale=0.33]{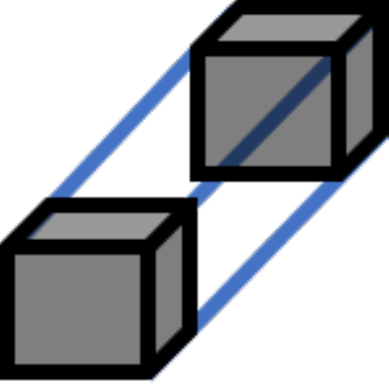}} \hspace{2mm}
\sidesubfloat[]{\includegraphics[scale=0.33]{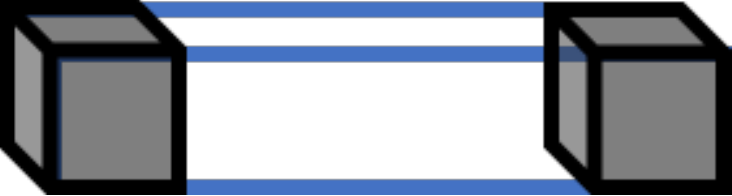}}  
\caption{Different configurations of excitation pairs (grey cubes) and the minimal boxes containing them.}
\label{excitation_config}
\end{figure*}
\setlength{\tabcolsep}{0.6em}

\subsection{Assumptions and limitations}

When applying the sorting procedure there are some practical limitations worth noting. 

\begin{itemize}

\item For the deformability of pair-creation operators, we assume that the excitation patches are sufficiently far apart, i.e. their separation is much larger than their extent. If this condition is not respected one can find non-zero commutation matrix ranks associated with certain flat-rod configurations, even for models with no corresponding string operators.  

\item When checking the deformability of a pair-creation operator, we utilize only a subset of all the possible constraints that arise from commutation of the operator with stabilizer generators away from the excitations. For models that are not fully deformable in table~\ref{table_invariants} we have checked up to third order constraints at least, where zeroth order refers to constraints on vertices due to commutation relations, first order refers to constraints on pairs of vertices and so on. 
This is supplemented by a direct search for nontrivial string operators, informed by the deformability results, and the demonstration of particles with subdimensional mobility via generalized Gauss's laws. 
When such rigid string operators are found they demonstrate conclusively that a model is not fully deformable, which is indicated by a $\times$ in table~\ref{table_invariants}. 
When a model is not fully deformable to third order but no rigid string operator is found the result is inconclusive as it is consistent with type-I or type-II, which is indicated by a $?$ in table~\ref{table_invariants}. 
% Hence, we use non-deformability of pair-creation operators only as a supporting remark for type-1 models, instead of conclusive evidence. 
% In the type-1 models we study, we complement our non-deformability remark by demonstrating existence of rigid string operators and correspondingly particles of sub-dimensional mobilities. 
While any finite order constraint can be checked in principle due to the general result in Appendix~\ref{sec:SCdeformability}, this process becomes increasingly complex as one considers higher order constraints. 
A nontrivial application of higher order constraints is discussed in Appendix~\ref{sec:deformabilityexamples}. 

\item The commutation matrix ranks for string-membrane and membrane-membrane configurations can only be calculated for rods and membranes of a tractable finite size. Hence, it is possible in principle that nontrivial operators of very large widths not reached in our numerics have not been accounted for. It would be interesting if one could upper bound the string operator width that needs to be considered for translation invariant topological stabilizer models to close this loophole, but we do not have such a bound presently. 
\item We have only considered membrane-membrane configurations where the two membranes are squares i.e. the aspect ratio is 1 and $L_x=L_y=L_z$. For this choice of aspect ratio, the membrane-membrane configuration does not detect the presence of planons with mobility in planes that intersect the membranes in lattice planes, along a diagonal, such as the plane spanned by $\hat{x}+\hat{y},\,\hat{y}-\hat{z}$. 
% capture a string operator that is exactly along the diagonal at a 45 degree angle, line on which a planon's plane intersects with the lattice plane. 
For models with such planons, we can perform a modular transformation to map the diagonal planons into alignment with a lattice plane. In the list of models with planons we have considered, only cubic code 16 has planons along such a diagonal plane. These planons indeed show up in the membrane-membrane configuration after a modular transformation is performed. In fact, after performing a modular transformation to map the planons into a lattice plane, we found the resulting model is equivalent to cubic code 15 up to the relabeling of axes.  
\end{itemize}

\section{Discussion and Conclusions} 
\label{sec:conclusions}

In this work we have provided a set of analytic and numerical tools that can be applied to translation invariant topological stabilizer models to sort them into one of several qualitatively distinct classes of topological quantum order: TQFT, foliated type-I, fractal type-I or {type-II}. 
Our methods provide a recipe to sort an unknown topological stabilizer model into one of these classes, as explained in section~\ref{algo}. 

% as follows
% \begin{itemize}
%     \item First check whether the operators that create a pair of excitation patches in any configuration are deformable to flat-rods. If they are, the model is a stack of TQFT and type-II:
%     \begin{itemize} 
%         \item Next check the membrane-membrane commutation matrix, if the result is always 0, the model is a TQFT. 
%         \item Check the flat-rod membrane commutation matrices, if they are all 0 the model is type-II. More generally the value of $n_{\text{rods}}^{3D}$ indicates the number of copies of 3D toric code contained in the model. 
%     \end{itemize}
%     If the operators are not deformable to flat-rods, we conjecture the model must be type-I: 
%     \begin{itemize}
%         \item Next check the scaling of the membrane-membrane commutation matrix, if it is linear with no correction the model is foliated type-I. Otherwise, in the case there are fluctuating correction terms, the model is is fractal type-I.
%         \item Finally, check the generalized Gauss's laws within a cube to find the mobility of single stabilizer excitations to determine whether the model is fracton, lineon, or planon, corresponding to the mobility of the most restricted particle. 
%     \end{itemize}
% \end{itemize} 

Our coarse sorting of topological stabilizer models into classes with qualitatively similar physical properties constitutes a step towards a full classification. 
Such a classification would likely require a more careful formulation of generalized S-matrix quantities for fracton phases that give rise to true local unitary invariants. We speculate that such S-matrices would have to be customized to fit the fusion rules and particle mobilities on a model by model basis. Research along these lines has been initiated in Ref.~\onlinecite{Pai2019}, which is complimentary to the 'one size fits all' approach taken here. 

Over the coarse of this example centric study we have formulated several conjectures about the structure of general translation invariant topological 3D qubit stabilizer models up to local unitary circuits: 
\begin{itemize}
    \item We conjecture that each independent 3D particle implies a 3D toric code (possibly with emergent fermions) can be disentangled via local unitary. Hence a stabilizer model with only 3D particles is equivalent to copies of the 3D toric code, possibly including one copy with a fermionic point particle in the case of a non CSS model. This is because a stack of two 3D toric codes, one with a bosonic point particle and the other with a fermionic point particle is equivalent to a stack of two 3D toric codes both with fermionic point particles, up to a change of excitation basis. 
    \item We similarly conjecture that each independent 2D particle implies a 2D toric code can be disentangled. We further conjecture that any planon stabilizer model is equivalent to a stack of toric codes, which may be a slightly stronger statement. 
\end{itemize}
The properties of each different class of topological order discussed in this paper are key to understanding the entanglement renormalization flow for general 3D topological stabilizer models. If our conjectures are true, stabilizer models with only 3D particles flow to RG fixed points and for each 2D particle, a stack of 2D toric codes can be extracted in a bifurcating renormalization flow. Furthermore, the use of the stack of 2D toric codes as a resource in the definition of foliated phases is essentially unique amongst stabilizer models. Results along these lines will be presented in a forthcoming work~\cite{Dua_RG_2019}. 
Several other directions for future work have presented themselves during the coarse of this study. These include: the generalization of our methods beyond stabilizer models using tensor network techniques following Ref.~\onlinecite{ribbons}, adapting our methods to apply to subsystem symmetry-protected phases, and the search for a rigorous no strings condition for stabilizer models that is necessary and sufficient while also remaining practical to verify.

%In a separate work~\cite{RG_stabilizer_models}, we discuss the renormalization of 3D stabilizer models as a tool to understand the structure theorem~\cite{Haah2018} for 3D stabilizer models. We conjecture that in a given stabilizer model, condensing particles of different mobilities, starting from let's say the highest mobility particle, we extract out some fundamental models. For example, extracting out a 3D particle is equivalent to extracting a 3D toric code under the assumption that all TQFT stabilizer models are just copies of 3D toric code. This can be done for models within the classes of topological Order we have defined in this work. 

\acknowledgments
AD thanks Jeongwan Haah and Mengzhen Zhang for useful discussions. This work is supported by start-up funds at Yale University (DW and MC), NSF under award number DMR-1846109 and the Alfred P. Sloan foundation (MC).

\bibliographystyle{apsrev_1}
\bibliography{DomBib}

\appendix

\begin{widetext}

\section{Sufficient condition for Deformability}
\label{sec:SCdeformability}

In this appendix, we formulate a sufficient condition and a general recipe which can be employed to deform a pair-creation operator to a certain ``normal form" i.e. to clean a pair-creation operator to a flat-rod structure for example. The main merit of our approach is that it can be applied to any stabilizer model. In particular, we have no constraint on the number of physical qubits placed on each vertices. Without loss of generality, we assume that there are $n_v$ qubits for each vertex. Broadly speaking, there are two types of elementary moves we are considering.

The first move is what we refer to as ``3D corner cleaning." We formulate a condition under which a corner of a string segment can be ``cleaned." That is, given a Pauli operator $P$ representing a pair-creation operator, we would like to show that there is another Pauli operator $P^\prime$ whose support is reduced by a corner of the support of $P$. Of course, $P$ and $P^\prime$ will be a different operator in general, but their action on the ground state will be identical. The second move is what we refer to as ``2D corner cleaning". This is analogous to the 3D corner cleaning, except for the fact that we are considering a pair-creation operator supported on a two-dimensional sheet of thickness one. 

Both of these moves can be understood in terms of a procedure called ``local cleaning". Local cleaning refers to a process by which a pair-creation operator is converted to an equivalent pair-creation operator supported on a smaller subsystem by removing one of the vertices. For a given vertex $a$, we formulate a condition under which the support of the pair-creation operator can be cleaned to have a trivial action on $a$.

In the ensuing discussion, it is convenient to use a well-known symplectic representation of a Pauli operator. Consider a Pauli operator acting on $n$ qubits. Any such operator, up to a phase, can be represented as a $2n$-dimensional vector with a $\mathbb{Z}_2$ base field. For example, $(x_1,x_2,z_1,z_2)$ represents an operator $X_1^{x_1}X_2^{x_2} Z_1^{z_1}Z_2^{z_2}$, where the subscript refers to the qubit index. More generally, the first $n$ entries represent the Pauli $X$ operators and the remaining $n$ entries represent the Pauli $Z$ operators. A commutation relation between Pauli operators can be formulated concisely as follows. Two Pauli operators represented by $v_1$ and $v_2$ commute with each other if and only if
\begin{equation}
    v_1 \Omega v_2^T = 0,
\end{equation}
where 
\begin{equation*}
    \Omega := 
    \begin{pmatrix}
    0_n & I_n \\
    I_n & 0_n
    \end{pmatrix}
\end{equation*}
is the symplectic matrix. Here $I_n$ is a $n\times n$ identity matrix and $0_n$ is a $n\times n$ matrix with zero entries.

\subsection{Local cleaning}
We begin by formulating some intuitive notions precisely.
\begin{definition}
For $a \in A$,
\begin{enumerate}
    \item $\text{In}(a,A)$ is the set of stabilizers with nontrivial support on $a$ which are supported on $A$.
    \item $\text{Out}(a,A)$ is the set of stabilizers with nontrivial support on $a$ which are supported on $\{a\} \cup A^c$.
\end{enumerate}
\end{definition}
In other words, $\text{In}(a,A)$ is a set of stabilizers that are ``inside of" $A$ that acts nontrivially on $a$. Similarly, $\text{Out}(a,A)$ is a set of stabilizers that are ``outside of" $A$ that acts nontrivially on $a$. 

Now, we can formulate our strategy in words. By using the fact that the pair-creation operator must commute with the stabilizers, we will derive a nontrivial constraint on the Pauli operator acting on $a$. The constraint comes from the fact that this Pauli operator must commute with any stabilizer in $\text{Out}(a,A)$. Then, we would like to deform the pair-creation operator by multiplying a judicious set of stabilizers in $\text{In}(a,A)$. Note that multiplying such a stabilizer does not expand the support of the pair-creation operator. 

A nontrivial question is whether one can always find a set of stabilizers in $\text{In}(a,A)$ that achieves this goal. To gain some intuition, it is instructive to consider an example, say, Cubic code 10, see appendix~\ref{zoo:cc10}. 
In this case, we have $2$ qubits per vertex. The space of Pauli operators on these two qubits can be represented by a $2\times 2=4$-dimensional $\mathbb{Z}_2$-valued vector space. For each corner, there are two constraints coming from the commutation relation with the two stabilizers, one X-type stabilizer and the other Z-type, that share just that corner. These are linear constraints, so a set of vectors that obeys these constraint forms a linear subspace. This effectively constrains a set of possible Paulis that the pair-creation operator can have at $a$. Now the question is whether any such Pauli can be removed by multiplying a stabilizer in $\text{In}(a,A)$. As we will show below, this again can be formulated in the language of linear algebra.

In this somewhat restricted example, we made two important observations that apply more generally. First, both the commutation relation with $\text{Out}(a,A)$ as well as the reduction of the support of the pair-creation operator by multiplying a stabilizer from $\text{In}(a,A)$ can be formulated in the language of linear algebra over $\mathbb{Z}_2$. Second, in this discussion, we did not need the information about the entire stabilizer itself. Rather, all that mattered was the restriction of the stabilizers in $\text{In}(a,A)$ and $\text{Out}(a,A)$ on $a$.

\begin{definition}
We define constraint matrices $\mathcal{C}_{\text{in}}(a,A)$ and $\mathcal{C}_{\text{out}}(a,A)$ as 
\begin{equation}
    \begin{aligned}
        \mathcal{C}_{\text{in}}(a,A)&=
        \begin{pmatrix}
        s_1^{(i)}|_a \\
        \vdots \\
        s_{n_i}^{(i)}|_a
        \end{pmatrix} \\
        \mathcal{C}_{\text{out}}(a,A) &=
        \begin{pmatrix}
        s_1^{(o)}|_a \\
        \vdots \\
        s_{n_o}^{(i)}|_a
        \end{pmatrix},
    \end{aligned}
\end{equation}
where $s_i^{(i)} \in \text{In}(a,A)$ and $s_i^{(o)} \in \text{Out}(a,A)$. For each stabilizer $s,$ $s|_a$ is a symplectic representation of a restriction of $s$ to $a$. 
\end{definition}
To be clear, a restriction of a stabilizer to $a$ is defined as follows. Let $s=p_a\otimes O$ for some operator $O$ and a Pauli operator $p_a$ acting on $a$. A restriction of $s$ to $a$ is $p_a$.

Here is our key lemma.
\begin{lemma}(Local cleaning lemma)
Consider a pair-creation operator $P$ acting on $A$. Let $a\in A$. If
\begin{equation}
\ker (\mathcal{C}_{\text{out}}(a,A) \Omega) = \Ima(\mathcal{C}_{\text{in}}(a,A))\label{eq:local_cleaning}
\end{equation}
then there exists an equivalent pair-creation operator $P'$ acting on. $A\setminus \{ a\}.$
\end{lemma}
\begin{proof}
Let $v$ be a restriction of the pair-creation operator to $a$. Because $v$ must commute with all the stabilizers in $\text{Out}(a,A)$, we have
\begin{equation}
    \mathcal{C}_{\text{out}}(a,A) \Omega v^T = 0 .
\end{equation}
If the condition is satisfied, the orthogonal complement of the vector space spanned by rows of $\mathcal{C}_{\text{out}}(a,A) \Omega$ must be spanned by the rows of $\mathcal{C}_{\text{in}}(a,A)$.

Therefore, no matter what $v$ is, there is a stabilizer in $\text{In}(a,A)$ whose restriction on $a$, in the symplectic representation, is $v$. Therefore, there is a stabilizer supported on $A$, which, upon multiplying with $P$, reduces the support to $A\setminus \{a\}$. 
\end{proof}

There is a generalization of this lemma which will prove useful for certain models, e.g., CC8. The main idea is to consider a set of sites as opposed to a single site. That is, $a$ is no longer a single site but rather a subset of $A$. A restriction of a stabilizer to $a$ will thus generally become a longer vector. 
\begin{lemma}(Generalized local cleaning lemma)
Consider a pair-creation operator $P$ acting on $A$. Let $B\subset A$. If
\begin{equation}
\ker (\mathcal{C}_{\text{out}}(B,A) \Omega)|_a = \Ima(\mathcal{C}_{\text{in}}(a,A))\label{eq:generalized_local_cleaning}
\end{equation}
then there exists an equivalent pair-creation operator $P'$ acting on. $A\setminus \{a\}.$
\end{lemma}
Here $|_a$ is projection onto the Pauli operator acting on $a\in B$. The proof is analogous to the local cleaning lemma.

\subsection{Cleaning strategy~\label{section:cleaning_strategy}}
The local cleaning lemma is extremely general. For any model, we can apply the following greedy strategy to reduce the support of the pair-creation operator. Namely, 
\begin{enumerate}
    \item For each $a$ in the support of the pair-creation operator, check eq.~\ref{eq:local_cleaning}. 
    \item If the condition is satisfied, clean $a$.
    \item Repeat.
\end{enumerate}

There is a freedom in choosing the order of the cleaning. However, often it is convenient to follow these steps. Without loss of generality, suppose we have a pair-creation operator contained in some box. We first clean this box so that the operator is supported on a box that tightly encloses the excitations. We will explain this process once we set up our notation. Provided that this is done already, we repeat the following.
\begin{enumerate}
    \item Pick a corner of the box for which eq.~\ref{eq:local_cleaning} is satisfied. Repeatedly apply the same cleaning move for every corner in the same direction. 
    \item This results in a box with a 2D structure. 
    \item Clean the 2D structure in a similar manner.
    \item Go to the other corners and repeat the same procedure.
\end{enumerate}

There is a notable exception for which this strategy does not work. This is CC8. The main challenge here lies in using Lemma 2. The stabilizer generators of Cubic code 8 are given by
\begin{align}
\begin{array}{c}
\drawgenerator{XI}{II}{IX}{XI}{XI}{XX}{XX}{XX}
\quad
\drawgenerator{IZ}{ZZ}{ZZ}{ZZ}{IZ}{II}{ZI}{IZ}
\end{array}\, .
\end{align}
Table \ref{tab:CC8_corner} summarizes the $\mathcal{C}_{\text{out}},\mathcal{C}_{\text{in}}$ for each corner of a cube. The origin $(0,0,0)$ is at the $XX$ of the $X$-stabilizer. Note that there are $8$ different types of corners. We will label these types following the convention of Fig.~\ref{fig:corner_conventions}.

\begin{figure}[!htb]
    \centering
    \begin{align}
    \drawgenerator{A}{B}{C}{D}{A'}{B'}{C'}{D'}
    \end{align}
    \caption{There are $8$ different types of corners. The shape of the pair-creation operator doesn't have to cube or cuboid but can always be confined in a cuboidal box as shown. Two corners, if they are of the same type, have identical constraint matrices.}
    \label{fig:corner_conventions}
\end{figure}

Now that we have set up our convention, let us explain the procedure to clean a general operator to an operator in a box that tightly encloses the excitations. As a starting point, without loss of generality, we can assume that the operator is supported on some finite cuboid. Suppose that the type-$B$ corner can be cleaned. By repeatedly cleaning the type-$B$ corner, we will end up having a cuboid with two sheets attached to it. The sheet normal to the $x$-direction will have a corner of type $A$ which we can attempt to clean. The sheet normal to the $y$-direction will have a corner of type $C'$ which we can also attempt to clean. Suppose at least one of them can be cleaned. By repeating the cleaning process, one ends up having a single sheet. The remaining sheet can be cleaned if at least one of its corners (e.g., $D'$ or $C'$ on a sheet normal to the $y$-direction and $A$ or $D'$ on a sheet normal to the $x$-direction) can be cleaned. If either of these two conditions holds, we can completely remove the parts of the cuboid whose $z$ coordinate is larger than the $z$ coordinates of the excitations. More precisely, we will be able to remove parts of the cuboid whose $z$ coordinate satisfies $z> z_{\text{max}}$, where $z_{\text{max}}$ is highest $z$ coordinate in the support of the excitation. By the symmetries of the models we have considered, one can remove the parts of the cuboid whose $z$ coordinate is smaller than the lowest $z$ coordinate of the excitations. Note that we did not need to begin this procedure by cleaning corner type $B$. In order to clean the top portion of the cuboid, one could have started by cleaning vertices or corners of type $A$, $C'$, or $D'$ followed by cleaning the appropriate set of sheets for those choices. One can similarly squash the cuboid down to the minimal box in the other directions as well. We verified that this was possible in all the models we considered, except for the Hsieh-Halasz-II model.

\begin{table}[h]
    \centering
    \begin{tabular}{cccc}
    \hline
    Corner type & $\mathcal{C}_{in}$ & $\mathcal{C}_{out}$ & Cleanable?  \\
    \hline
    $B$ & $\begin{pmatrix} 0&0&1&1 \end{pmatrix}$& $\begin{pmatrix} 1&1&0&0 \\ 0&0&1&1 \\ 0&0&0&1 \end{pmatrix}$ & Yes \\
    \hline
    $C'$ & $\begin{pmatrix} 1&1&0&0 \\ 0&0&1&0 \end{pmatrix}$ & $\begin{pmatrix} 0&1&0&0 \\ 0&0&1&1 \end{pmatrix}$ & Yes \\
    \hline
    $A$ & $\begin{pmatrix} 1&0&0&0 \\ 0&0&0&1 \end{pmatrix}$ & $\begin{pmatrix} 1&0&0&0 \\ 0&0&0&1 \end{pmatrix}$ & No \\
    \hline
    $D$ & $\begin{pmatrix} 1&0&0&0 \\ 0&0&1&1 \end{pmatrix}$ & $\begin{pmatrix} 1&1&0&0 \\ 0&0&0&1 \end{pmatrix}$ & Yes \\
    \hline
    $D'$ & $\begin{pmatrix} 1&1&0&0 \\ 0&0&0&1 \end{pmatrix}$ & $\begin{pmatrix} 1&0&0&0 \\ 0&0&1&1 \end{pmatrix}$ & Yes \\
    \hline
    $C$ & $\begin{pmatrix} 0&1&0&0 \\ 0&0&1&1 \end{pmatrix}$ & $\begin{pmatrix} 1&1&0&0 \\ 0&0&1&0 \end{pmatrix}$ & Yes\\
    \hline
    $A'$ & $\begin{pmatrix} 1&0&0&0 \\ 0&0&0&1 \end{pmatrix}$ & $\begin{pmatrix} 1&0&0&0 \\ 0&0&0&1 \end{pmatrix}$ & No\\
    \hline
    $B'$ & $\begin{pmatrix} 1&1&0&0 \end{pmatrix}$& $\begin{pmatrix} 0&0&1&1  \end{pmatrix}$ & No\\
         \hline
    \end{tabular}
    \caption{Cleanability criteria of CC8 from eq.~\ref{eq:local_cleaning}. For the convention on the corner type, see Fig.~\ref{fig:corner_conventions}.}
    \label{tab:CC8_corner}
\end{table}
We will also consider a cleaning strategy for corners of two-dimensional sheets that are normal to the $\hat{x}$, $\hat{y}$, and $\hat{z}$ directions. Let us first summarize the result and sketch the proof.

\begin{table}[!htb]
    \centering
    \begin{tabular}{cccc}
        \hline
         Sheet& Corner type & Cleanable?  & Lemma \\
         \hline
        \multirow{2}{*}{$\hat{x}$} & $B,C,B',C'$ & Yes & 1 \\
        & $A,D,A',D'$ & Yes & 2 \\
        \hline
        \multirow{2}{*}{$\hat{y}$} & $A,B,A',B'$ & Yes & 1 \\
        & $C,D,C',D'$ & No &  \\
        \hline
        \multirow{2}{*}{$\hat{z}$} & $B,D,B',D'$ & Yes & 1 \\
        & $A,C,A',C'$ & Yes & 2 \\
        \hline
    \end{tabular}
    \caption{Cleanability criteria of CC8 for corners of 2D sheet. By the $\mathbb{Z}_2$ symmetry that exchanges the $X$ and $Z$ stabilizer, the cleanability of two corners in the opposite direction are identical. The vector $\hat{x},\hat{y},\hat{z}$ represents the normal vector of the sheets. For the corners, we follow the convention of Fig.~\ref{fig:corner_conventions}.
    \label{tab:CC8_2dcorner}}
\end{table}

Here, let us sketch how these facts can be used to deform any string segment into union of flat rods. First of all, without loss of generality, suppose two excitations were created by an operator in some finite box. We would like to reduce the support of this operator to a box that tightly encloses the two excitations; see Fig.~\ref{excitation_config}. Then, we would like to clean the type-B corners. This leaves us with two sheets normal to the $x$ and $y$-direction. Because the sheet normal to the $x$ direction can be cleaned. What remains is the sheet normal to the $y$ direction. Because at least one of its corners can be cleaned, this sheet can be cleaned as well. This process cleaned the portion of the operator that is on top of the figures shown in Fig.~\ref{excitation_config} in the $z$-direction. We can apply the analagous process on the other sides, leading us to Fig.~\ref{excitation_config}.

\subsubsection{Corner cleaning on sheets}
Let us discuss how the corner cleaning condition on a 2D surface can be derived. For this discussion, it is convenient to use the following convention. We specify the corners of the sheets by (i) specifying the normal vector of the sheet and (ii) referring to the corners in Fig.~\ref{fig:corner_conventions}. Of course, the actual corner in consideration may differ from the one described in Fig.~\ref{fig:corner_conventions}. However, whether a corner can be cleaned or not is determined by the constraint matrix, and this constraint matrix will be determined uniquely by the type of corner. Specifically, for each sheet, there are $4$ different types of corners, as one can see in Fig.~\ref{fig:corner_conventions}. Once we specify this information, the constraint matrix is uniquely defined. 

We will now consider cleaning of sheets (of thickness one) normal to each lattice direction. Consider the sheet normal to the $x$-direction. For the corner types $B$, $C$, $B^\prime$ or $C^\prime$ in Fig.~\ref{fig:corner_conventions}, we have
\begin{equation}
    \mathcal{C}_{\text{out}} = 
    \begin{pmatrix}
    1&1&0&0\\
    0&1&0&0\\
    0&0&1&1\\
    0&0&0&1
    \end{pmatrix},
\end{equation}
which is full rank. Therefore, this corner can be cleaned. For the corner types $A$, $D$, $A^\prime$ or $D^\prime$, without loss of generality, we can consider three vertices $a,b,c$ described in Fig.~\ref{fig:x1-1}. This leads to the constraint matrix in Eq.~\eqref{eq:x1-1}.

\begin{equation}
\mathcal{C}_{\text{out}}=
    \begin{blockarray}{cccccccccccc}
 \BAmulticolumn{2}{c}{a_{X}} & \BAmulticolumn{2}{c}{b_{X}} & \BAmulticolumn{2}{c}{c_{X}} &     \BAmulticolumn{2}{c}{a_{Z}} & \BAmulticolumn{2}{c}{b_{Z}} & \BAmulticolumn{2}{c}{c_{Z}} &\\
 \begin{block}{(cc|cc|cc|cc|cc|cc)}
1&1&0&0&0&0&0&0&0&0&0&0\\
1&0&0&0&0&0&0&0&0&0&0&0\\
0&0&0&0&0&0&0&1&0&0&0&0\\
0&0&0&0&1&1&0&0&0&0&0&0\\
0&0&0&0&1&0&0&0&0&0&0&0\\
0&0&0&0&0&0&0&0&0&0&0&1\\
0&0&0&0&0&0&0&0&0&1&1&0\\
0&0&0&0&0&0&1&1&0&1&1&1\\
0&1&1&0&0&0&0&0&0&0&0&0\\
1&1&1&1&1&1&0&0&0&0&0&0\\
\end{block}
    \end{blockarray}\label{eq:x1-1}
\end{equation}

From row $1,2,3,6,7,8$ of eq.~\ref{eq:x1-1}, we can see that the kernel of eq.~\ref{eq:x1-1} restricted to $a_X$ must be $(0,0)$. Let $r_i$ be the $i$th row vector. From row $1$ and $2$, one can see that the kernel of $\mathcal{C}_{out}$ must obey the constraint $a_X=(0,0)$. From row $3$ and $r_6 + r_7 + r_8$, one can see that the kernel restricted to $a_Z$ must be $(0,0)$. Therefore, the corner at $(y,z)=(1,0)$ can be cleaned as well.

\begin{figure}
    \centering
    \includegraphics[width=0.125\columnwidth]{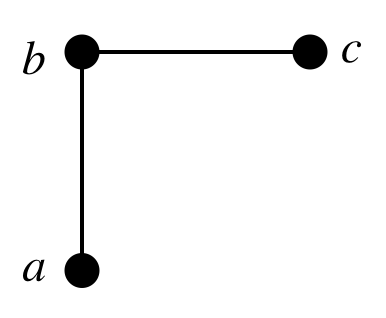}
    \caption{A set of points used in the second order constraint in cleaning of a corner of type $A$, $D$, $A^\prime$ or $D^\prime$. The corner is at position $a$ as shown. Points $a$, $b$ and $c$ lie in the $yz$ plane.}
    \label{fig:x1-1}
\end{figure}

Now, let us consider the sheet normal to the $y$-direction. For the corner types $A$, $B$, $A^\prime$ or $B^\prime$, in Fig.~\ref{fig:corner_conventions}, we have
\begin{equation}
 \mathcal{C}_{\text{out}}=
 \begin{pmatrix}
 1&1&0&0\\
 1&0&0&0\\
 0&0&0&1\\
 0&0&1&1
 \end{pmatrix},
\end{equation}
which is full rank. Therefore, these types of corners can be cleaned. We were unable to verify if the type-$C,D,C',D'$ corners can be cleaned.

Lastly, let us consider the sheet normal to the $z$-direction. For the corner types $A$, $C$, $A^\prime$ or $C^\prime$, we have
\begin{equation}
     \mathcal{C}_{\text{out}}=
     \begin{pmatrix}
     1&0&0&0\\
     0&1&0&0\\
     0&0&0&1\\
     0&0&1&1
     \end{pmatrix},
\end{equation}
which is full rank. Therefore, these types of corners can be cleaned. For the corner types $B$, $D$, $B^\prime$ or $D^\prime$, we have the constraint matrix in eq.~\ref{eq:z11}, which is based on the configuration of $5$ vertices described in Fig.~\ref{fig:z11}.

\begin{equation}
\mathcal{C}_{\text{out}}=
    \begin{blockarray}{cccccccccccccccccccc}
    \BAmulticolumn{2}{c}{a_{X}} & \BAmulticolumn{2}{c}{b_{X}} & \BAmulticolumn{2}{c}{c_{X}} & \BAmulticolumn{2}{c}{d_{X}} & \BAmulticolumn{2}{c}{e_{X}} & \BAmulticolumn{2}{c}{a_{Z}} & \BAmulticolumn{2}{c}{b_{Z}} & \BAmulticolumn{2}{c}{c_{Z}} & \BAmulticolumn{2}{c}{d_{Z}} & \BAmulticolumn{2}{c}{e_{Z}} \\
    \begin{block}{(cc|cc|cc|cc|cc|cc|cc|cc|cc|cc)}
1&1&0&0&0&0&0&0&0&0& 0&0&0&0&0&0&0&0&0&0\\
0&0&0&0&1&1&0&0&0&0& 0&0&0&0&0&0&0&0&0&0\\
0&0&0&0&0&0&0&0&1&1& 0&0&0&0&0&0&0&0&0&0\\
0&0&0&0&0&0&0&0&0&0& 0&1&0&0&0&0&0&0&0&0\\
0&0&0&0&0&0&0&0&0&0& 0&0&0&0&0&1&0&0&0&0\\
0&0&0&0&0&0&0&0&0&0& 0&0&0&0&0&0&0&0&0&1\\
0&1&1&1&1&0&0&0&0&0& 0&0&0&0&0&0&0&0&0&0\\
1&0&1&1&1&1&0&0&0&0& 0&0&0&0&0&0&0&0&0&0\\
0&0&0&0&0&1&1&1&1&0& 0&0&0&0&0&0&0&0&0&0\\
0&0&0&0&1&0&1&1&1&1& 0&0&0&0&0&0&0&0&0&0\\
0&0&0&0&0&0&0&0&0&0& 0&1&0&1&1&0&0&0&0&0\\
0&0&0&0&0&0&0&0&0&0& 1&1&0&0&0&1&0&0&0&0\\
0&0&0&0&0&0&0&0&0&0& 0&0&0&0&0&1&0&1&1&0\\
0&0&0&0&0&0&0&0&0&0& 0&0&0&0&1&1&0&0&0&1\\
    \end{block}
    \end{blockarray}\label{eq:z11}
\end{equation}

In eq.~\ref{eq:z11}, from the row $1,2,5,6,7,8,14$ we can conclude that the restriction on $c$ is $(0,0)$. To see why, let $r_i$ be the $i$th row vector. Note that $r_1+r_7+r_8$ becomes $(0,1)$  on $c_X$. Together with row $2,$ this implies that $c_X$ must be $(0,0)$. Also, $r_6+r_{14}$ becomes $(1,1)$ on $c_Z$ and $r_5$ becomes $(0,1)$ on $c_Z$. Because these two are linearly independent, $c_Z=(0,0)$. Therefore, the corner at $(x,y)=(1,1)$ can be cleaned. More precisely, this corner can be cleaned if in its vicinity we have an arrangement described in Fig.~\ref{fig:z11}.

\begin{figure}
    \centering
    \includegraphics[width=0.1875\columnwidth]{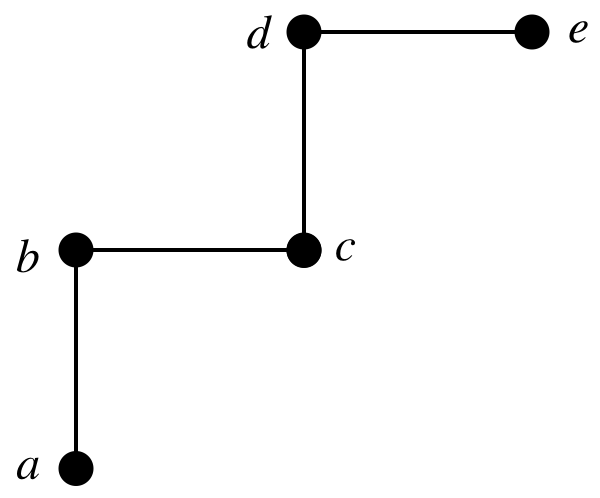}
    \caption{A set of points used in the third order constraint involving the vertex with type $B$, $D$, $B^\prime$ or $D^\prime$ at position $c$ as shown. Points $a$, $b$, $c$, $d$ and $e$ lie in the $xy$ plane.}
    \label{fig:z11}
\end{figure}

Therefore, we conclude that, for any ``stairwell" configuration, one can clean the vertex of the central protruding corner. It is important to note that this is possible independent of the details of the other vertices in the vicinity. Specifically, in Fig.~\ref{fig:z11}, the fact that $c$ can be cleaned is independent of what $a,b,d,$ and $e$ are. By repeatedly applying this fact, one can clean the entire 2D sheet into a one-dimensional subsystem.

\section{Example of deforming pair-creation operators: Cubic code 8}
\label{sec:deformabilityexamples}
In this appendix, we discuss deformability of all pair-creation operators as shown in Fig.~\ref{excitation_config} for cubic code 8. We do not employ the general symplectic formulation in terms of the cleaning lemma used in the Sec.~\ref{sec:SCdeformability}. Instead we equivalently describe deformability in terms of commutation constraints on local Pauli operators. The stabilizer generators of cubic code 8 are given by
\begin{align}
\begin{array}{c}
\drawgenerator{XI}{II}{IX}{XI}{XI}{XX}{XX}{XX}
\quad
\drawgenerator{IZ}{ZZ}{ZZ}{ZZ}{IZ}{II}{ZI}{IZ}
\end{array}
\end{align}

\begin{figure}[H]
\vspace{6mm}
\centering
\sidesubfloat[]{\includegraphics[scale=0.28]{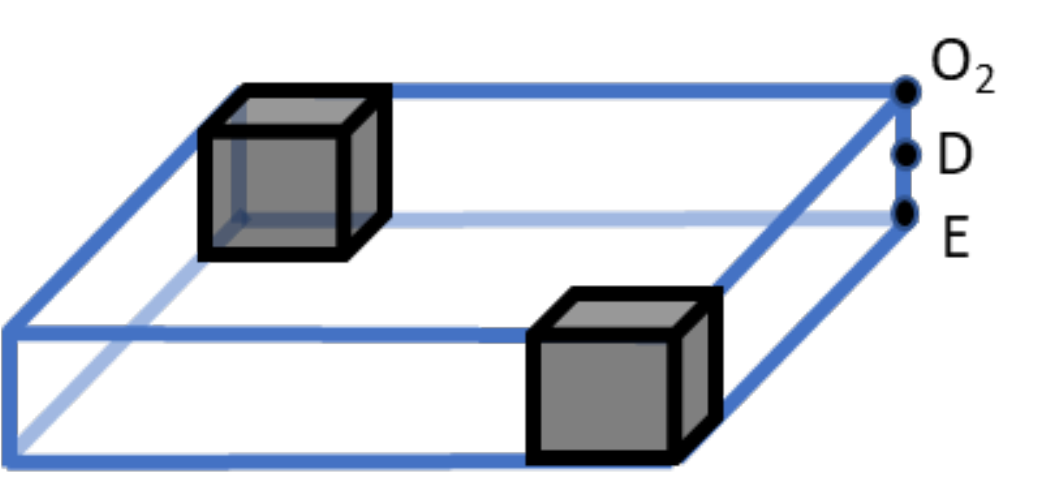}}
\sidesubfloat[]{\includegraphics[scale=0.28]{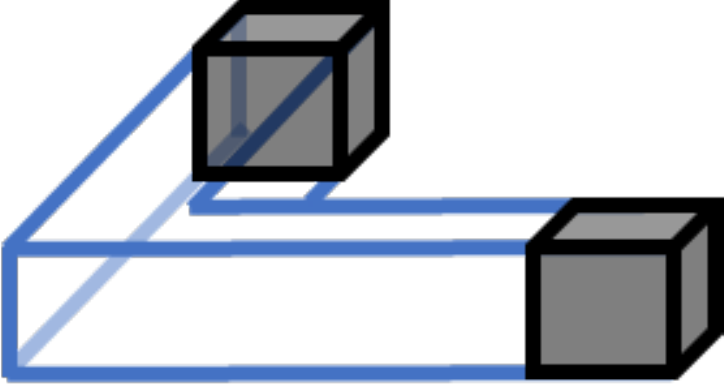}}
\sidesubfloat[]{\includegraphics[scale=0.28]{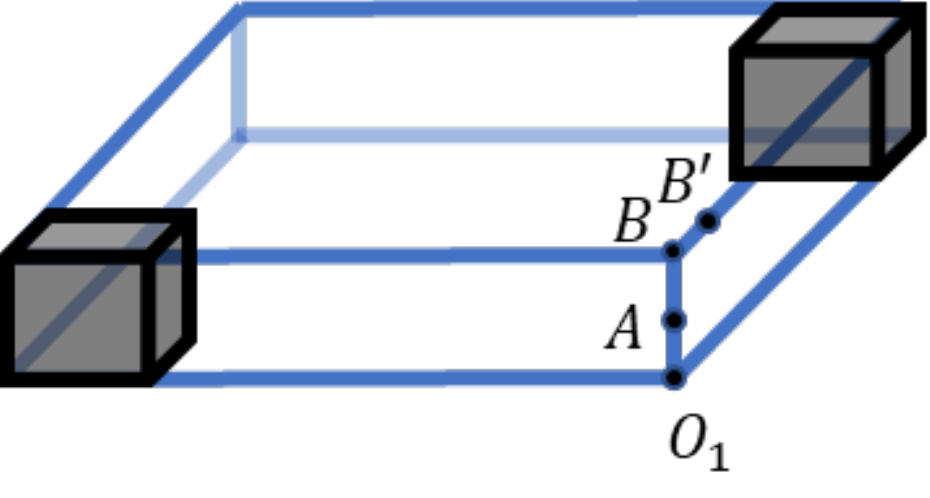}}
\sidesubfloat[]{\includegraphics[scale=0.28]{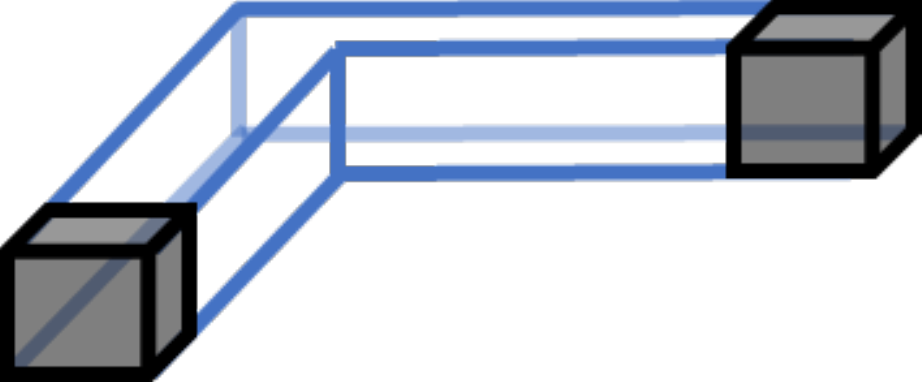}}
\caption{Deformation of $xy$ planar boxes. If the excitations fit within a box along a lattice plane, we refer to it as a planar box.}
\label{CC8_p1}
\end{figure}

For the configuration shown in Fig.~\ref{CC8_p1}~(a), we start by cleaning the corner edge where 3 of the vertices are denoted $O_2$, $D$ and $E$. We show that this edge is completely and independently cleanable and then the same process can be applied subsequently to the protruding edges. $[O_2,ZZ]=0$ implies $O_2=II,XX$. One can keep cleaning the column of operators on these corner edges of the box by multiplying $X$ stabilizers until one is left with a pair of operators acting on $D$ and $E$. $[D\otimes E, IZ\otimes ZZ]=0$, $[D,ZZ]=0$ and $[E,IZ]=0$ implies $D\otimes E$ is either $II\otimes II$ or $XX\otimes XI$ which can be cleaned by multiplying with X stabilizers. In this manner, the protruding edges of the box can be cleaned to yield the flat-rod configuration as shown in Fig.~\ref{CC8_p1}~(b). 

\begin{figure}[H]
\vspace{6mm}
\centering
\sidesubfloat[]{\includegraphics[scale=0.28]{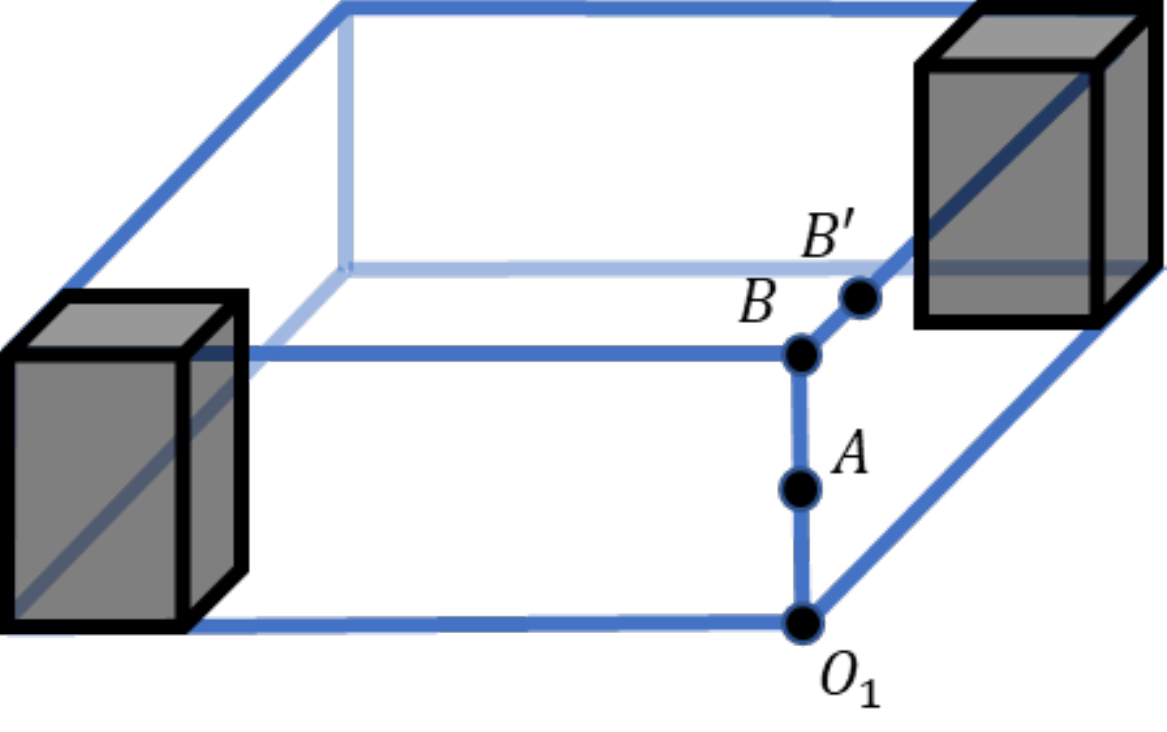}}
\sidesubfloat[]{\includegraphics[scale=0.28]{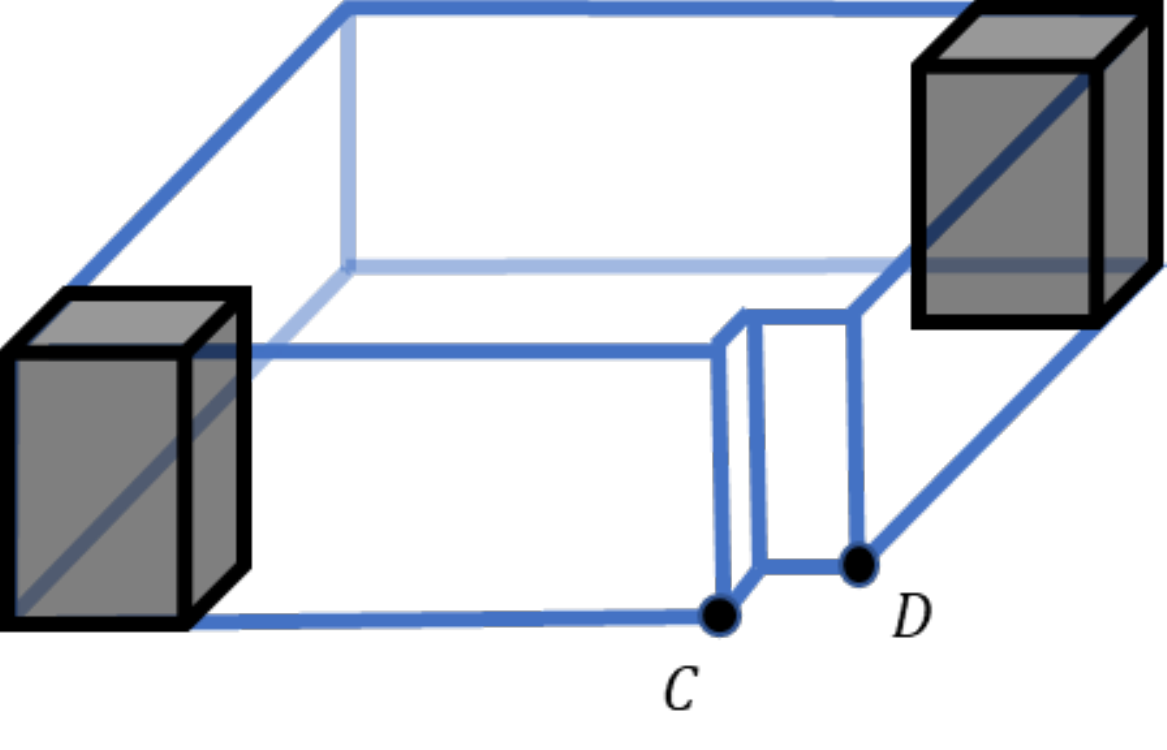}}
\sidesubfloat[]{\includegraphics[scale=0.28]{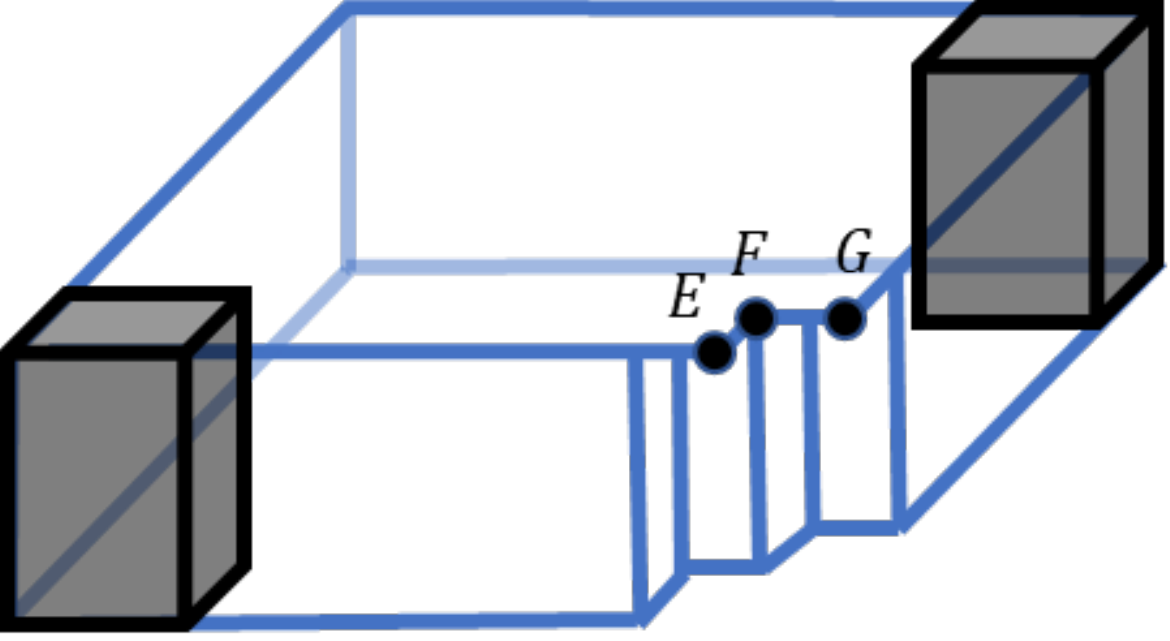}}
\sidesubfloat[]{\includegraphics[scale=0.28]{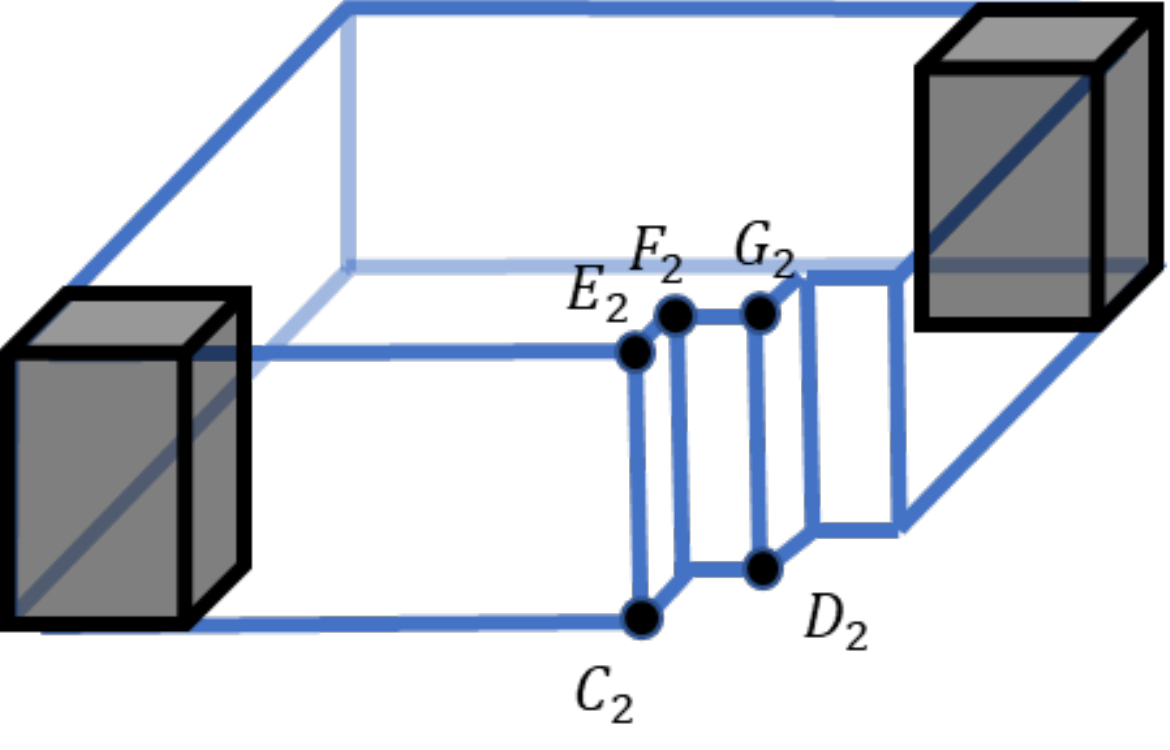}}\\
\sidesubfloat[]{\includegraphics[scale=0.28]{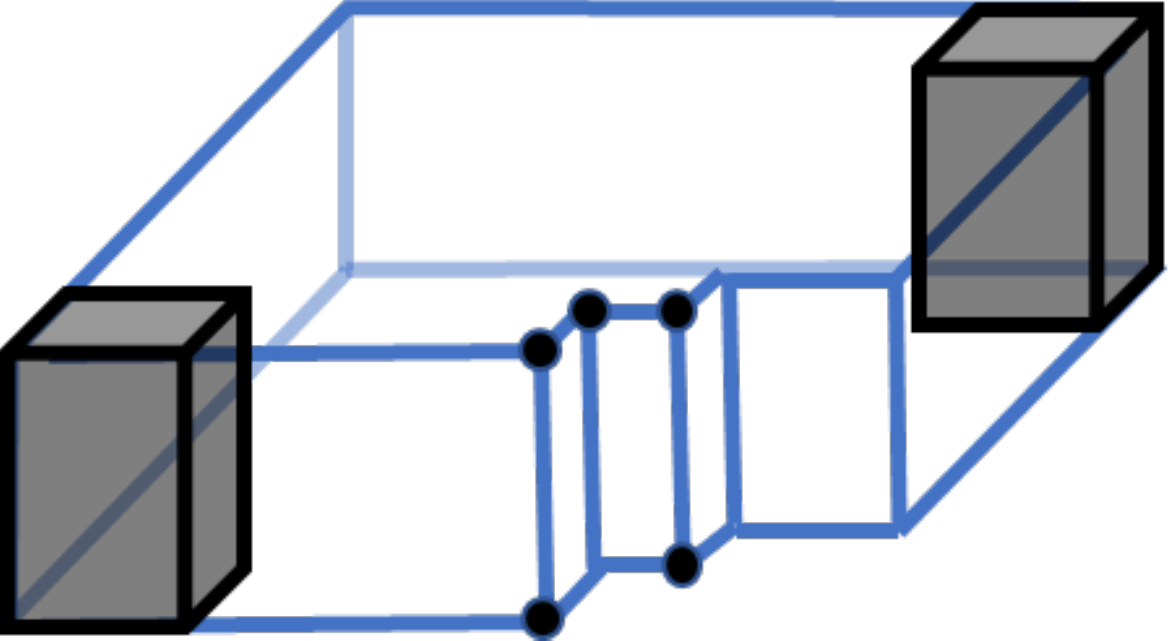}}
\sidesubfloat[]{\includegraphics[scale=0.28]{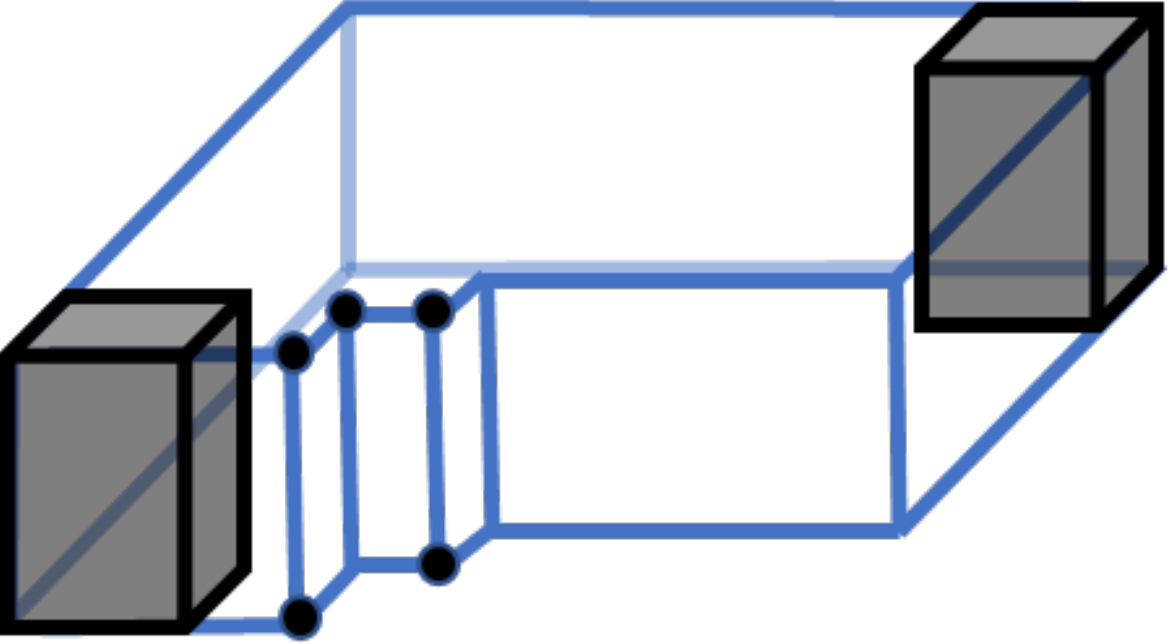}}
\sidesubfloat[]{\includegraphics[scale=0.28]{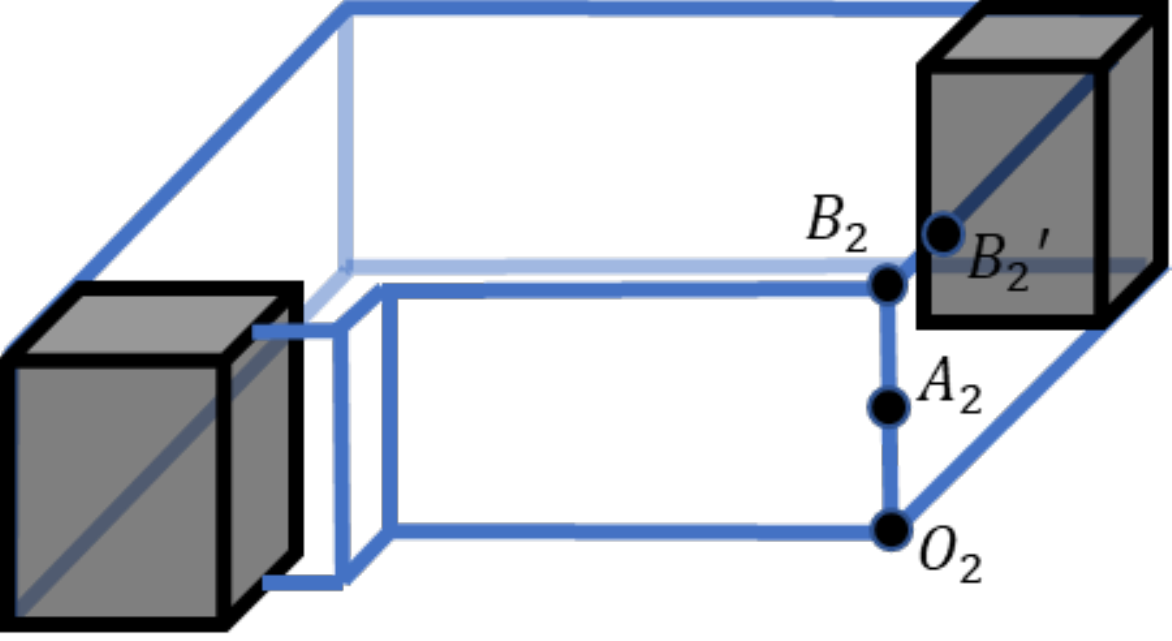}}
\sidesubfloat[]{\includegraphics[scale=0.28]{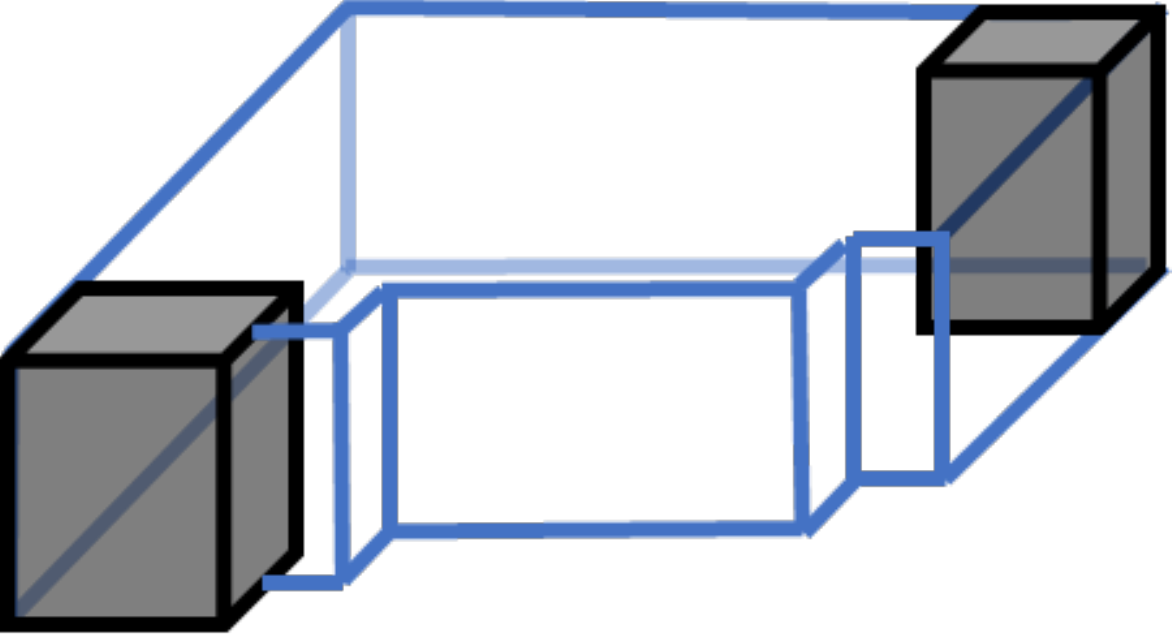}}\\
\sidesubfloat[]{\includegraphics[scale=0.28]{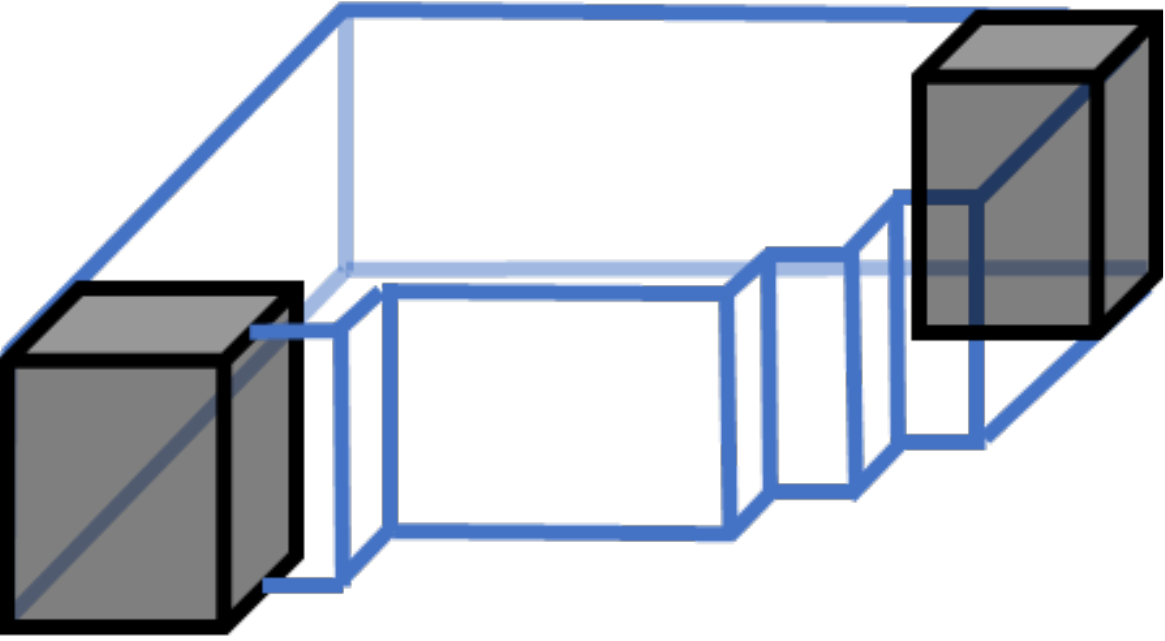}}
\sidesubfloat[]{\includegraphics[scale=0.28]{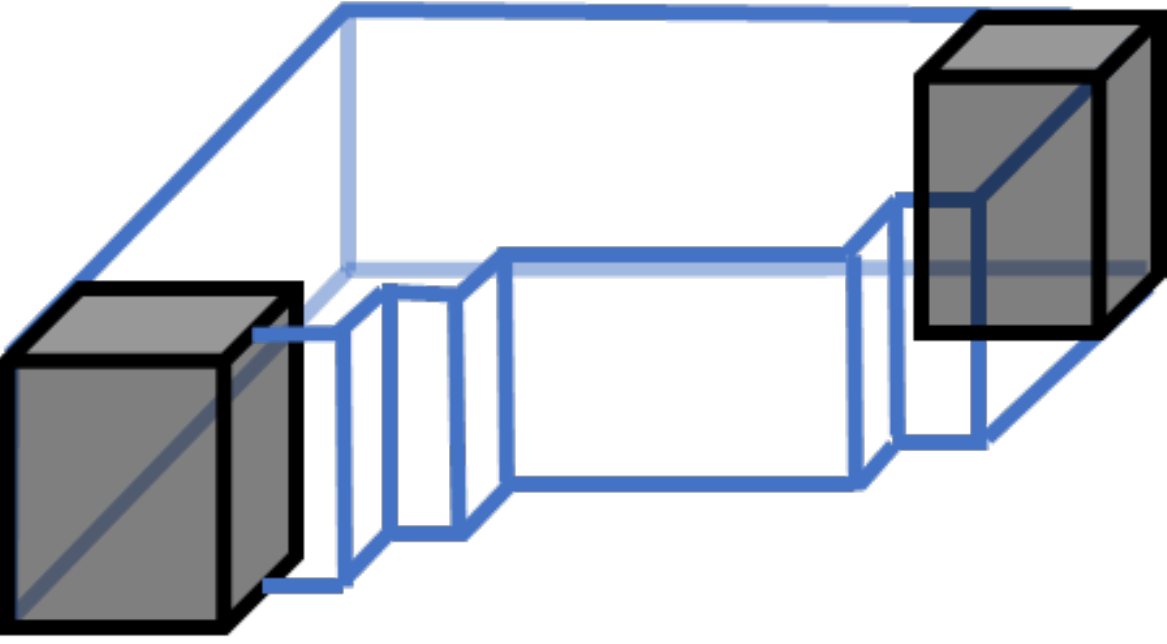}}
\sidesubfloat[]{\includegraphics[scale=0.28]{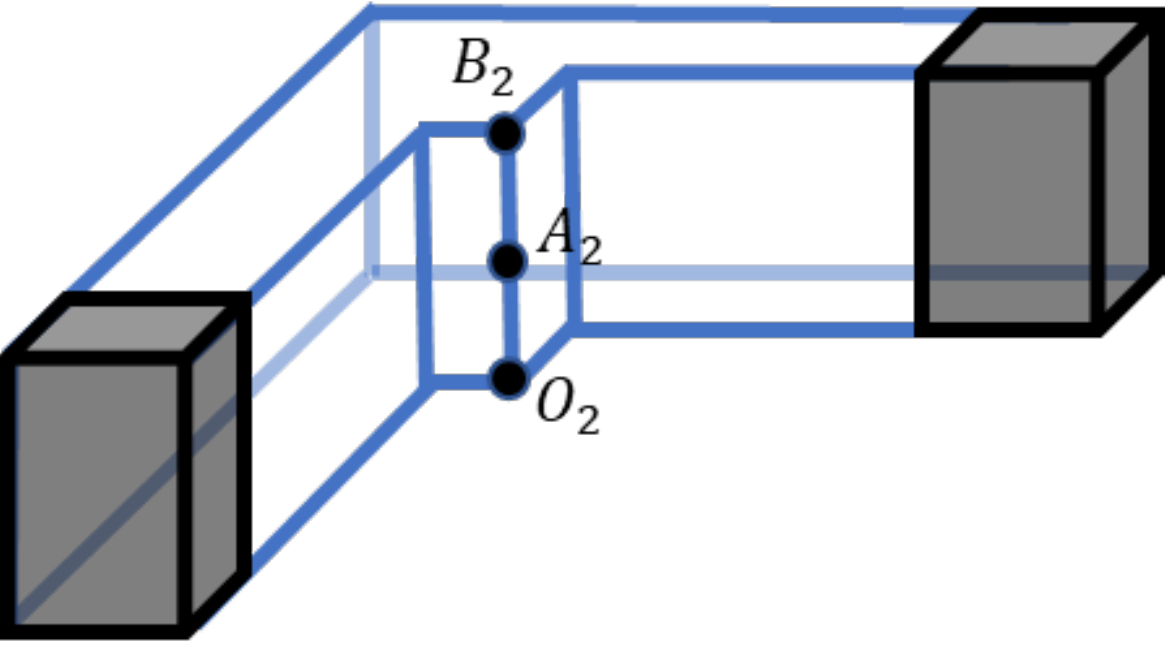}}
\sidesubfloat[]{\includegraphics[scale=0.28]{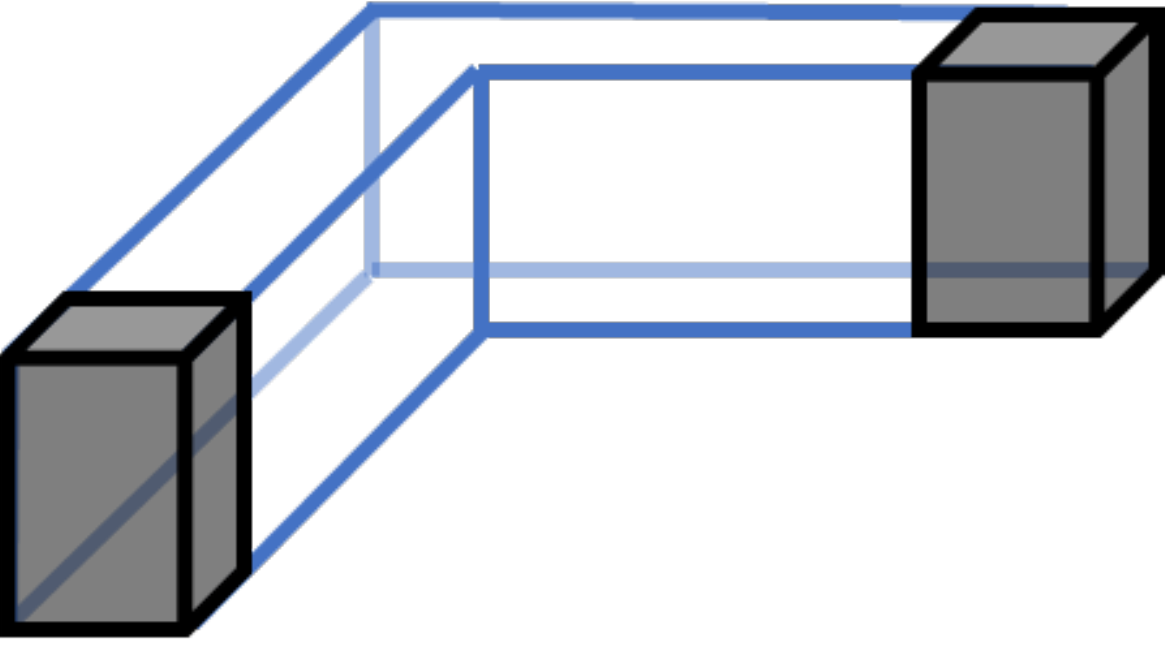}}
\caption{Steps in deformation of the first $xy$ planar box. There is no good edge for cleaning, so we employ higher order commutation constraints.}
\label{CC8_p1e}
\end{figure}
For the configuration of excitations shown in Fig.~\ref{CC8_p1}~(c), the deformation procedure is more complicated. We explain it in detail in Fig.~\ref{CC8_p1e}. In Fig.~\ref{CC8_p1e}~(a), $[O_1,IZ]=0$ implies $O_1=XI$ or $O_1=II$. If $O_1=XI$, it can be cleaned by multiplying with an $X$ stabilizer. Subsequent corners like $A$ have the same constraint and can be cleaned in the same manner. When only the top corner $B$ is left, it obeys constraints $[B,IZ]=0$ and $[B,ZZ]=0$ due to the commutation with the edges $IZ-II$ and $ZZ-II$ of the $Z$-stabilizer. For example, $[B\otimes B^\prime, ZZ\otimes II]=0$. In this manner, the edge joining $O_1$, $A$ and $B$ is cleanable and hence, we get Fig.~\ref{CC8_p1e}~(b). Now, $C$ and $D$ as shown have the same constraint as $O_1$ and can be cleaned. Subsequent equivalent cleaning leads to configuration in Fig.~\ref{CC8_p1e}~(c) where $[E,IZ]$, $[G,IZ]=0$, $[G,ZZ]=0$ and $[E\otimes F\otimes G, ZZ\otimes II\otimes IZ]=0$. Since $[G,IZ]=0$, this implies that $[E,ZZ]=0$. The constraints imply $E=G=II$. Hence, the columns starting from $C$ and $D$ are cleaned. Similarly, we can clean the columns containing $C_2,E_2$ and $D_2,G_2$ respectively. An equivalent procedure leads to the configuration in Fig.~\ref{CC8_p1e}~(g). Again, the edge joining $O_2$, $A_2$ and $B_2$ is cleanable. Subsequent cleaning using the cleaning of columns as described, eventually leads to the configuration shown in Fig.~\ref{CC8_p1e}~(k) which has a cleanable edge joining $B_2$, $A_2$ and $O_2$. Hence, we get the flat-rod configuration of Fig.~\ref{CC8_p1e}~(l). 

\begin{figure}[H]
\vspace{6mm}
\centering
\sidesubfloat[]{\includegraphics[scale=0.28]{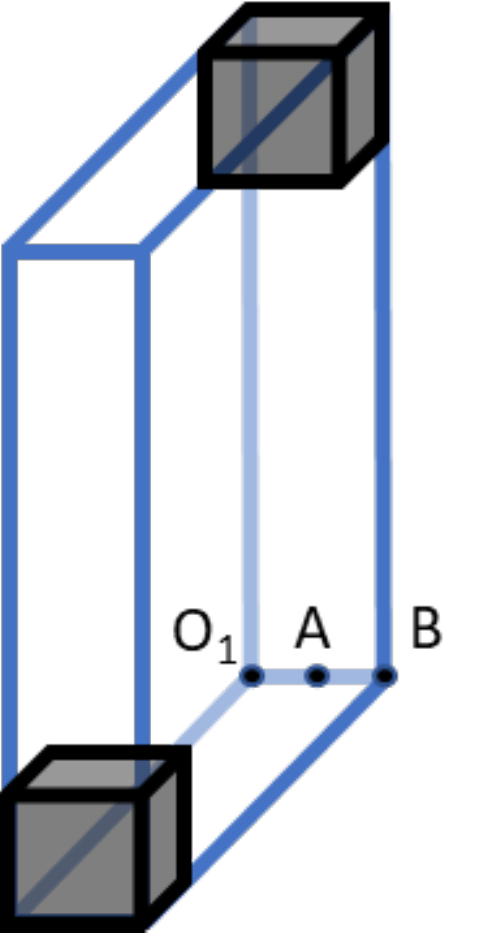}}
\sidesubfloat[]{\includegraphics[scale=0.28]{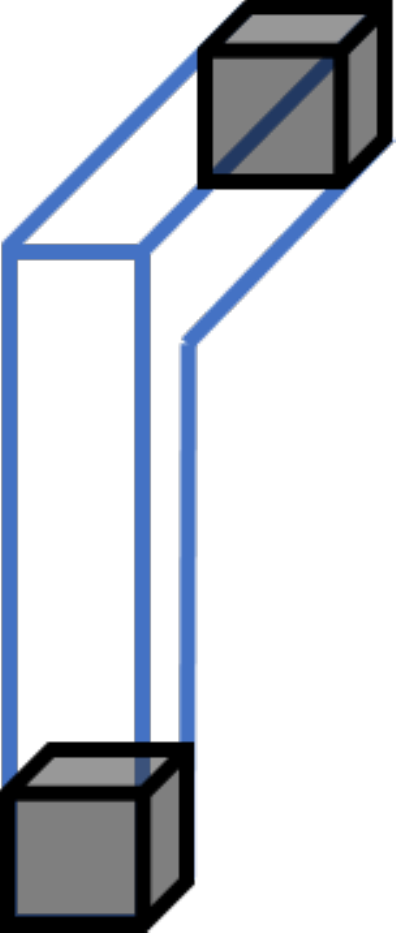}}
\sidesubfloat[]{\includegraphics[scale=0.28]{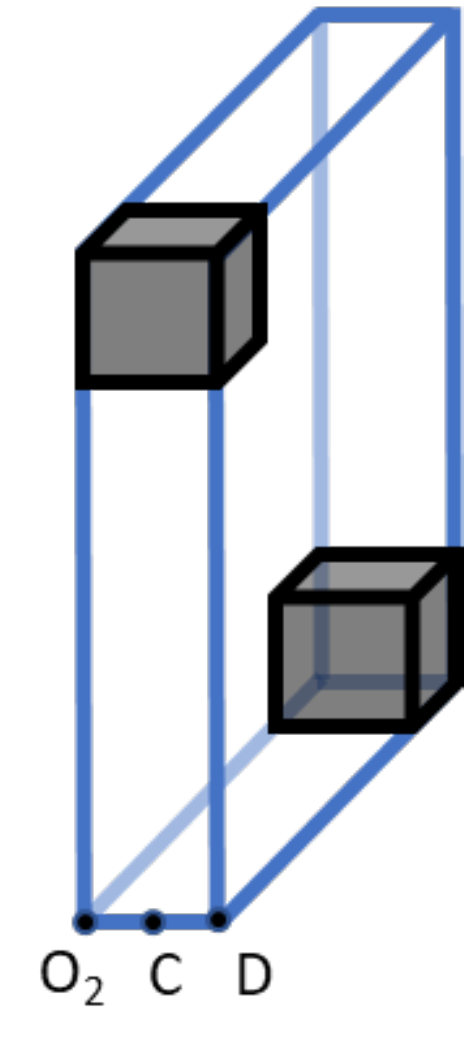}}
\sidesubfloat[]{\includegraphics[scale=0.28]{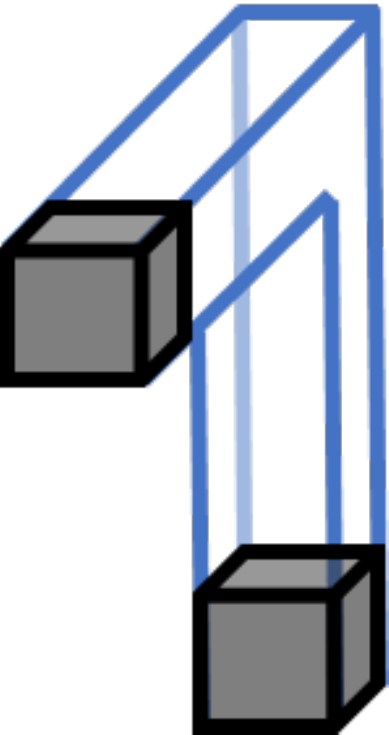}}
\caption{Deformation of $xz$ planar boxes}
\label{CC8_p2}
\end{figure}

In the pair-creation operator in Fig.~\ref{CC8_p2}~(a), $[O_1,ZZ]=0$ implies $O_1=XX,II$. Hence, $O_1$ can be cleaned by multiplying the $X$ stabilizer. Subsequent cleaning of an equivalent corner in this manner will finally leave a pair of vertices $A$ and $B$ which should satisfy the constraints $[A\otimes B, IZ\otimes ZZ]=0$, $[A,ZZ]=0$ and $[B,IZ]=0$. Together these imply that $A\otimes B$ is either $II\otimes II$ or $XX \otimes XI$ which can be again cleaned by multiplying with $X$ stabilizers. Hence, the whole edge has been cleaned. Subsequent cleaning of protruding edges in this manner will lead to the flat-rod configuration as shown. Similarly for the configuration shown in Fig.~\ref{CC8_p2}~(c), one can clean $O_2$ as it obeys $[O_2,ZI]=0$ which implies it is either $IX$ or $II$ and hence can be cleaned by multiplying with $X$ stabilizers. $[C\otimes D, IZ\otimes ZI]=0$, $[C,ZI]=0$ and $[D,IZ]=0$ imply $C\otimes D= IX\otimes XI$ which can be again cleaned by multiplying with $X$ stabilizers. Subsequent cleaning leads to the configuration shown in Fig.~\ref{CC8_p2}~(d).  

\begin{figure}[H]
\vspace{6mm}
\centering
\sidesubfloat[]{\includegraphics[scale=0.28]{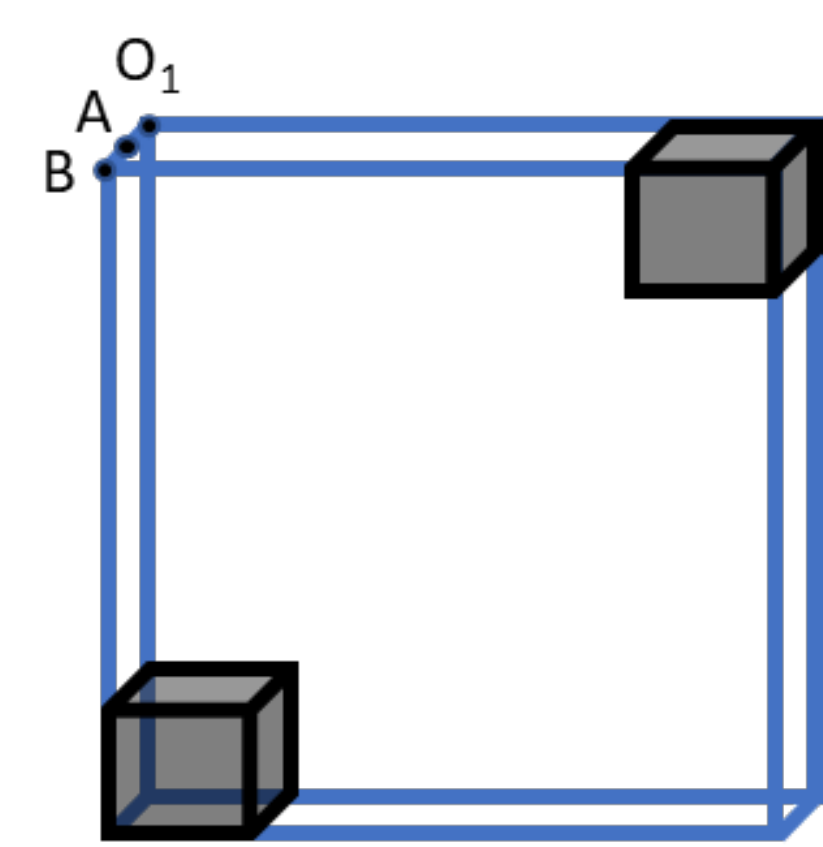}}
\sidesubfloat[]{\includegraphics[scale=0.28]{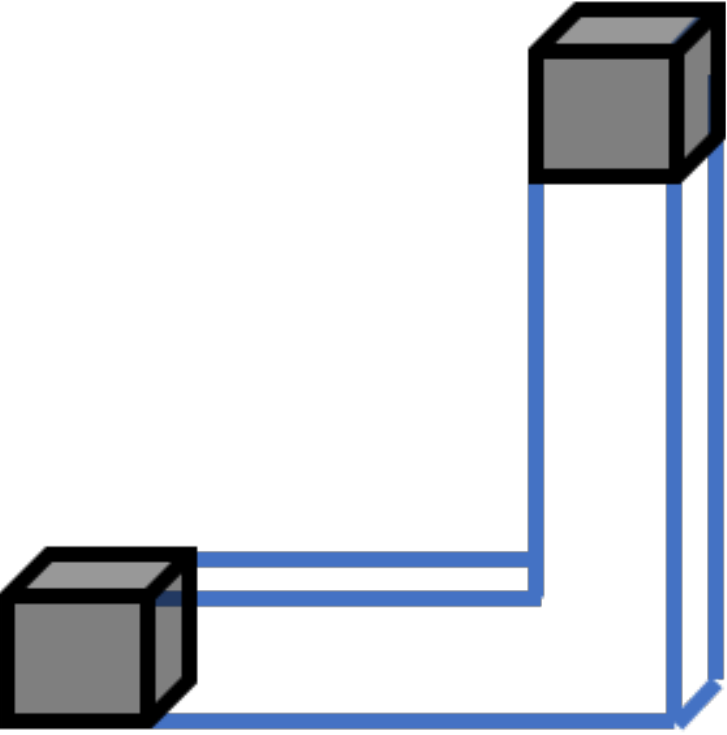}}
\sidesubfloat[]{\includegraphics[scale=0.28]{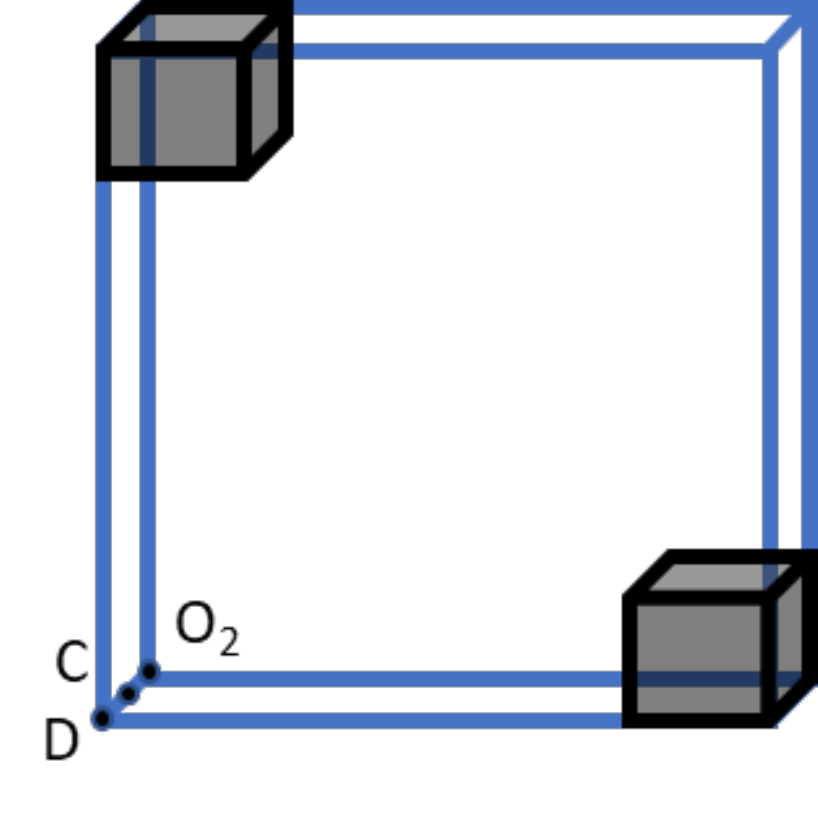}}
\sidesubfloat[]{\includegraphics[scale=0.28]{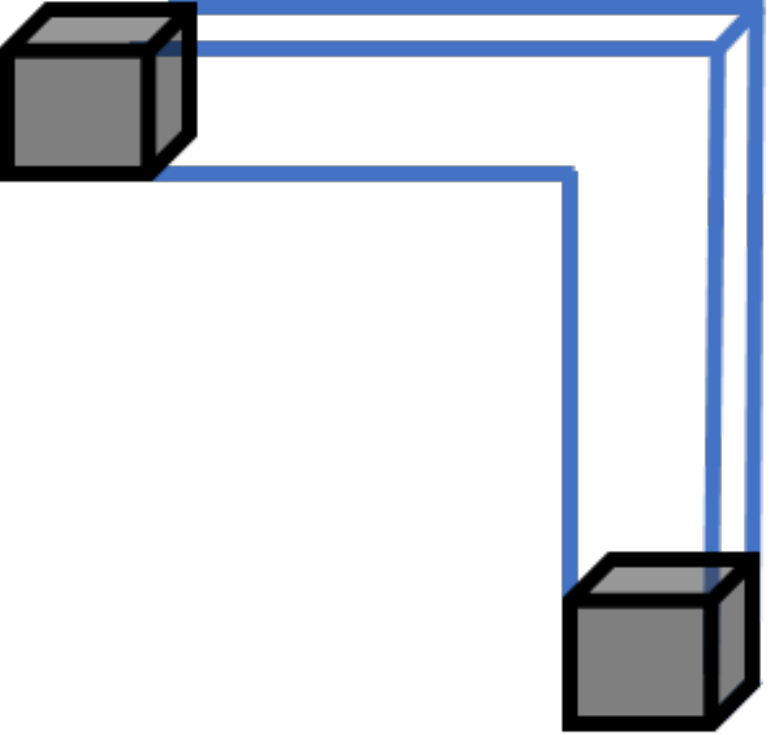}}\caption{Deformation of $yz$ planar boxes}
\label{CC8_p3}
\end{figure}
In Fig.~\ref{CC8_p3}~(a), since $[O_1,ZZ]=0$, $O_1$ and resulting corners can be cleaned as before by multiplying $X$ stabilizer. Final pair of vertices $A$ and $B$ can be cleaned because $[A\otimes B, IZ\otimes ZZ]=0$, $[A,ZZ]=0$ and $[B,IZ]=0$ i.e. since $IZ$ and $ZZ$ are independent. Hence, subsequent cleaning of edges leads to the flat-rod configuration in Fig.~\ref{CC8_p3}~(b). Similarly, in Fig.~\ref{CC8_p3}~(c), the pair of vertices of $Z$-stabilizer that impose commutation constraints on the edge of the box joining $O_2$, $C$ and $D$ are $ZZ$ and $ZI$ which are independent. Hence the edge of the box as shown is cleanable. Subsequent cleaning leads to the flat-rod configuration shown in Fig.~\ref{CC8_p3}~(d).

The cleaning of 3D configurations can be understood in terms of the cleaning of planar boxes we have done in the previous figures. We first clean the box containing the pair of excitation in 3D to 3 planar boxes joined together. This can be done by first cleaning a corner, say $O_1$ and all subsequent corners that creates. The planar boxes thus produced can be cleaned using the strategy empoyed for the planar boxes shown in previous figures leading to the final configuration of flat-rods. 

\begin{figure}[H]
\vspace{6mm}
\centering
\sidesubfloat[]{\includegraphics[scale=0.28]{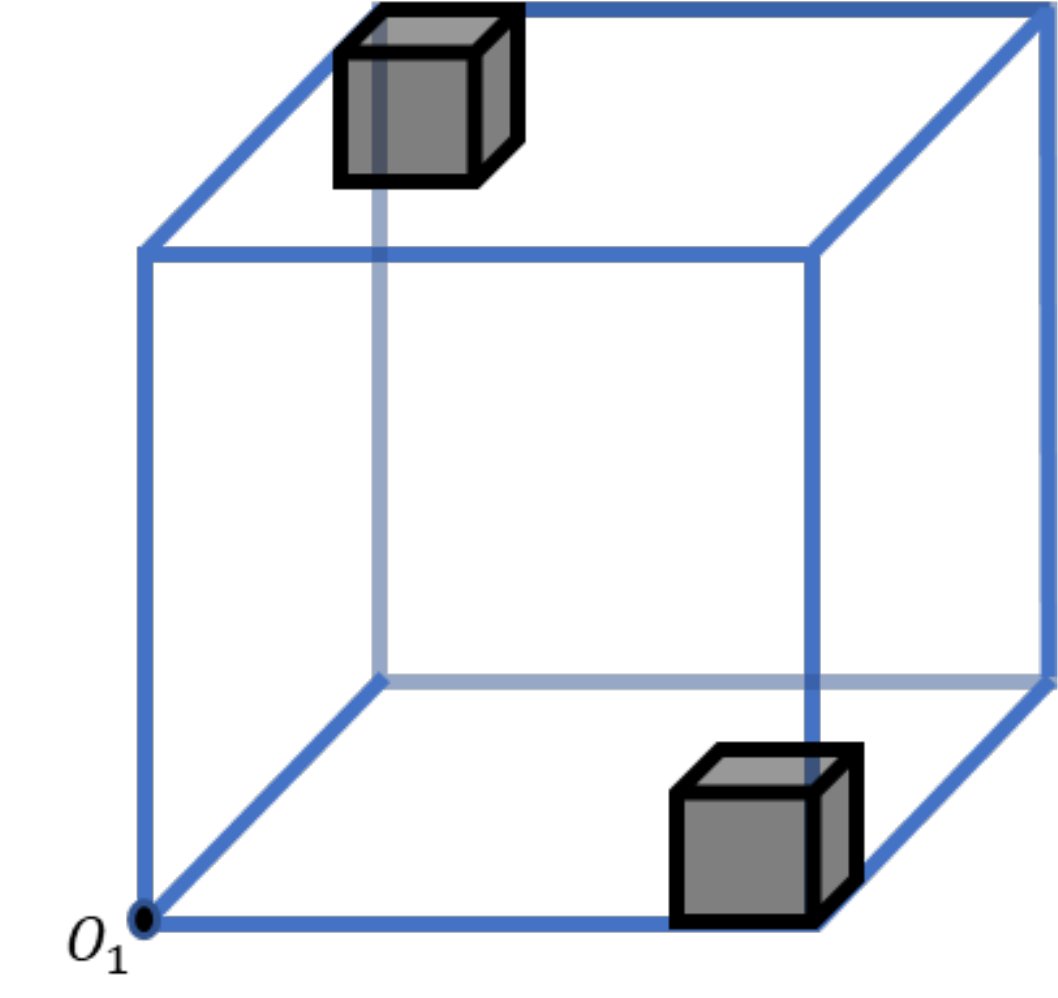}}
\sidesubfloat[]{\includegraphics[scale=0.28]{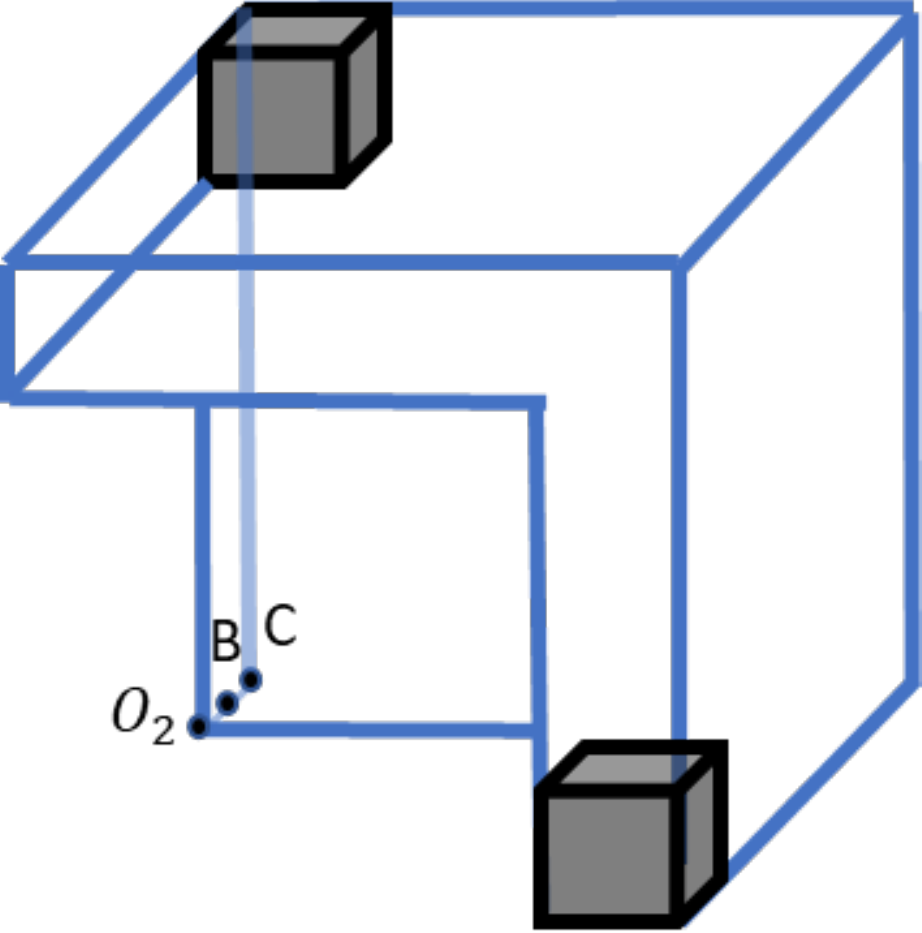}}
\sidesubfloat[]{\includegraphics[scale=0.28]{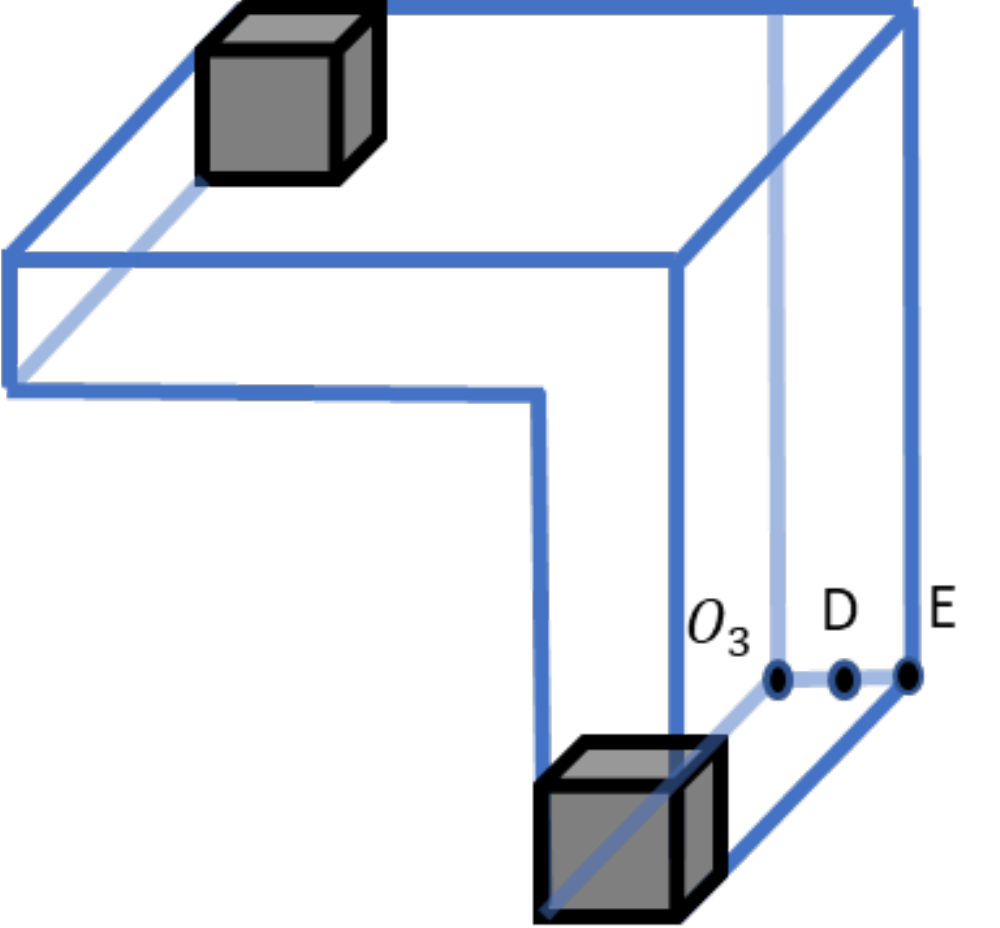}}
\sidesubfloat[]{\includegraphics[scale=0.28]{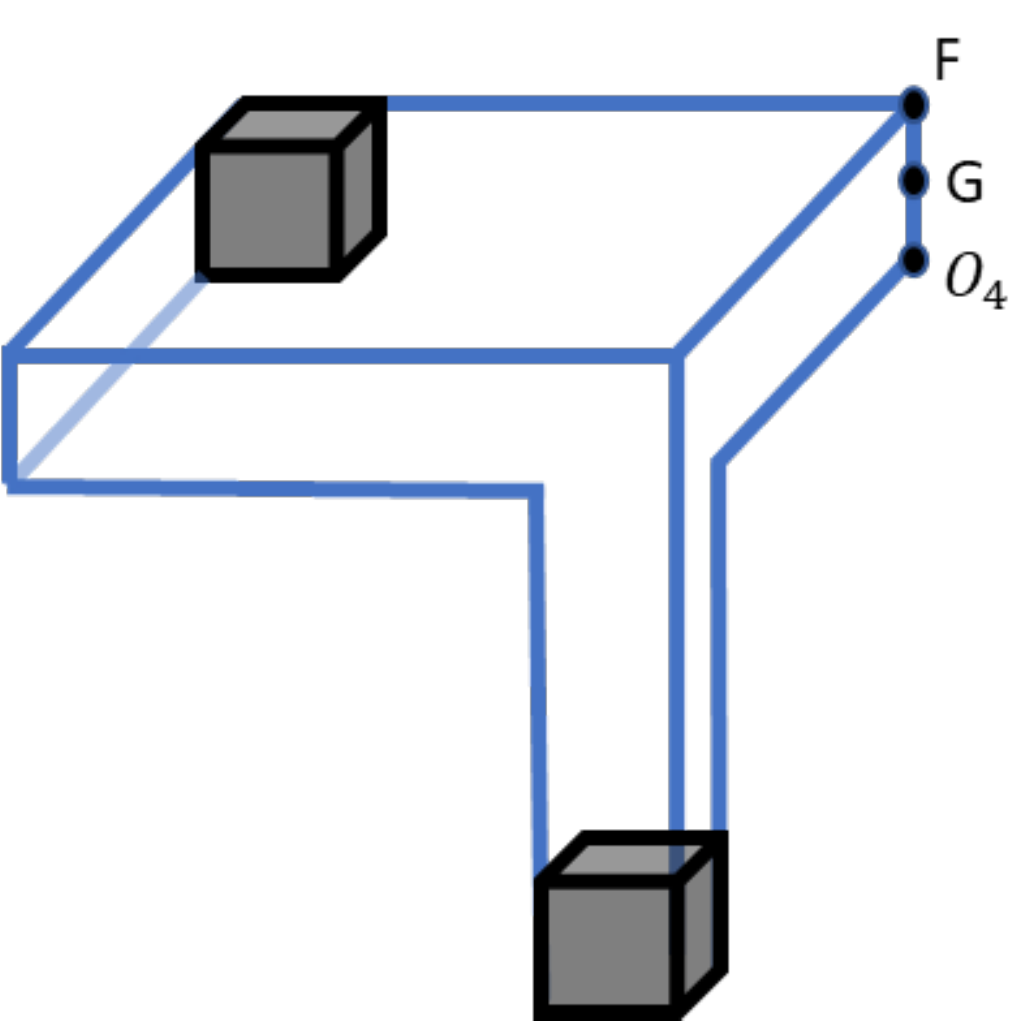}}
\sidesubfloat[]{\includegraphics[scale=0.28]{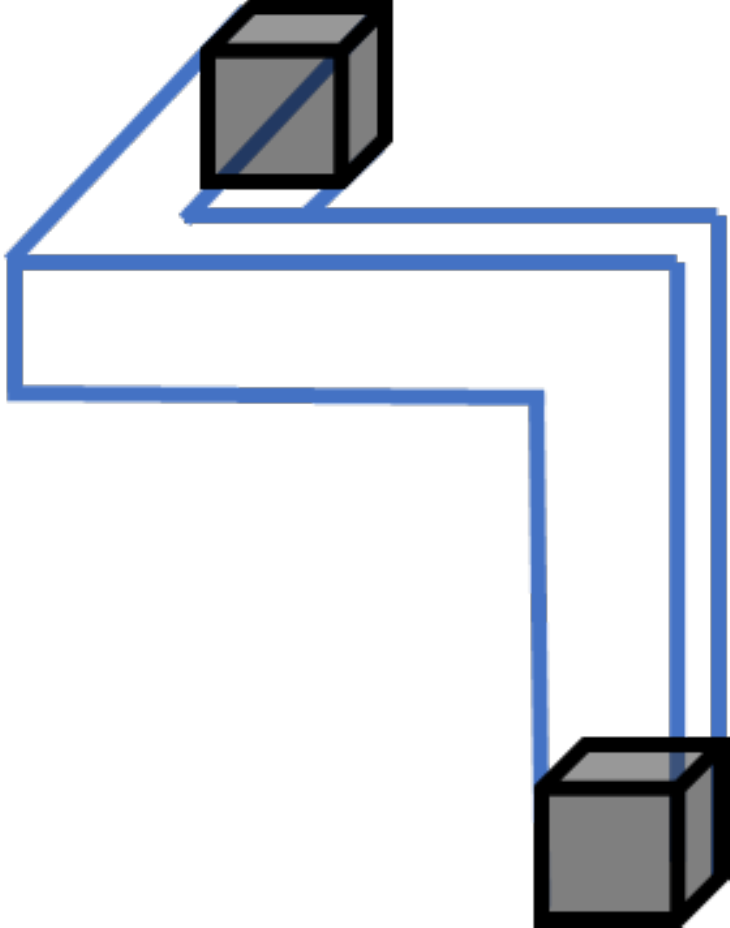}}
\caption{Deformation of 3D configuration of excitations of Fig.~\ref{excitation_config}~(a)}
\label{CC8_3D1}
\end{figure}
For example, in the first 3D configuration shown in Fig.~\ref{CC8_3D1}~(a), we clean the corner $O_1$ by multiplying with an $X$ stabilizer since $[O_1, ZI]=0$ implies $O_1=IX$ or $II$. From the cleaning in Fig.~\ref{CC8_p3}~(c), we know that the edge joining $O_2$, $B$ and $C$ in Fig.~\ref{CC8_3D1}~(b) is cleanable and hence the plane can be cleaned to reach the configuration in Fig.~\ref{CC8_3D1}~(c). Again, from the cleaning in Fig.~\ref{CC8_p2}~(a), we know that the edge joining $O_3$, $D$ and $E$ is cleanable and hence, subsequent cleaning of such edges gives Fig.~\ref{CC8_3D1}~(a). Using the cleaning done for Fig.~\ref{CC8_p1}~(c), we arrive at the flat-rod configuration shown. 

\begin{figure}[H]
\vspace{6mm}
\centering
\sidesubfloat[]{\includegraphics[scale=0.28]{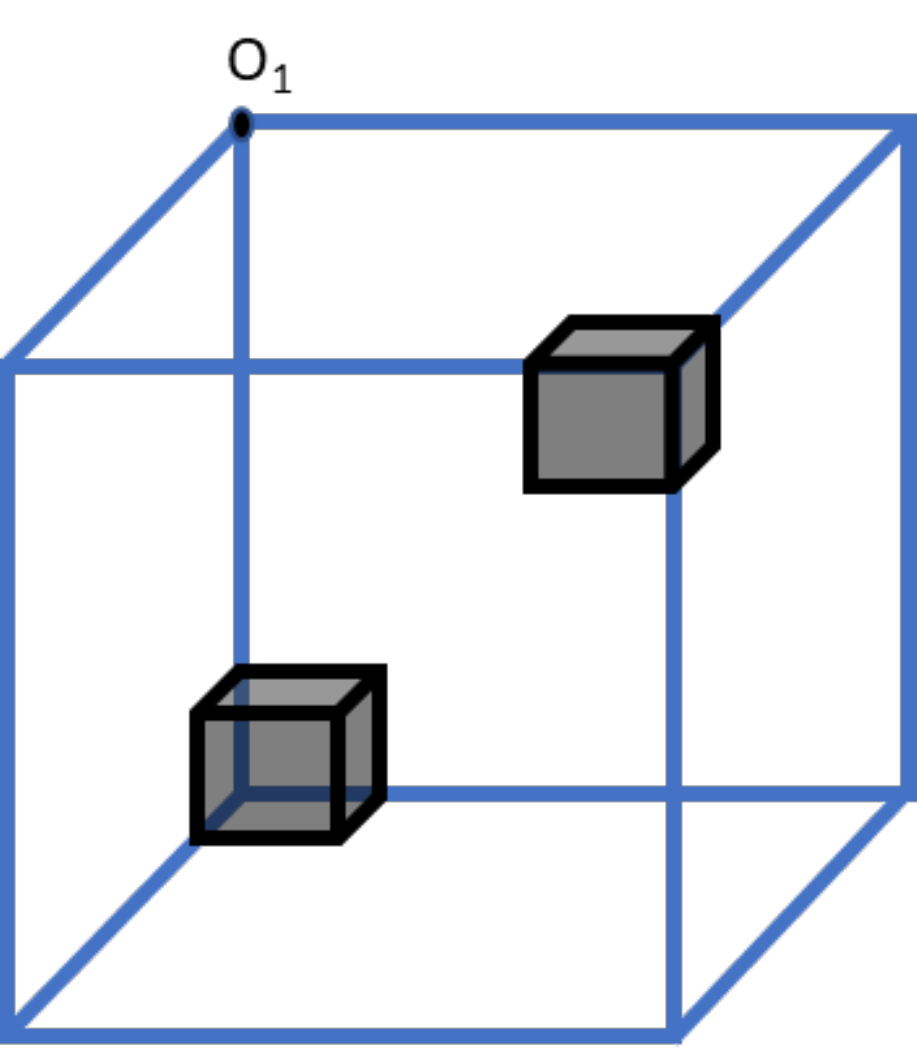}}
\sidesubfloat[]{\includegraphics[scale=0.28]{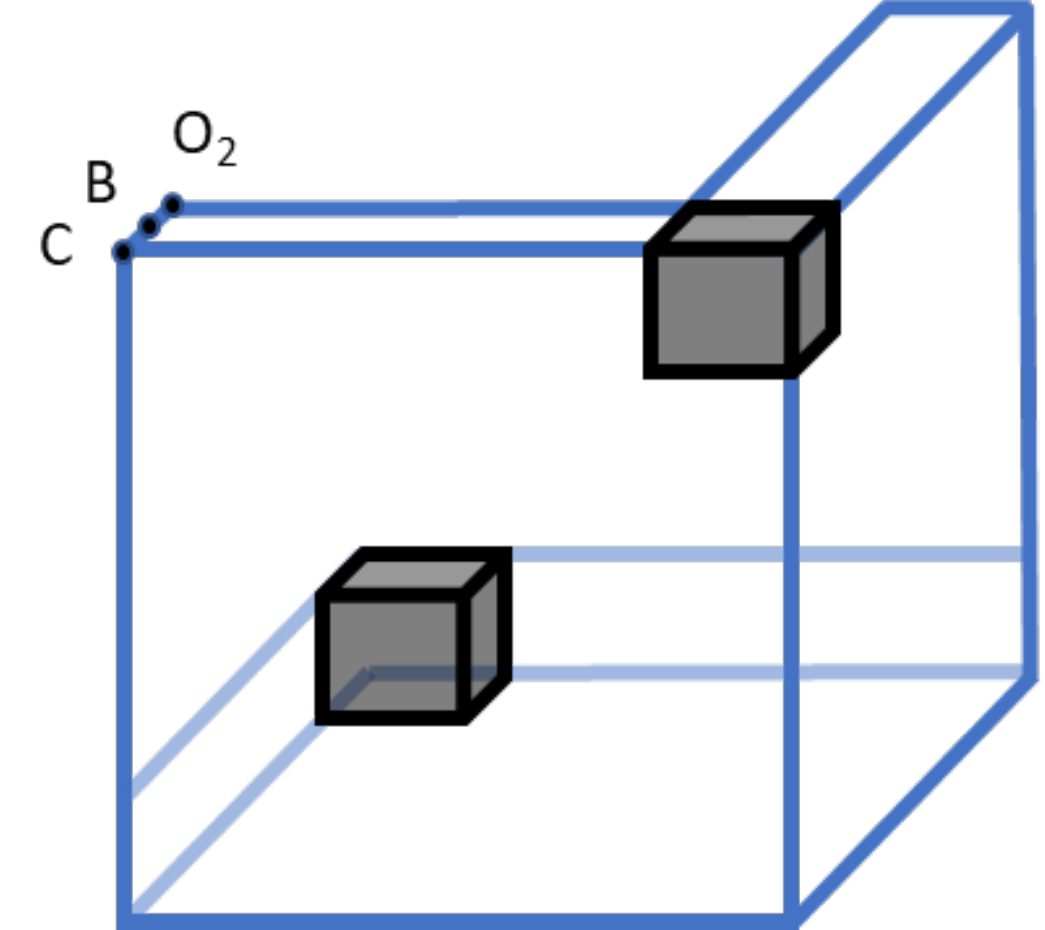}}
\sidesubfloat[]{\includegraphics[scale=0.28]{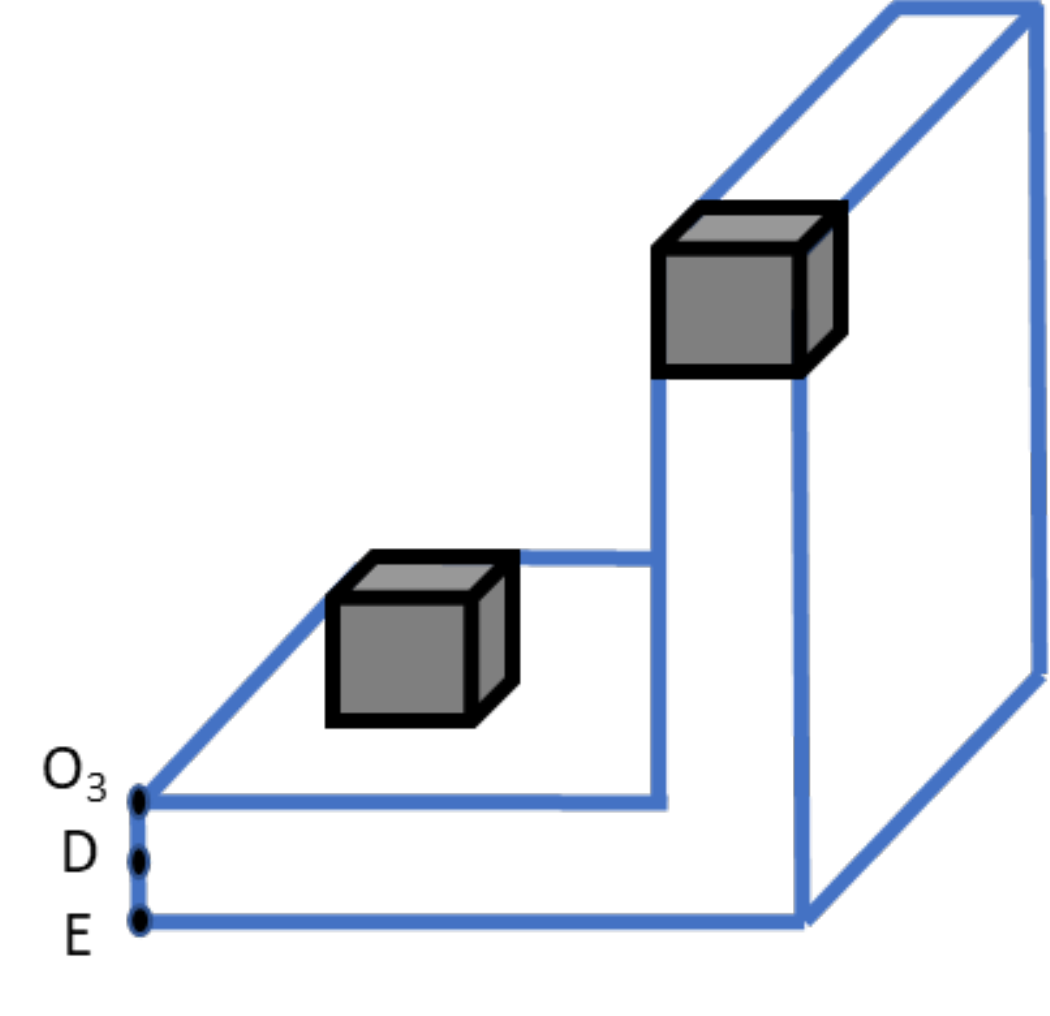}}
\sidesubfloat[]{\includegraphics[scale=0.28]{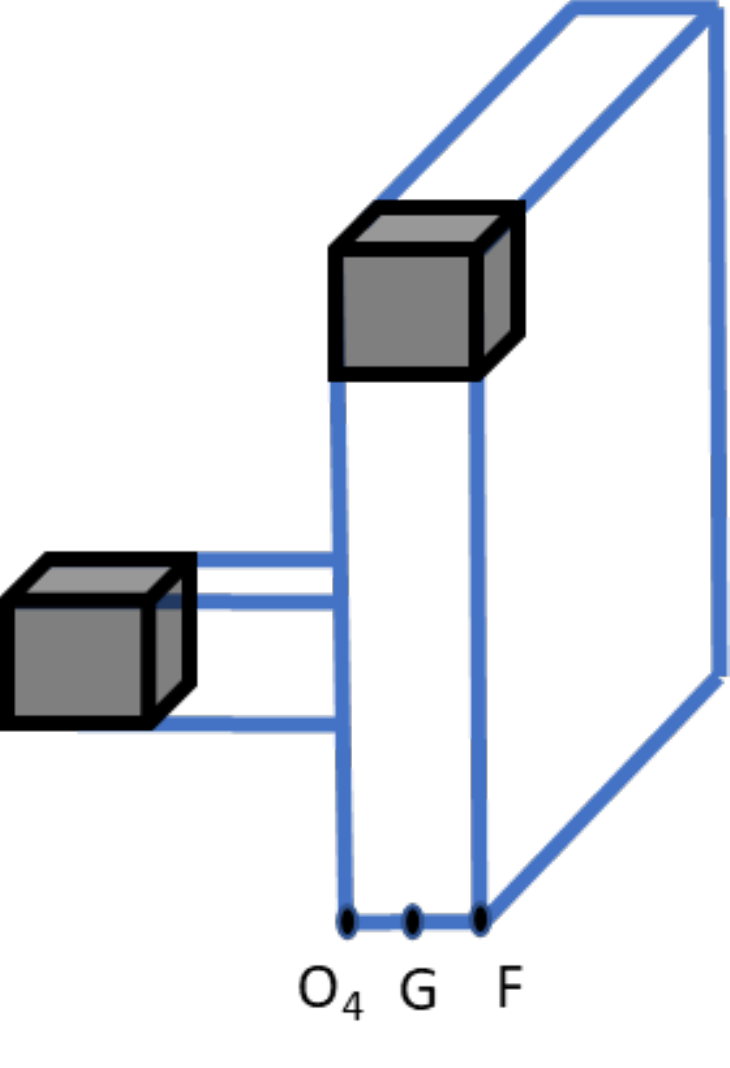}}
\sidesubfloat[]{\includegraphics[scale=0.28]{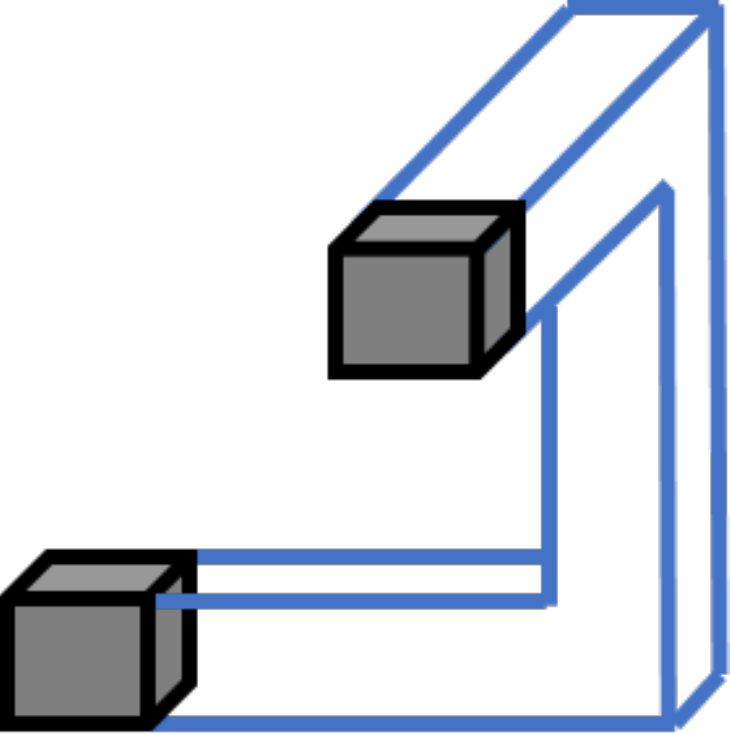}}
\caption{Deformation of the 3D configuration of excitations from Fig.~\ref{excitation_config}~(b)}
\label{CC8_3D2}
\end{figure}

Similar to the previous configuration of excitations, the configuration in Fig.~\ref{CC8_3D2}~(a) can be cleaned via the corners, $O_1$ and so on, to arrive at a configuration with three planar boxes. All the edges as marked by vertices on the three planar boxes are cleanable as shown before and hence, the configuration is easily cleaned to a flat-rod configuration. 

\begin{figure}[H]
\vspace{6mm}
\centering
\sidesubfloat[]{\includegraphics[scale=0.28]{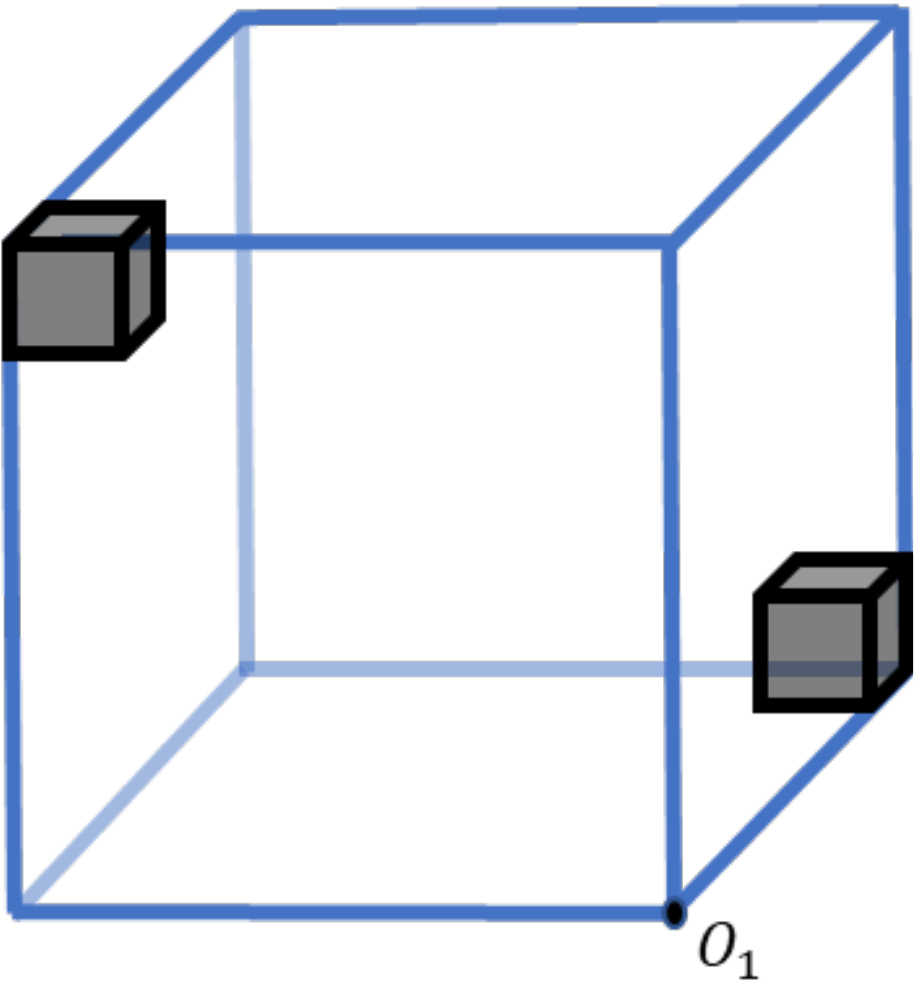}}
\sidesubfloat[]{\includegraphics[scale=0.28]{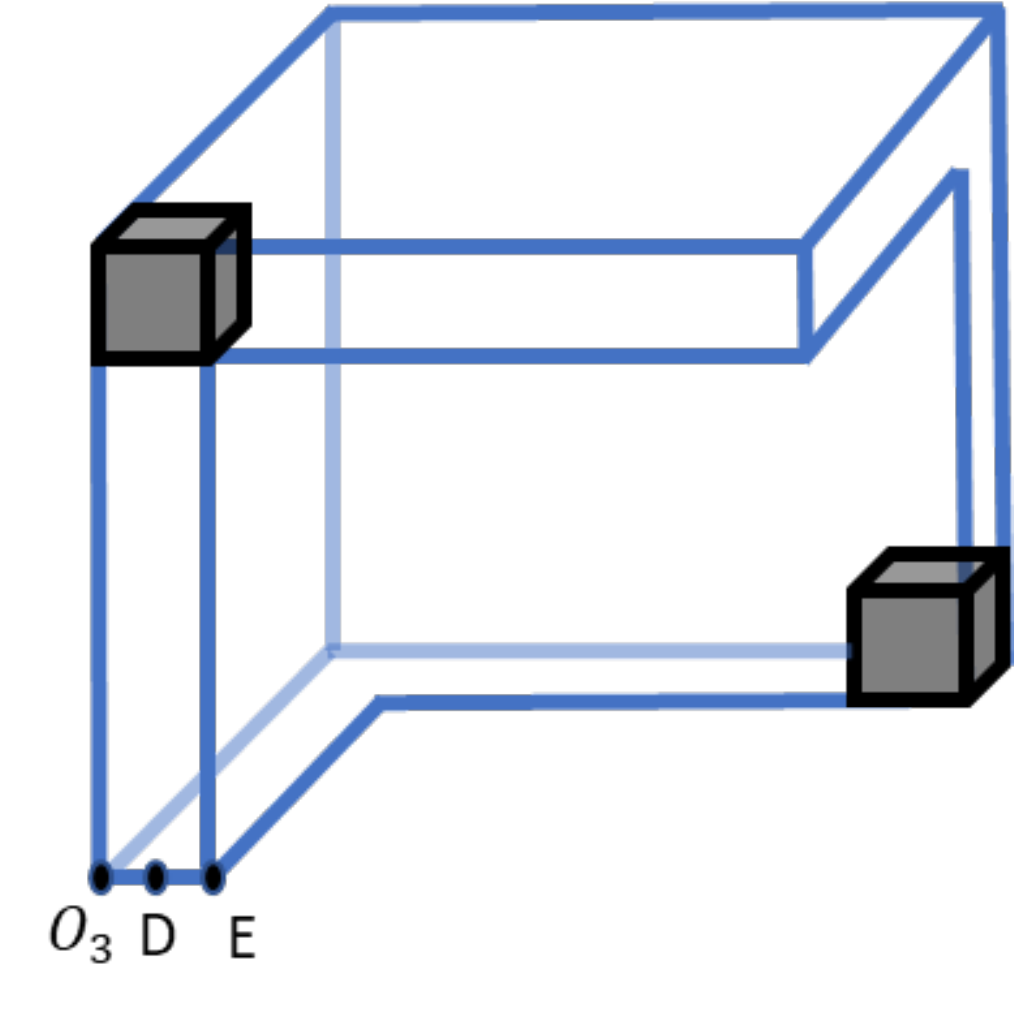}}
\sidesubfloat[]{\includegraphics[scale=0.28]{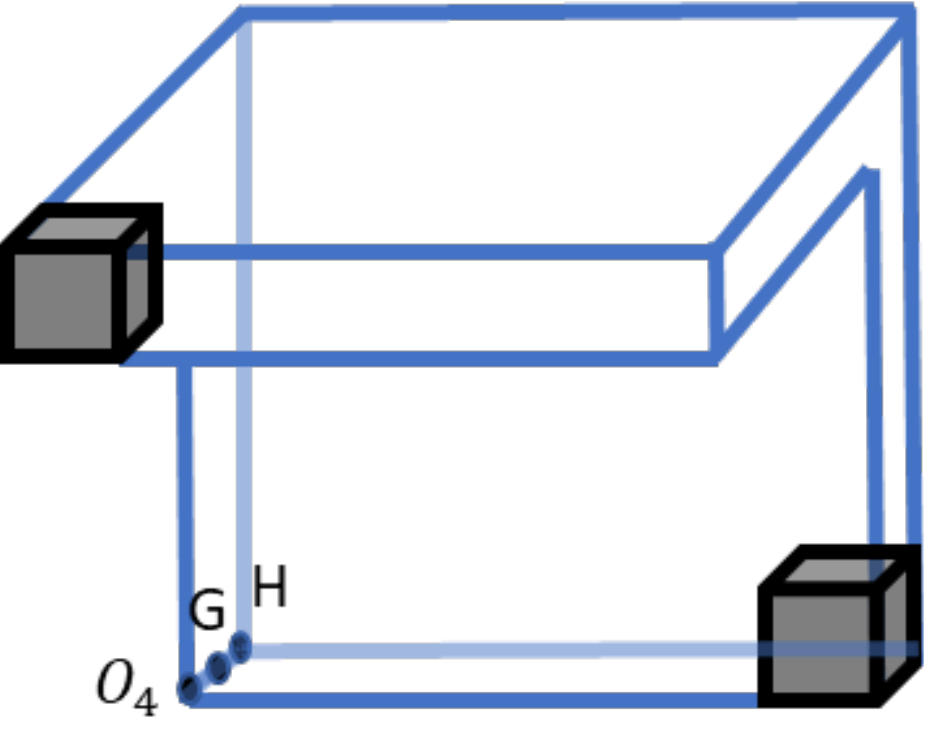}}
\sidesubfloat[]{\includegraphics[scale=0.28]{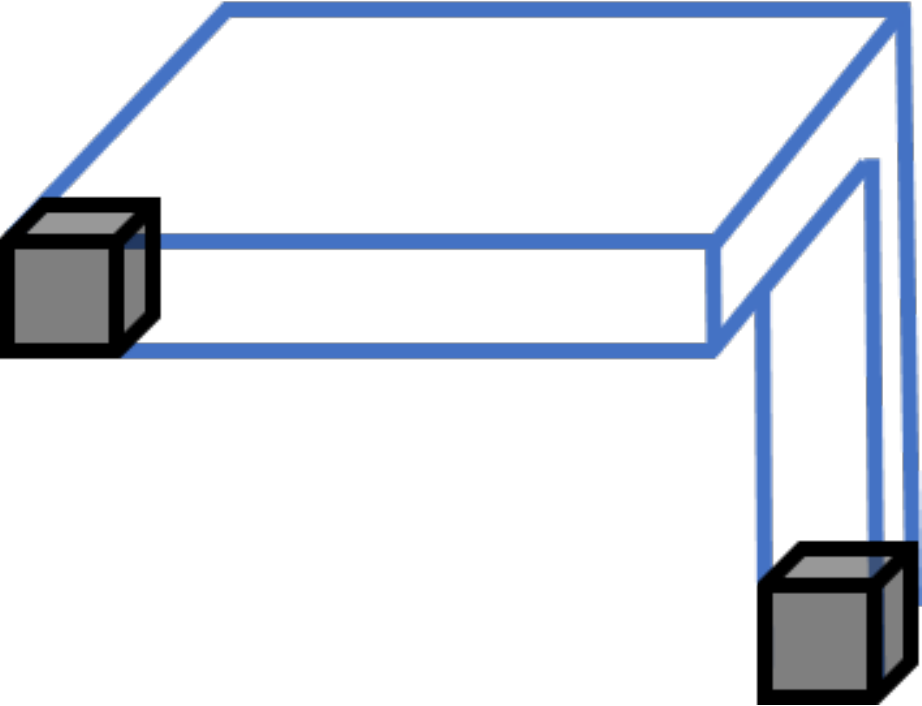}}
\sidesubfloat[]{\includegraphics[scale=0.28]{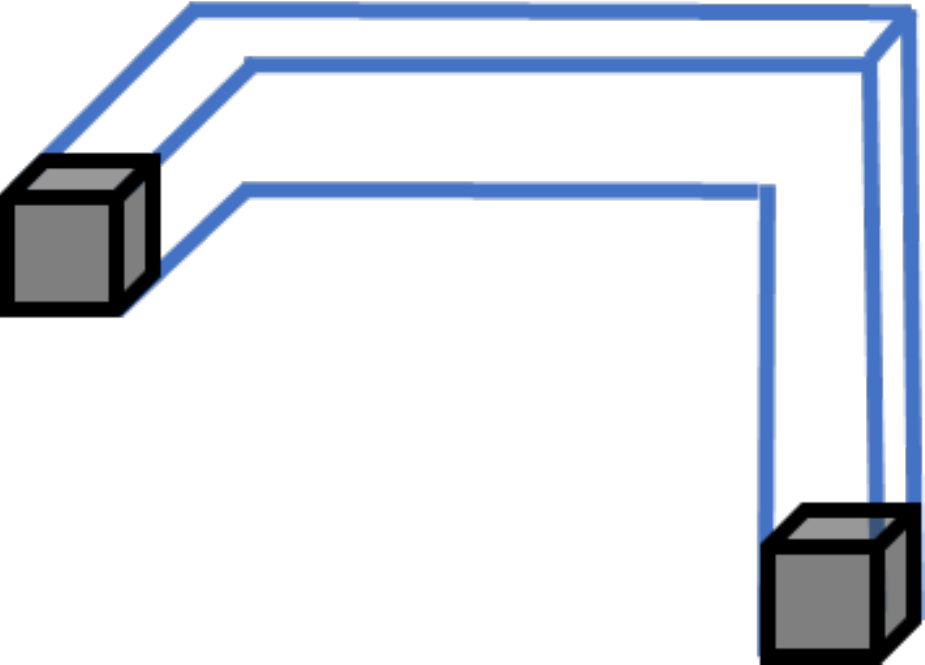}}
\caption{Deformation of the 3D configuration of excitations in Fig.~\ref{excitation_config}~(c)} 
\label{CC8_3D3}
\end{figure}
A similar process can be followed for the configuration in Fig.~\ref{CC8_3D3}~(a). Due to the cleanability of $O_1$, $O_3-D-E$ and $O_4-G-H$, we arrive at the configuration in Fig.~\ref{CC8_3D3}~(d). Now, we have an $xy$ planar box equivalent to the one for which we used higher order commutation constraints. Using the same cleaning procedure, we can go from Fig.~\ref{CC8_3D3}~(d) to (e). 

\begin{figure}[H]
\vspace{6mm}
\centering
\sidesubfloat[]{\includegraphics[scale=0.28]{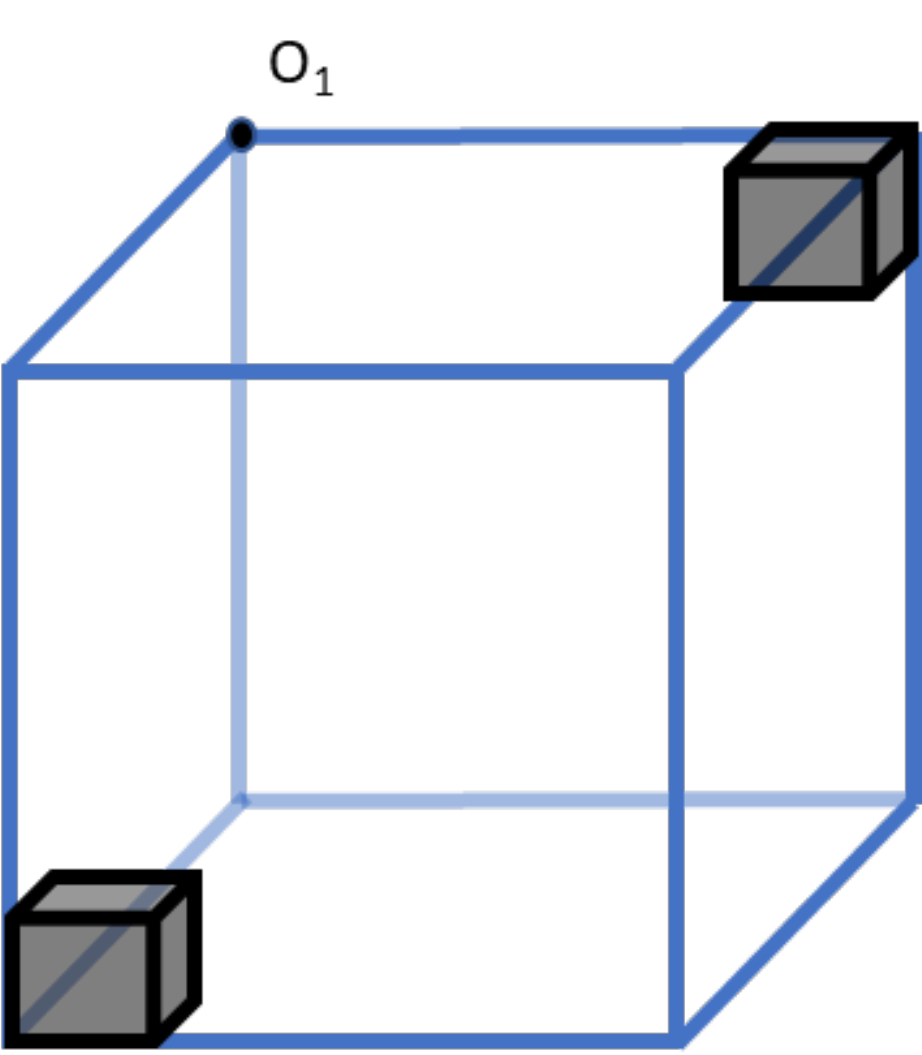}}
\sidesubfloat[]{\includegraphics[scale=0.28]{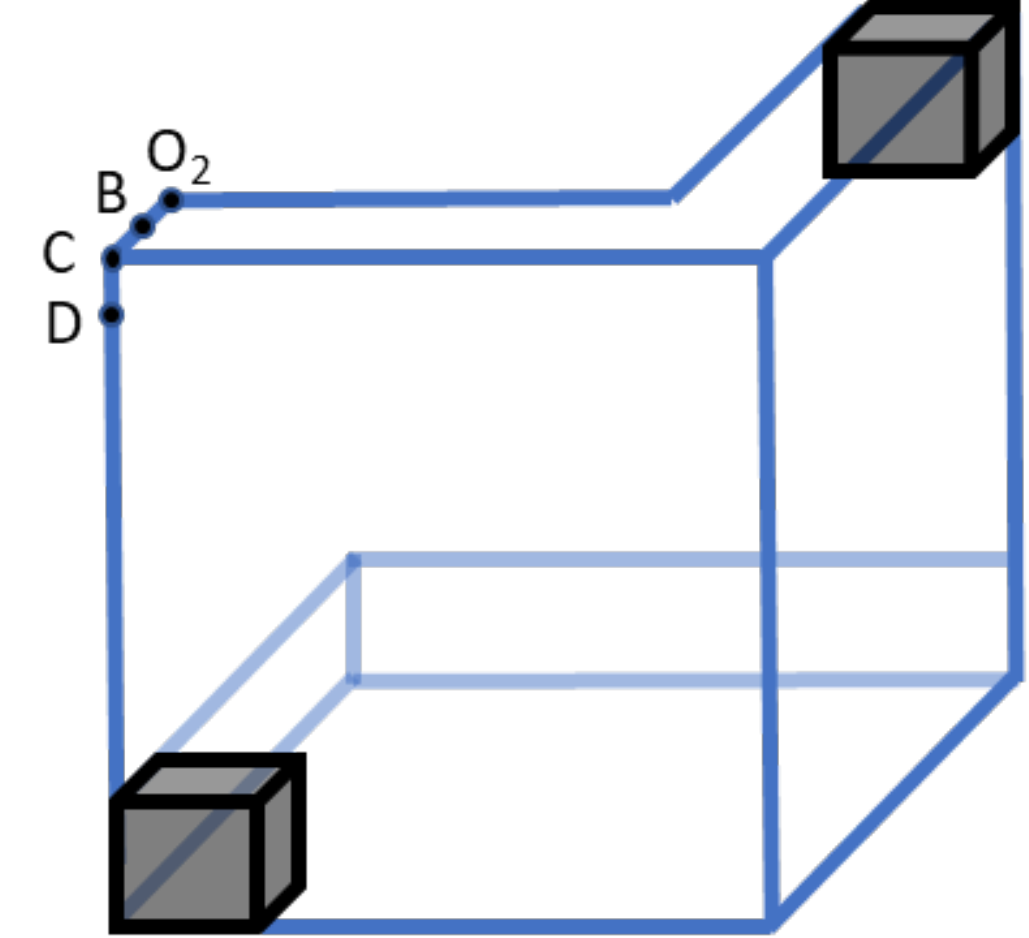}}
\sidesubfloat[]{\includegraphics[scale=0.28]{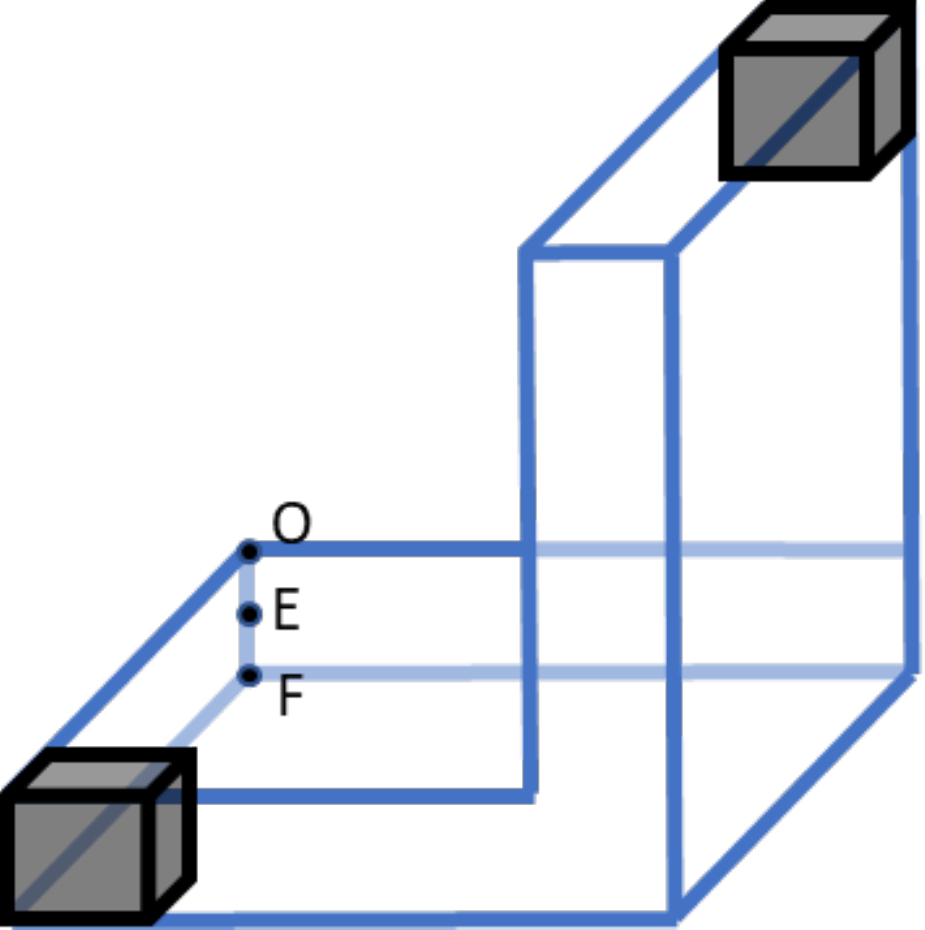}}
\sidesubfloat[]{\includegraphics[scale=0.28]{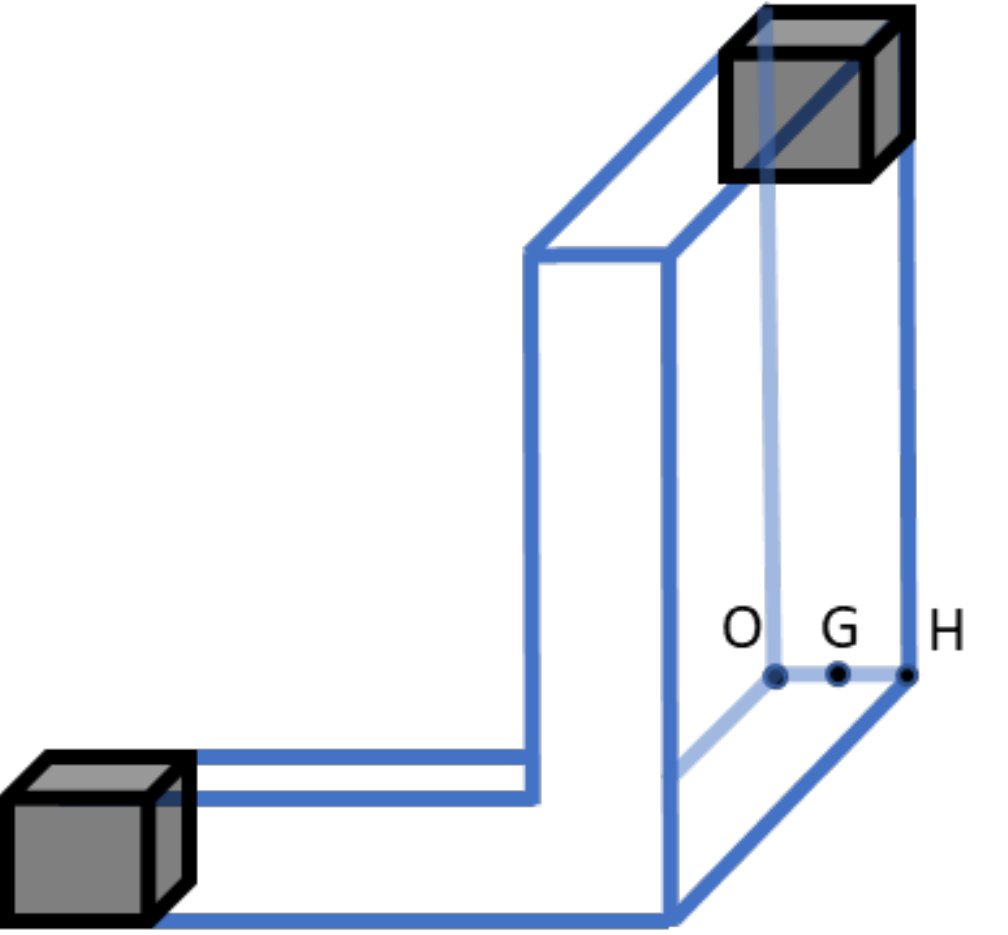}}
\sidesubfloat[]{\includegraphics[scale=0.28]{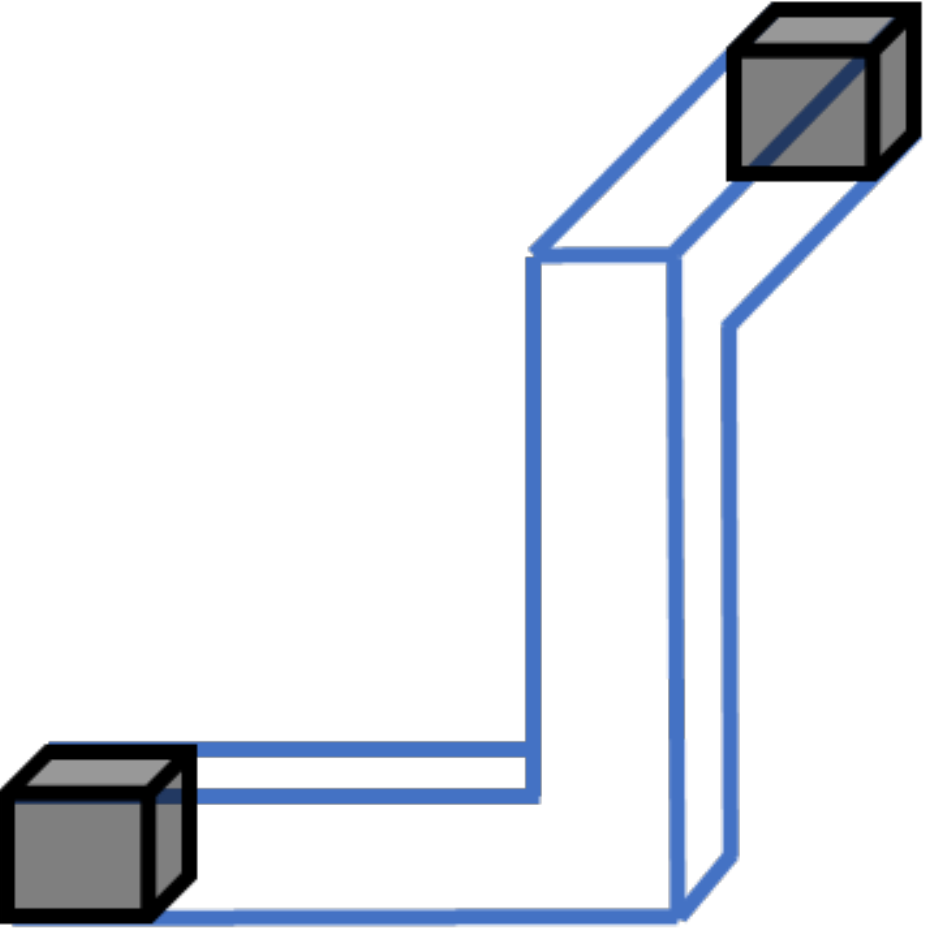}}
\caption{Deformation of 3D configuration of excitations of Fig.~\ref{excitation_config}d}
\label{CC8_3D4}
\end{figure}
We follow the same process for the configuration in Fig.~\ref{CC8_3D4}~(a) and find that the protruding edges of all planar boxes marked with vertices are cleanable. Hence, by the cleanability of planar boxes, we arrive at the final flat-rod configuration.

\section{Membrane-membrane operator commutation scaling}
\label{sec:memmemscaling}
In this appendix, we present tables that exhibit the scaling behavior of the commutation matrix rank for membrane-membrane configurations. 
Our labeling of the models is as follows: XC for the X-cube model, CB for the checkerboard model, Chm for Chamon's model, HH-I for Hsieh-Halasz's type-I model, HH-II for Hsieh-Halasz's second model which may be type-I or II, 2DTC$_{ij}$ for a stack of 2D toric codes parallel to the $i\,j$ plane, 2FXC for the twice foliated X-cube model, 3DTC for the 3D toric code, SFSL for the Sierpinski fractal spin liquid model and CCi for cubic code i. The final row in each table contains the result of the scaling behaviour, $\ell$ stands for linear, $c$ stands for constant, both up to fluctuating corrections of order constant. When the rank is precisely $0$ that is indicated instead. 
% $0$ and $f$ for fluctuating but order of constant scaling behavior. 
\begin{table}[H]
    \centering
\begin{tabular}{c|cccccccc|ccccccc}
 & \multicolumn{8}{c|}{Foliated type-I} & \multicolumn{7}{c}{Fractal type-I ($d=0,1,2$)}\tabularnewline
\hline 
$L$ & XC & CB & Chm & HH-I & 2DTC$_{xy}$ & 2DTC$_{yz}$ & 2DTC$_{xz}$ & 2FXC & CC11 & CC12 & CC13 & CC14 & CC15 & CC16 & CC17\tabularnewline
\hline 
14 & 30 & 30 & 18 & 60 & 0 & 0 & 30 & 30 & 4 & 26 & 28 & 26 & 10 & 15 & 4\tabularnewline
15 & 32 & 32 & 20 & 64 & 0 & 0 & 32 & 32 & 4 & 30 & 28 & 28 & 10 & 16 & 4\tabularnewline
16 & 34 & 34 & 20 & 68 & 0 & 0 & 34 & 34 & 4 & 32 & 32 & 30 & 10 & 13 & 4\tabularnewline
17 & 36 & 36 & 22 & 72 & 0 & 0 & 36 & 36 & 4 & 32 & 32 & 32 & 10 & 12 & 4\tabularnewline
18 & 38 & 38 & 22 & 76 & 0 & 0 & 38 & 38 & 4 & 36 & 36 & 34 & 10 & 10 & 4\tabularnewline
19 & 40 & 40 & 24 & 80 & 0 & 0 & 40 & 40 & 4 & 38 & 36 & 36 & 10 & 8 & 4\tabularnewline
\hline 
 & $\ell$ & $\ell$ & $\ell$ & $\ell$ & 0 & 0 & $\ell$ & $\ell$ & $c$ & $\ell$ & $\ell$ & $\ell$ & $c$ & $c$ & $c$\tabularnewline
\end{tabular}\\
    \caption{Scaling of the commutation matrix rank with $L=L_{x}=L_{y}=L_{z}$ and $w=3$ for the membrane-membrane configuration with membranes along the $xy$ and $yz$ planes that have open boundary conditions along their $\hat{y}$ oriented edges. This table contains foliated Type-I models and fractal type-I models with planons.}
    \label{tab:my_label}
\end{table}

\begin{table}[H]
    \centering
\begin{tabular}{c|c|ccccc|c|ccccccc}
 & TQFT & \multicolumn{5}{c|}{Type-II ($d=0$)} &  &  & \multicolumn{6}{c}{Fractal type-I ($d=0,1$)}\tabularnewline
\hline 
$L$ & 3DTC & CC1 & CC2 & CC3 & CC4 & CC7 & HH-II & SFSL & CC0 & CC5 & CC6 & CC8 & CC9 & CC10\tabularnewline
\hline 
14 & 0 & 14 & 14 & 14 & 8 & 15 & 16 & 0 & 16 & 10 & 10 & 8 & 12 & 15\tabularnewline
15 & 0 & 16 & 12 & 16 & 8 & 16 & 12 & 0 & 12 & 10 & 10 & 8 & 16 & 16\tabularnewline
16 & 0 & 12 & 12 & 15 & 8 & 12 & 14 & 0 & 15 & 10 & 10 & 8 & 12 & 15\tabularnewline
17 & 0 & 16 & 14 & 14 & 8 & 14 & 16 & 0 & 15 & 10 & 10 & 8 & 16 & 14\tabularnewline
18 & 0 & 12 & 14 & 15 & 8 & 16 & 14 & 0 & 15 & 10 & 10 & 8 & 12 & 14\tabularnewline
19 & 0 & 12 & 16 & 14 & 8 & 16 &   & 0 &  & 10 & 10 & 8 & 12 & 12\tabularnewline
\hline 
 & 0 & $c$ & $c$ & $c$ & $c$ & $c$ & $0$ & $c$ & $c$ & $c$ & $c$ & $c$ & $c$ & $c$\tabularnewline
\end{tabular}\\
    \caption{Scaling of the commutation matrix rank with $L=L_{x}=L_{y}=L_{z}$ and $w=3$ for the membrane-membrane configuration with membranes along the $xy$ and $yz$ planes that have open boundary conditions along their $\hat{y}$ oriented edges. This table contains TQFT, Type-II and fractal Type-I models with no planons.}
    \label{tab:my_label}
\end{table}

\begin{table}[H]
    \centering
   \begin{tabular}{c|cccccccc|ccccccc}
 & \multicolumn{8}{c|}{Foliated type-I} & \multicolumn{7}{c}{Fractal type-I ($d=0,1,2$)}\tabularnewline
\hline 
$L$ & XC & CB & Chm & HH-I & 2DTC$_{xy}$ & 2DTC$_{yz}$ & 2DTC$_{xz}$ & 2FXC & CC11 & CC12 & CC13 & CC14 & CC15 & CC16 & CC17\tabularnewline
\hline 
14 & 30 & 30 & 18 & 60 & 30 & 0 & 0 & 0 & 4 & 10 & 10 & 8 & 10 & 15 & 52\tabularnewline
15 & 32 & 32 & 20 & 64 & 32 & 0 & 0 & 0 & 4 & 10 & 10 & 8 & 10 & 16 & 60\tabularnewline
16 & 34 & 34 & 20 & 68 & 34 & 0 & 0 & 0 & 4 & 10 & 10 & 8 & 10 & 13 & 64\tabularnewline
17 & 36 & 36 & 22 & 72 & 36 & 0 & 0 & 0 & 4 & 10 & 10 & 8 & 10 & 12 & 64\tabularnewline
18 & 38 & 38 & 22 & 76 & 38 & 0 & 0 & 0 & 4 & 10 & 10 & 8 & 10 & 10 & 72\tabularnewline
19 & 40 & 40 & 24 & 80 & 40 & 0 & 0 & 0 & 4 & 10 & 10 & 8 & 10 & 8 & 76\tabularnewline
\hline 
 & $\ell$ & $\ell$ & $\ell$ & $\ell$ & $\ell$ & 0 & 0 & 0 & $c$ & $c$ & $c$ & $c$ & $c$ & $c$ & $\ell$\tabularnewline
\end{tabular}\\
    \caption{Scaling of the commutation matrix rank with $L=L_{x}=L_{y}=L_{z}$ and $w=3$ for the membrane-membrane configuration with membranes along the $xz$ and $yz$ planes that have open boundary conditions along their $\hat{z}$ oriented edges. This table contains foliated Type-I models and fractal type-I models with planons.}
\label{}
\end{table}

\begin{table}[H]
    \centering
\begin{tabular}{c|c|ccccc|c|ccccccc}
 & TQFT & \multicolumn{5}{c|}{Type-II ($d=0$)} &  &  & \multicolumn{6}{c}{Fractal type-I ($d=0,1$)}\tabularnewline
\hline 
$L$ & 3DTC & CC1 & CC2 & CC3 & CC4 & CC7 & HH-II & SFSL & CC0 & CC5 & CC6 & CC8 & CC9 & CC10\tabularnewline
\hline 
14 & 0 & 14 & 14 & 16 & 8 & 7 & 16 & 0 & 16 & 0 & 0 & 15 & 12 & 10\tabularnewline
15 & 0 & 16 & 12 & 12 & 8 & 4 & 12 & 0 & 12 & 0 & 0 & 16 & 16 & 14\tabularnewline
16 & 0 & 12 & 12 & 16 & 8 & 4 & 14 & 0 & 15 & 0 & 0 & 14 & 12 & 15\tabularnewline
17 & 0 & 16 & 14 & 16 & 8 & 4 & 16 & 0 & 15 & 0 & 0 & 16 & 16 & 16\tabularnewline
18 & 0 & 12 & 14 & 16 & 8 & 6 & 14 & 0 & 15 & 0 & 0 & 14 & 12 & 14\tabularnewline
19 & 0 & 12 & 16 & 12 & 8 & 8 &   & 0 &  & 0 & 0 & 16 & 12 & 14\tabularnewline
\hline 
 & 0 & $c$ & $c$ & $c$ & $c$ & $c$ & 0 & $c$ & 0 & 0 & $c$ & $c$ & $c$ & $c$\tabularnewline
\end{tabular}\\
    \caption{Scaling of the commutation matrix rank with $L=L_{x}=L_{y}=L_{z}$ and $w=3$ for the membrane-membrane configuration with membranes along the $xz$ and $yz$ planes that have open boundary conditions along their $\hat{z}$ oriented edges. This table contains TQFT, Type-II and fractal Type-I models with no planons.}
\label{}
\end{table}

\begin{table}[H]
    \centering
\begin{tabular}{c|cccccccc|ccccccc}
 & \multicolumn{8}{c|}{Foliated type-I} & \multicolumn{7}{c}{Fractal type-I ($d=0,1,2$)}\tabularnewline
\hline 
$L$ & XC & CB & Chm & HH-I & 2DTC$_{xy}$ & 2DTC$_{yz}$ & 2DTC$_{xz}$ & 2FXC & CC11 & CC12 & CC13 & CC14 & CC15 & CC16 & CC17\tabularnewline
\hline 
14 & 30 & 30 & 19 & 60 & 0 & 30 & 0 & 30 & 26 & 8 & 0 & 10 & 26 & 15 & 4\tabularnewline
15 & 32 & 32 & 20 & 64 & 0 & 32 & 0 & 32 & 30 & 8 & 0 & 10 & 30 & 16 & 4\tabularnewline
16 & 34 & 34 & 21 & 68 & 0 & 34 & 0 & 34 & 32 & 8 & 0 & 10 & 32 & 13 & 4\tabularnewline
17 & 36 & 36 & 22 & 72 & 0 & 36 & 0 & 36 & 32 & 8 & 0 & 10 & 32 & 12 & 4\tabularnewline
18 & 38 & 38 & 23 & 76 & 0 & 38 & 0 & 38 & 36 & 8 & 0 & 10 & 36 & 10 & 4\tabularnewline
19 & 40 & 40 & 24 & 80 & 0 & 40 & 0 & 40 & 38 & 8 & 0 & 10 & 38 & 8 & 4\tabularnewline
\hline 
 & $\ell$ & $\ell$ & $\ell$ & $\ell$ & 0 & 0 & $\ell$ & $\ell$ & $\ell$ & $c$ & 0 & $c$ & $l$ & $c$ & $c$\tabularnewline
\end{tabular}\\
   \caption{Scaling of the commutation matrix rank with $L=L_{x}=L_{y}=L_{z}$ and $w=3$ for the membrane-membrane configuration with membranes along the $xy$ and $xz$ planes that have open boundary conditions along their $\hat{x}$ oriented edges. This table contains foliated Type-I models and fractal type-I models with planons.}
    \label{}
\end{table}

\begin{table}[H]
    \centering
\begin{tabular}{c|c|ccccc|c|ccccccc}
 & TQFT & \multicolumn{5}{c|}{Type-II ($d=0$)} &  &  & \multicolumn{6}{c}{Fractal type-I ($d=0,1$)}\tabularnewline
\hline 
$L$ & 3DTC & CC1 & CC2 & CC3 & CC4 & CC7 & HH-II & SFSL & CC0 & CC5 & CC6 & CC8 & CC9 & CC10\tabularnewline
\hline 
14 & 0 & 14 & 14 & 15 & 12 & 8 & 16 & 2 & 15 & 10 & 10 & 15 & 14 & 15\tabularnewline
15 & 0 & 16 & 16 & 16 & 12 & 8 & 12 & 2 & 12 & 10 & 10 & 16 & 16 & 16\tabularnewline
16 & 0 & 16 & 16 & 14 & 12 & 8 & 14 & 2 & 15 & 10 & 10 & 16 & 16 & 15\tabularnewline
17 & 0 & 16 & 16 & 14 & 14 & 8 & 16 & 2 & 15 & 10 & 10 & 16 & 16 & 14\tabularnewline
18 & 0 & 14 & 14 & 15 & 16 & 8 & 14 & 2 & 15 & 10 & 10 & 16 & 14 & 14\tabularnewline
19 & 0 & 12 & 12 & 14 & 14 & 8 & 12 & 2 & 12 & 10 & 10 & 14 & 12 & 12\tabularnewline
\hline 
 & 0 & $c$ & $c$ & $c$ & $c$ & $c$ & $c$ & $c$ & $c$ & $c$ & $c$ & $c$ & $c$ & $c$\tabularnewline
\end{tabular}
    \caption{Scaling of the commutation matrix rank with $L=L_{x}=L_{y}=L_{z}$ and $w=3$ for the membrane-membrane configuration with membranes along the $xy$ and $xz$ planes that have open boundary conditions along their $\hat{x}$ oriented edges. This table contains TQFT, Type-II and fractal Type-I models with no planons.}
    \label{tab:my_label}
\end{table}

\section{Topological stabilizer model zoo}
\label{zoo}
% \begin{align}
% \begin{array}{c}
% \drawgenerator{0}{1}{2}{3}{4}{5}{6}{7}
% \quad
% \drawgenerator{xz}{xyz}{x}{xy}{y}{1}{yz}{z}
% \end{array}
% \end{align}

In this appendix we collect the topological stabilizer models known to us, sorted according to the type of topological order they support. 

\subsection{ Type-II models}
  \label{Type2zoo}
In this section we collect known examples of type-II models: cubic codes 1, 2, 3, 4, 7, 8, 10, from Ref.~\onlinecite{haah2011local},  Yoshida's type-II qubit and qutrit fractal spin liquids~\cite{yoshida2013exotic}, Kim's type-II  qutrit and qudit models~\cite{kim20123d}. 
Cubic codes 1 to 4 were shown to be type-II codes in Ref.~\onlinecite{haah2011local} and our results indicate cubic codes 7, 8 and 10 are also type-II. 
We also include Halasz and Hsieh's ``type-II'' model~\cite{hsieh_halasz_partons}, although it has not been shown to be type-II. Our results are inconclusive for this model as they are consistent with fractal type-I or type-II.

\subsubsection{Cubic code 1}
The stabilizer generators of cubic code 1 are given by
\begin{align}
\begin{array}{c}
\drawgenerator{XI}{II}{IX}{XI}{IX}{XX}{XI}{IX}
\quad
\drawgenerator{ZI}{ZZ}{IZ}{ZI}{IZ}{II}{ZI}{IZ}
\end{array}
\, .
\end{align}

\begin{table}[H]
\renewcommand{\arraystretch}{1.4}{
\begin{tabular}{c|cccccc c}
Model & Deformability & $n_{\text{rods}}^{3D}$ & $n_{\text{rods}}^{2D}\left(zx,xy,yz\right)$ &
$n_{\text{rods}}^{1D}\left(x,y,z\right)$ & $n_{m}\left(zx,xy,yz\right)$ & $d$ & Type\tabularnewline
\hline
CC1 & $\checkmark$ & 0 & $\left(0,0,0\right)$ & $\left(0,0,0\right)$ & $\left(c,c,c\right)$ & 0 & type-II\tabularnewline
\end{tabular}}
\caption{Operator data for 3D stabilizer models. For definitions, see table~\ref{table_invariants}}
\label{CC1_data}
\end{table}

\subsubsection{Cubic code 1B}
It was shown in Ref.~\onlinecite{haah2014bifurcation} that real-space entanglement renormalization of cubic code 1 yields cubic code 1, cubic code 1B and disentangled qubits in the trivial state. 
It was further shown that cubic code 1B bifurcates into two copies of itself under entanglement renormalization.  The stabilizer generators of cubic code 1B are given by
\begin{align}
\begin{array}{c}
\drawgenerator{}{}{XXXI}{}{IIXX}{IXIX}{}{XIII}
\quad
\drawgenerator{}{}{XIIX}{}{IIXI}{XXXX}{}{IXII}
\quad
\drawgenerator{ZZII}{IZIZ}{}{IIZI}{}{}{ZIZZ}{}
\quad
\drawgenerator{ZIII}{ZZZZ}{}{IIIZ}{}{}{IZZI}{}
\end{array}\, .
\end{align}
Cubic code 1B is type-II since any nontrivial string operators in the model would imply string operators in cubic code 1, which was shown to contain no string operators~\cite{haah2011local}.

\subsubsection{Cubic code 2}
The stabilizer generators of cubic code 2 are given by
\begin{align}
\begin{array}{c}
\drawgenerator{XI}{XI}{IX}{XI}{XX}{XI}{IX}{XX}
\quad
\drawgenerator{ZZ}{IZ}{ZI}{ZZ}{IZ}{IZ}{ZI}{IZ}
\end{array}
\, .
\end{align}

\begin{table}[H]
\renewcommand{\arraystretch}{1.4}{
\begin{tabular}{c|cccccc c}
Model & Deformability & $n_{\text{rods}}^{3D}$ & $n_{\text{rods}}^{2D}\left(zx,xy,yz\right)$ &
$n_{\text{rods}}^{1D}\left(x,y,z\right)$ & $n_{m}\left(zx,xy,yz\right)$ & $d$ & Type\tabularnewline
\hline
CC2 & $\checkmark$ & 0 & $\left(0,0,0\right)$ & $\left(0,0,0\right)$ & $\left(c,c,c\right)$ & 0 & type-II\tabularnewline
\end{tabular}}
\caption{Operator data. For definitions, see table~\ref{table_invariants}}
\label{CC2_data}
\end{table}

\subsubsection{Cubic code 3}
The stabilizer generators of cubic code 3 are given by
\begin{align}
\begin{array}{c}
\drawgenerator{XI}{XI}{IX}{II}{IX}{XX}{XI}{IX}
\quad
\drawgenerator{ZI}{ZZ}{IZ}{ZI}{IZ}{IZ}{ZI}{II}
\end{array}
\, .
\end{align}

\begin{table}[H]
\renewcommand{\arraystretch}{1.4}{
\begin{tabular}{c|cccccc c}
Model & Deformability & $n_{\text{rods}}^{3D}$ & $n_{\text{rods}}^{2D}\left(zx,xy,yz\right)$ &
$n_{\text{rods}}^{1D}\left(x,y,z\right)$ & $n_{m}\left(zx,xy,yz\right)$ & $d$ & Type\tabularnewline
\hline
CC3 & $\checkmark$ & 0 & $\left(0,0,0\right)$ & $\left(0,0,0\right)$ & $\left(c,c,c\right)$ & 0 & type-II\tabularnewline
\end{tabular}}
\caption{Operator data. For definitions, see table~\ref{table_invariants}}
\label{CC1_data}
\end{table}

\subsubsection{Cubic code 4}
The stabilizer generators of cubic code 4 are given by
\begin{align}
\begin{array}{c}
\drawgenerator{XI}{II}{IX}{XI}{XI}{XX}{IX}{IX}
\quad
\drawgenerator{IZ}{ZZ}{ZI}{ZI}{IZ}{II}{ZI}{IZ}
\end{array}
\, .
\end{align} 

\begin{table}[H]
\renewcommand{\arraystretch}{1.4}{
\begin{tabular}{c|cccccc c}
Model & Deformability & $n_{\text{rods}}^{3D}$ & $n_{\text{rods}}^{2D}\left(zx,xy,yz\right)$ &
$n_{\text{rods}}^{1D}\left(x,y,z\right)$ & $n_{m}\left(zx,xy,yz\right)$ & $d$ & Type\tabularnewline
\hline
CC4 & $\checkmark$ & 0 & $\left(0,0,0\right)$ & $\left(0,0,0\right)$ & $\left(c,c,c\right)$ & 0 & type-II\tabularnewline
\end{tabular}}
\caption{Operator data. For definitions, see table~\ref{table_invariants}}
\label{CC1_data}
\end{table}

\subsubsection{Cubic code 7}
The stabilizer generators of cubic code 7 are given by
\begin{align}
\begin{array}{c}
\drawgenerator{II}{XI}{IX}{II}{IX}{XX}{XI}{XX}
\quad
\drawgenerator{ZI}{ZZ}{IZ}{ZZ}{II}{IZ}{ZI}{II}
\end{array}
\, .
\end{align}
Our results indicate that this model is type-II. 

\begin{figure}[h!]
\vspace{6mm}
\centering
\sidesubfloat[]{\includegraphics[scale=0.28]{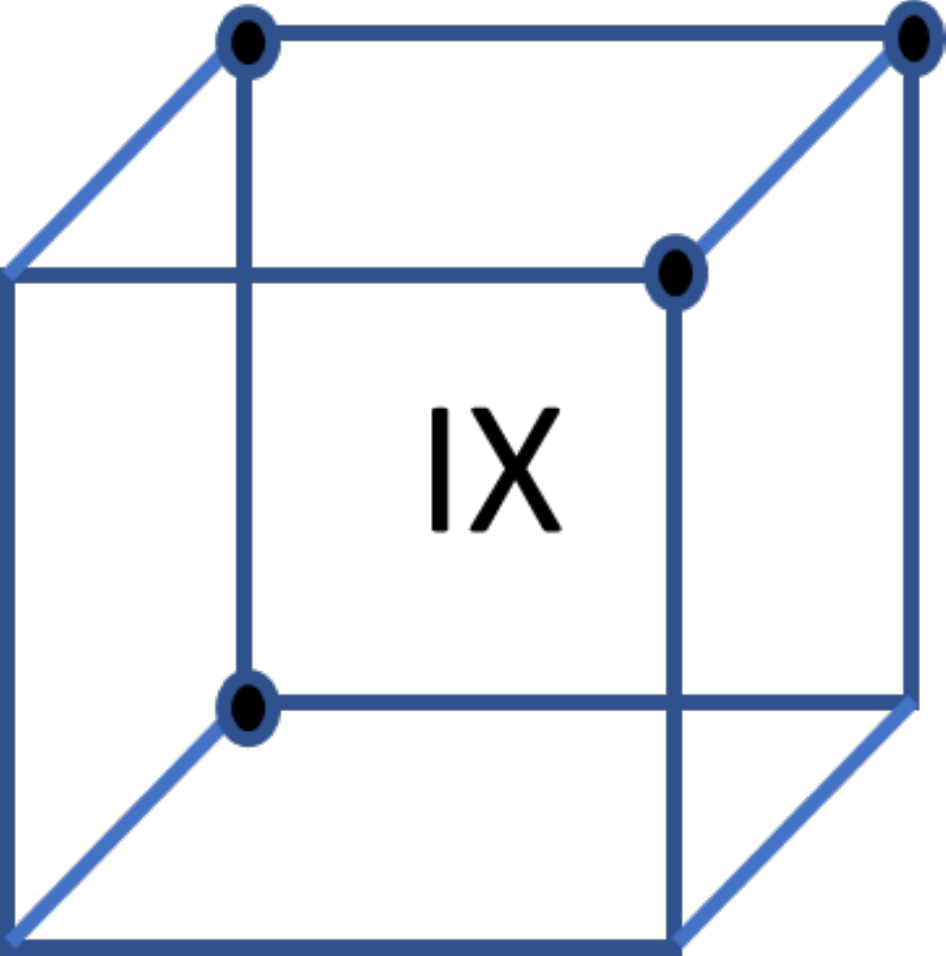}}
\sidesubfloat[]{\includegraphics[scale=0.28]{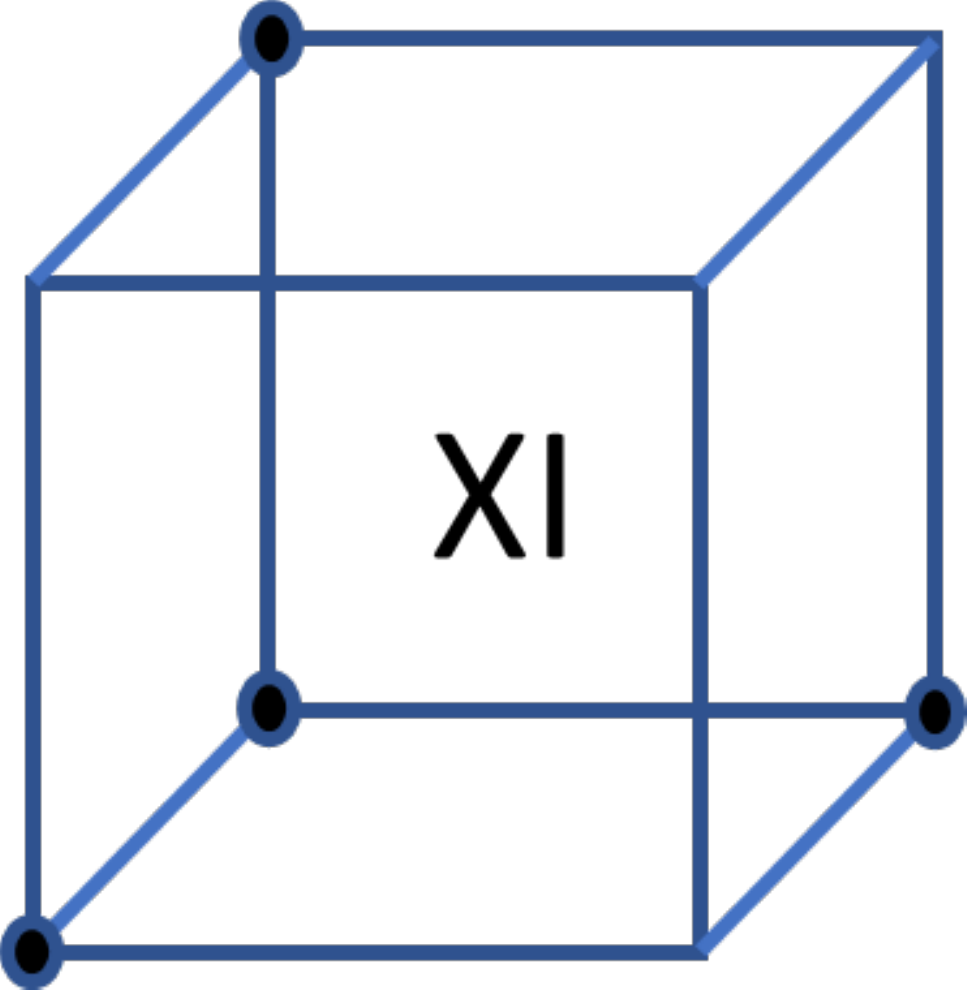}}
% \sidesubfloat[]{\includegraphics[scale=0.28]{CC7_X_string.pdf}}
\caption{Excitation patterns for CC7. }
\label{CC7b}
\end{figure}

\begin{table}[H]
\renewcommand{\arraystretch}{1.4}{
\begin{tabular}{c|cccccc c}
Model & Deformability & $n_{\text{rods}}^{3D}$ & $n_{\text{rods}}^{2D}\left(zx,xy,yz\right)$ &
$n_{\text{rods}}^{1D}\left(x,y,z\right)$ & $n_{m}\left(zx,xy,yz\right)$ & $d$ & Type\tabularnewline
\hline
CC7 & $\checkmark$ & 0 & $\left(0,0,0\right)$ & $\left(0,0,0\right)$ & $\left(c,c,c\right)$ & 0 &  type-II\tabularnewline
\end{tabular}}
\caption{Operator data. For definitions, see table~\ref{table_invariants}}
\label{CC1_data}
\end{table}
  
\subsubsection{Cubic code 8}
The stabilizer generators of cubic code 8 are given by
\begin{align}
\begin{array}{c}
\drawgenerator{XI}{II}{IX}{XI}{XI}{XX}{XX}{XX}
\quad
\drawgenerator{IZ}{ZZ}{ZZ}{ZZ}{IZ}{II}{ZI}{IZ}
\end{array}
\, .
\end{align}
Our results indicate this model is type-II.

\begin{table}[H]
\renewcommand{\arraystretch}{1.4}{
\begin{tabular}{c|cccccc c}
Model & Deformability & $n_{\text{rods}}^{3D}$ & $n_{\text{rods}}^{2D}\left(zx,xy,yz\right)$ &
$n_{\text{rods}}^{1D}\left(x,y,z\right)$ & $n_{m}\left(zx,xy,yz\right)$ & $d$ & Type\tabularnewline
\hline
CC8 & $\checkmark$ & 0 & $\left(0,0,0\right)$ & $\left(0,0,0\right)$ & $\left(c,c,c\right)$ & 0 & type-II\tabularnewline
\end{tabular}}
\caption{Operator data. For definitions, see table~\ref{table_invariants}}
\label{CC1_data}
\end{table}
\subsubsection{Cubic code 10}
\label{zoo:cc10}

The stabilizer generators of cubic code 10 are given by
\begin{align}
\begin{array}{c}
\drawgenerator{XI}{XI}{IX}{XI}{XI}{XX}{XX}{IX}
\quad
\drawgenerator{IZ}{ZZ}{ZZ}{ZI}{IZ}{IZ}{ZI}{IZ}
\end{array}\, .
\end{align}
Our results indicate this model is type-II. 

\begin{table}[H]
\renewcommand{\arraystretch}{1.4}{
\begin{tabular}{c|cccccc c}
Model & Deformability & $n_{\text{rods}}^{3D}$ & $n_{\text{rods}}^{2D}\left(zx,xy,yz\right)$ &
$n_{\text{rods}}^{1D}\left(x,y,z\right)$ & $n_{m}\left(zx,xy,yz\right)$ & $d$ & Type\tabularnewline
\hline
CC10 & $\checkmark$ & 0 & $\left(0,0,0\right)$ & $\left(0,0,0\right)$ & $\left(c,c,c\right)$ & 0 & type-II\tabularnewline
\end{tabular}}
\caption{Operator data. For definitions, see table~\ref{table_invariants}}
\label{CC1_data}
\end{table}

\subsubsection{Type-II fractal spin liquid}
The local stabilizer generators of this model~\cite{yoshida2013exotic} are given by
\begin{align}
\begin{array}{c}
\xymatrix@!0{%
 && XI \ar@{-}[dl] \ar@{-}[rr] &&   \\
&  XI \ar@{-}[rr] \ar@{-}[ur] &&  \ar@{-}[ur] && IX \ar@{-}[dl]   \\
  XI \ar@{-}[ur] && XX \ar@{-}[ll]\ar@{-}[ur] && II  \ar@{-}[uu]  \\
&  \ar@{.}[uu]\ar@{.}[dl]\ar@{.}[rr] && IX \ar@{-}[ur] \ar@{-}[uu] \\
  \ar@{-}[uu] && IX \ar@{-}[uu]\ar@{-}[ll]\ar@{-}[ur]
}
\quad
\xymatrix@!0{%
  &&& ZI \ar@{-}[rr]\ar@{-}[dl] &&    \\
&&  ZI \ar@{-}[dl]\ar@{-}[dd] &&  \ar@{-}[ur]\ar@{-}[ll]\ar@{-}[dd]   \\
 & II \ar@{-}[dd] && ZZ \ar@{.}[uu]\ar@{.}[rr]\ar@{.}[dl] && IZ \ar@{-}[uu]\ar@{-}[dl]   \\
 ZI  \ar@{-}[ur]&& \ar@{-}[dl]\ar@{-}[rr] && IZ  \\
 &&& IZ \ar@{-}[ll]\ar@{-}[ur]  \\
}
\end{array}
\, .
\end{align}
It was shown to be type-II in Ref.~\onlinecite{yoshida2013exotic}.

\subsubsection{Qutrit Type-II fractal spin liquid}
This is a type-II model with stabilizer generators given by
\begin{align}
\begin{array}{c}
\xymatrix@!0{%
&  IX \ar@{-}[rr] &&   \\
  IX \ar@{-}[ur] && XX \ar@{-}[ll]\ar@{-}[ur] && XI   \\
&  \ar@{.}[uu]\ar@{.}[dl]\ar@{.}[rr] && II \ar@{-}[ur] \ar@{-}[uu] \\
  \ar@{-}[uu] && XI \ar@{-}[uu]\ar@{-}[ll]\ar@{-}[ur]
}
\quad
\xymatrix@!0{%
  && IZ \ar@{-}[rr]\ar@{-}[dl] && &&   \\
&  II \ar@{-}[dl]\ar@{-}[dd] && \ar@{-}[ll]\ar@{-}[ur]\ar@{-}[dd]  \\
 IZ && ZZ \ar@{.}[uu]\ar@{.}[rr]\ar@{.}[dl] && ZI \ar@{-}[uu]\ar@{-}[dl]   \\
 &  \ar@{-}[uu]\ar@{-}[rr] && ZI    
}
\end{array}
\, ,
\end{align}
where $X$ and $Z$ are generalized Pauli operators associated with a qutrit such that $X^3=Z^3=1$. It was shown to be type-II in Ref.~\onlinecite{yoshida2013exotic}. 

\subsubsection{Kim's qutrit models}
This is another 3D qutrit stabilizer model which is likely a type-II model. It was shown numerically in Ref.~\onlinecite{kim20123d} that there are no strings operators up to width 20. The stabilizer generator is given by
\begin{align}
\begin{array}{c}
\drawgenerator{X}{Z}{XZ}{X\bar{Z}}{\bar{X}}{\bar{Z}}{\bar{X}\bar{Z}}{\bar{X}Z}
\end{array}
\, .
\end{align}

A related model is given by the following stabilizer generator,
\begin{align}
\begin{array}{c}
\drawgenerator{X}{Z}{XZ}{\bar{X} Z}{\bar{X}}{\bar{Z}}{\bar{X}\bar{Z}}{X \bar{Z}} 
\end{array}
\, .
\end{align}
There are also closely related ``symmetric" versions, where the back plane of the generator is conjugated, which have the same bulk properties. 

\subsubsection{Kim's $d=5$ qudit model}
This is a 3D qudit ($d=5$) stabilizer model. The no string condition was proven in Ref.~\onlinecite{kim20123d} with a string operator length to width aspect ratio 5.

\begin{align}
\begin{array}{c}
\drawgenerator{X}{Z}{XZ}{X^3\bar{Z}^3}{\bar{X}}{\bar{Z}}{\bar{X}\bar{Z}}{\bar{X}^3Z^3}
\end{array}
\, .
\end{align}
There is also a closely related ``symmetric" version where the back plane is conjugated. 

\subsubsection{HH-II model (?)}
This model was claimed but not shown to be type-II in Ref.~\onlinecite{hsieh_halasz_partons}. Our results were not conclusive, as the deformability condition was only checked to third order. 
This leaves open the possibility that the model is fractal type-I or type-II, which deserves further study. 
The local stabilizer generators are given by
\begin{align}
\begin{array}{c}
\drawgenerator{ZI}{IX}{II}{XZ}{ZI}{IX}{II}{XZ}
\quad
\drawgenerator{IY}{YX}{II}{XI}{IY}{YX}{II}{XI}
\end{array}
\, .
\end{align}

\begin{table}[H]
\renewcommand{\arraystretch}{1.4}{
\begin{tabular}{c|cccccc c}
Model & Deformability & $n_{\text{rods}}^{3D}$ & $n_{\text{rods}}^{2D}\left(zx,xy,yz\right)$ &
$n_{\text{rods}}^{1D}\left(x,y,z\right)$ & $n_{m}\left(zx,xy,yz\right)$ & $d$ & Type\tabularnewline
\hline
HH-II & ? & 0 & $\left(0,0,0\right)$ & $\left(0,0,0\right)$ & $\left(c,c,c\right)$ & 0 & inconclusive\footnote{Our results are consistent with fractal type-I or type-II.} 
\end{tabular}}
\caption{Operator data. For definitions, see table~\ref{table_invariants}}
\label{CC1_data}
\end{table}

\subsection{Fractal type-I  models with lineons}
\label{fractalType1zoo}
In this section we collect the fractal type-I models that support lineons, but do not appear to support planeons (although this is not rigorously shown). This includes cubic codes 0, 5, 6, 9, from Ref.~\onlinecite{haah2011local}, the Sierpinski~\cite{doi:10.1080/14786435.2011.609152} and the Fibonacci fractal spin liquids~\cite{yoshida2013exotic}. Except for cubic code 0, the rest of these examples are in fact lineon models, i.e. all excitations are lineons along a common direction. 

Below we hvae represented string operators in terms of their basic repeating unit segment, following Ref.~\onlinecite{haah2011local}. We refer to this unit as the basic string segment. We use $E[\hat{v}]_p$ to denote a Pauli operator $E$ at position $p$ times its translates along the unit vector $\hat{v}$. With this notation the basic string segment determines the whole string operator. 
For example, $XX[\hat{y}]_{(000)}XI[\hat{y}]_{(100)}$ represents a basic string segment with $XX$ and $XI$ at positions $(0,0,0)$ and $(1,0,0)$, respectively, which is repeated along the $\hat{y}$ direction to generate the full string operator. 

\subsubsection{Cubic code 0}
This is a non-CSS fracton model whose stabilizer generators are given by
\begin{align}
\begin{array}{c}
\drawgenerator{XX}{ZI}{ZY}{XY}{ZZ}{II}{XZ}{ZX}
\quad
\drawgenerator{ZZ}{II}{XZ}{ZX}{XX}{ZI}{ZY}{XY}
\end{array}
\, .
\end{align}
It has 3 basic string segments~\cite{haah2011local} given by
\begin{align}
    ZZ[10\bar{1}]_{(000)}XI[10\bar{1}]_{(001)}ZZ[10\bar{1}]_{(002)} \, ,
     \\
    ZX[\bar{1}10]_{(200)}XI[\bar{1}10]_{(100)}ZX[\bar{1}10]_{(000)} \, ,
     \\
    ZY[0\bar{1}1]_{(020)}XI[0\bar{1}1]_{(010)}ZY[0\bar{1}1]_{(000)}
    \, .
\end{align}

\begin{table}[H]
\renewcommand{\arraystretch}{1.4}{
\begin{tabular}{c|cccccc c}
Model & Deformability & $n_{\text{rods}}^{3D}$ & $n_{\text{rods}}^{2D}\left(zx,xy,yz\right)$ &
$n_{\text{rods}}^{1D}\left(x,y,z\right)$ & $n_{m}\left(zx,xy,yz\right)$ & $d$ & Type\tabularnewline
\hline
CC0 & $\checkmark$ & 0 & $\left(0,0,0\right)$ & $\left(0,0,0\right)$ & $\left(c,c,c\right)$ & 0 & fractal type-I\tabularnewline
\end{tabular}}
\caption{Operator data. For definitions, see table~\ref{table_invariants}}
\label{CC1_data}
\end{table}

\subsubsection{Cubic code 5}
This is a lineon model whose stabilizer generators are given by
\begin{align}
\begin{array}{c}
\drawgenerator{XI}{II}{IX}{II}{XI}{IX}{XX}{XX}
\quad
\drawgenerator{IZ}{ZI}{ZZ}{ZZ}{IZ}{II}{ZI}{II}
\end{array}
\, .
\end{align}

The $X$-sector excitation patterns and a string operator $IX[1\bar{1}2]_{(000)}XX[1\bar{1}2]_{(001)}$ are shown in Fig. \ref{CC5}. 
Cubic code 5 is equivalent to cubic code 9. The transformation relating these models is explained below in section~\ref{CC9_data}. 

\begin{figure}[H]
\vspace{6mm}
\centering
\sidesubfloat[]{\includegraphics[scale=0.28]{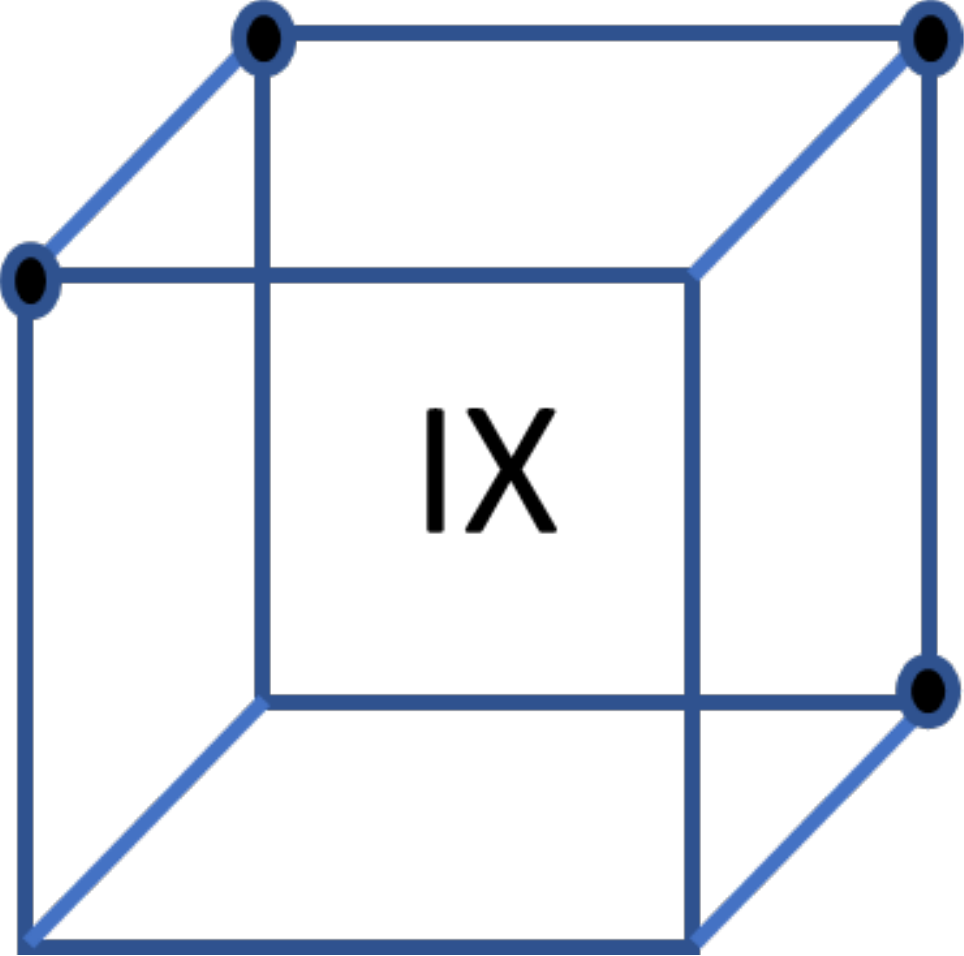}}
\sidesubfloat[]{\includegraphics[scale=0.28]{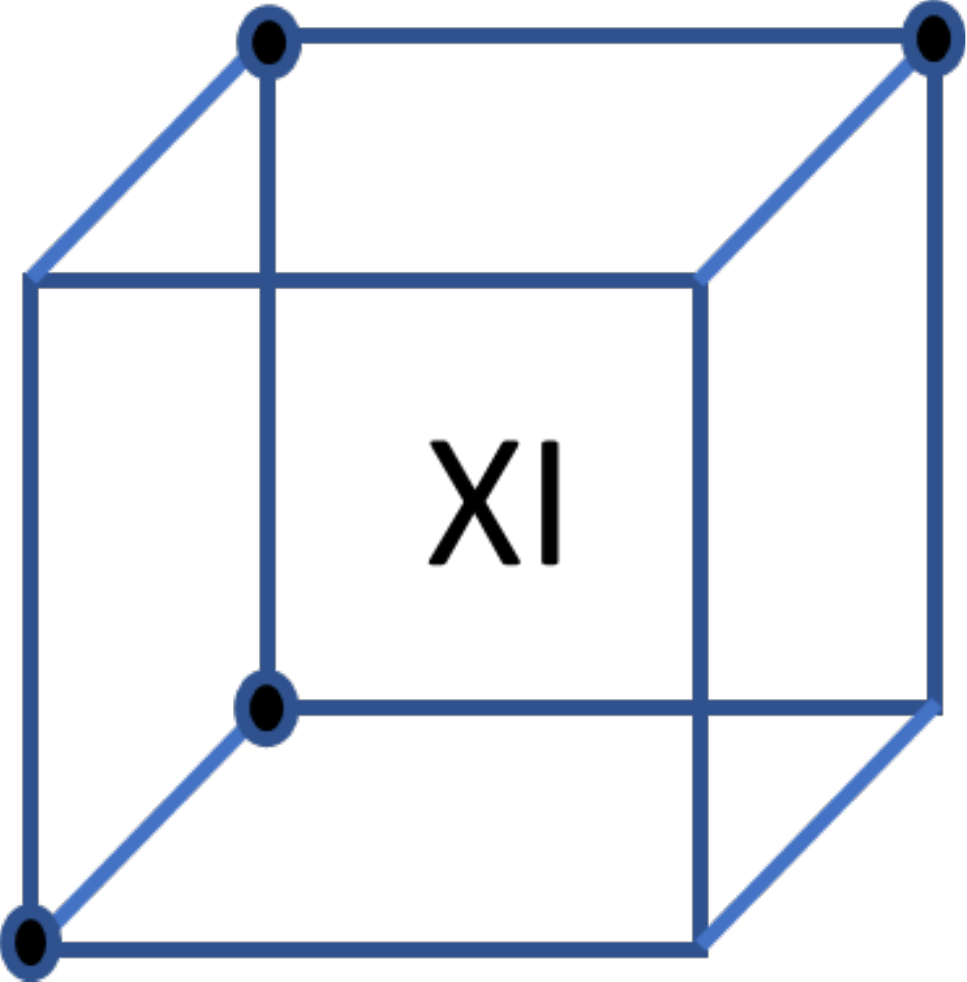}}
\sidesubfloat[]{\includegraphics[scale=0.28]{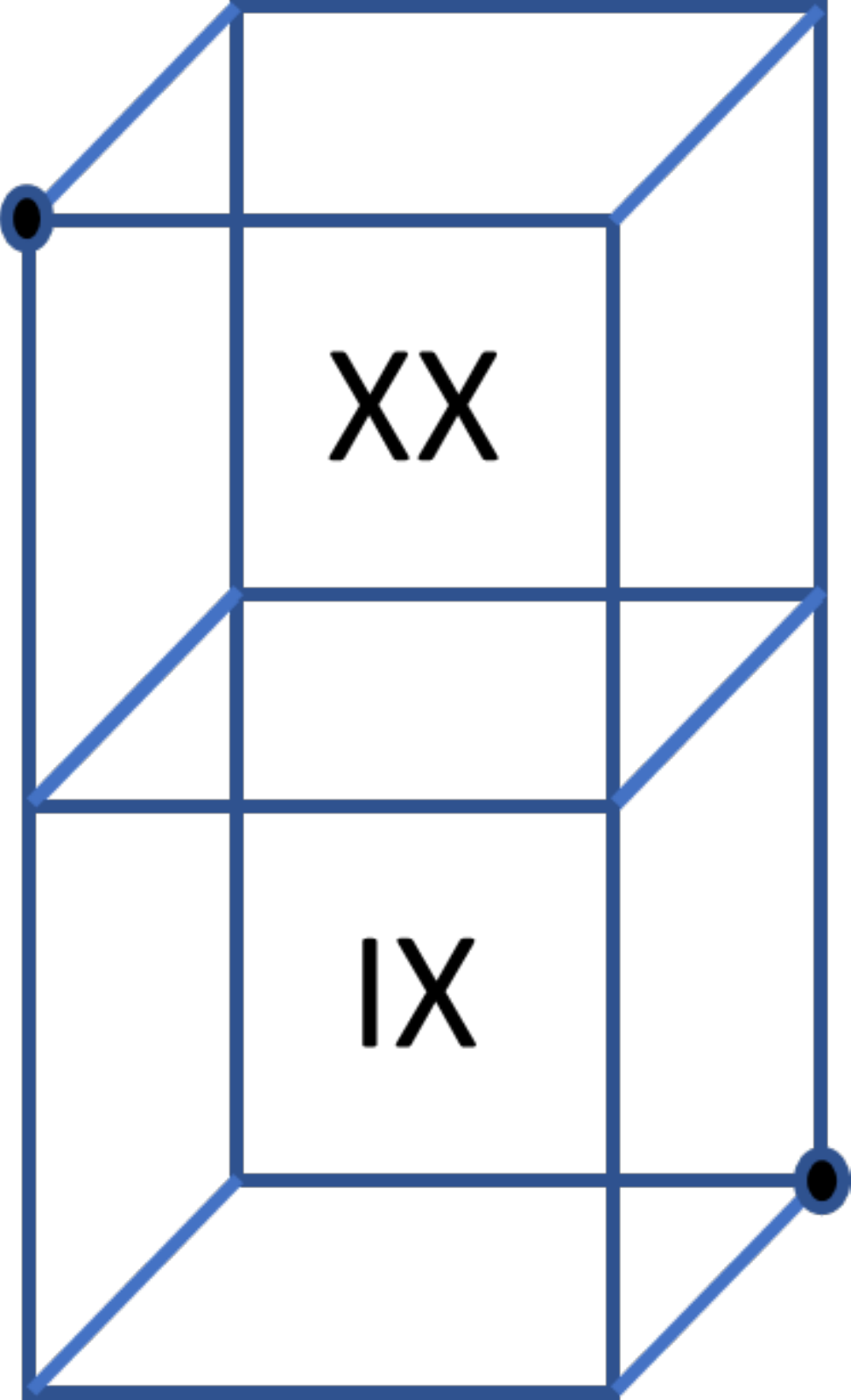}}
\caption{Excitation patterns and  a string operator for cubic code 5}
\label{CC5}
\end{figure}

\begin{table}[H]
\renewcommand{\arraystretch}{1.4}{
\begin{tabular}{c|cccccc c}
Model & Deformability & $n_{\text{rods}}^{3D}$ & $n_{\text{rods}}^{2D}\left(zx,xy,yz\right)$ &
$n_{\text{rods}}^{1D}\left(x,y,z\right)$ & $n_{m}\left(zx,xy,yz\right)$ & $d$ & Type\tabularnewline
\hline
CC5 & $\times$ & 0 & $\left(0,0,0\right)$ & $\left(0,0,0\right)$ & $\left(c,0,c\right)$ & 1 & fractal type-I\tabularnewline
\end{tabular}}
\caption{Operator data. For definitions, see table~\ref{table_invariants}}
\label{CC1_data}
\end{table}

\subsubsection{Cubic code 6}
This is a lineon model whose stabilizer generators are given by 
\begin{align}
\begin{array}{c}
\drawgenerator{XI}{II}{IX}{II}{XX}{XX}{XI}{IX}
\quad
\drawgenerator{ZZ}{ZZ}{IZ}{ZI}{IZ}{II}{ZI}{II}
\end{array}
\, .
\end{align}
The $X$-sector excitation patterns and a string operator $XX[\bar{1}02]_{(000)}XI[\bar{1}02]_{(001)}$ are shown in Fig.~\ref{CC6}.

\begin{figure}[H]
\vspace{6mm}
\centering
\sidesubfloat[]{\includegraphics[scale=0.28]{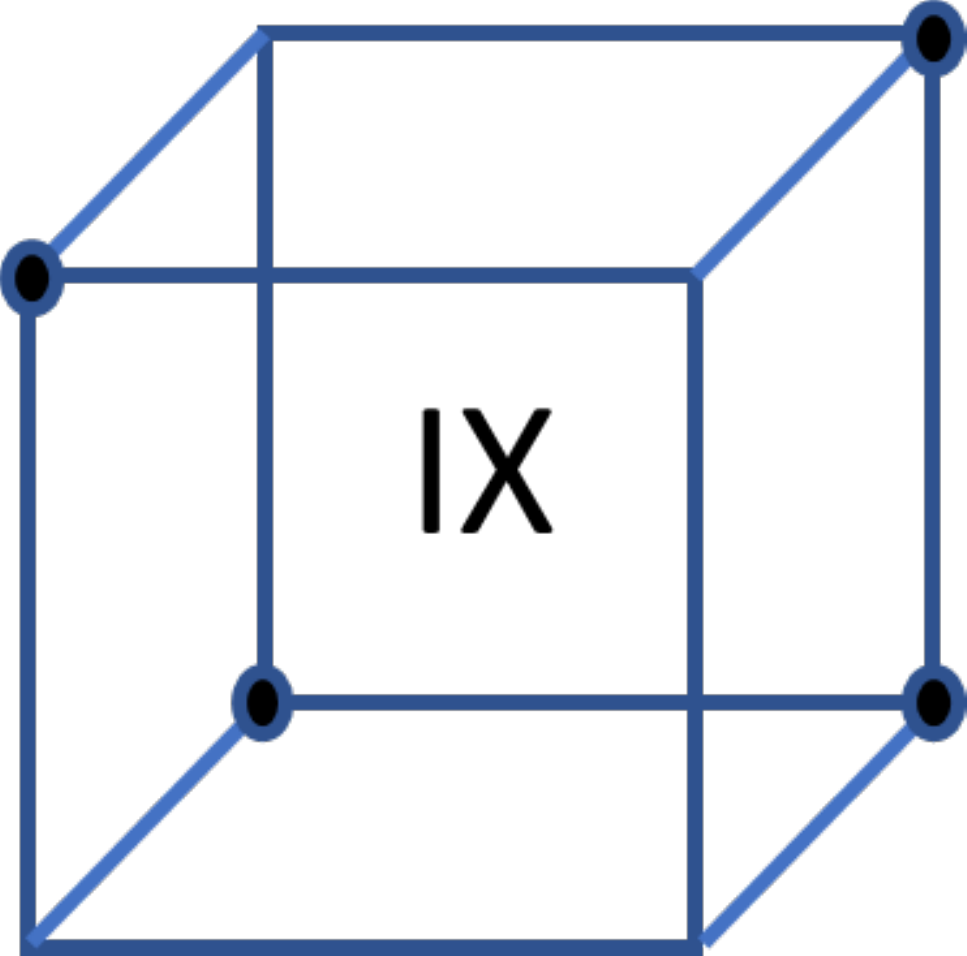}}
\sidesubfloat[]{\includegraphics[scale=0.28]{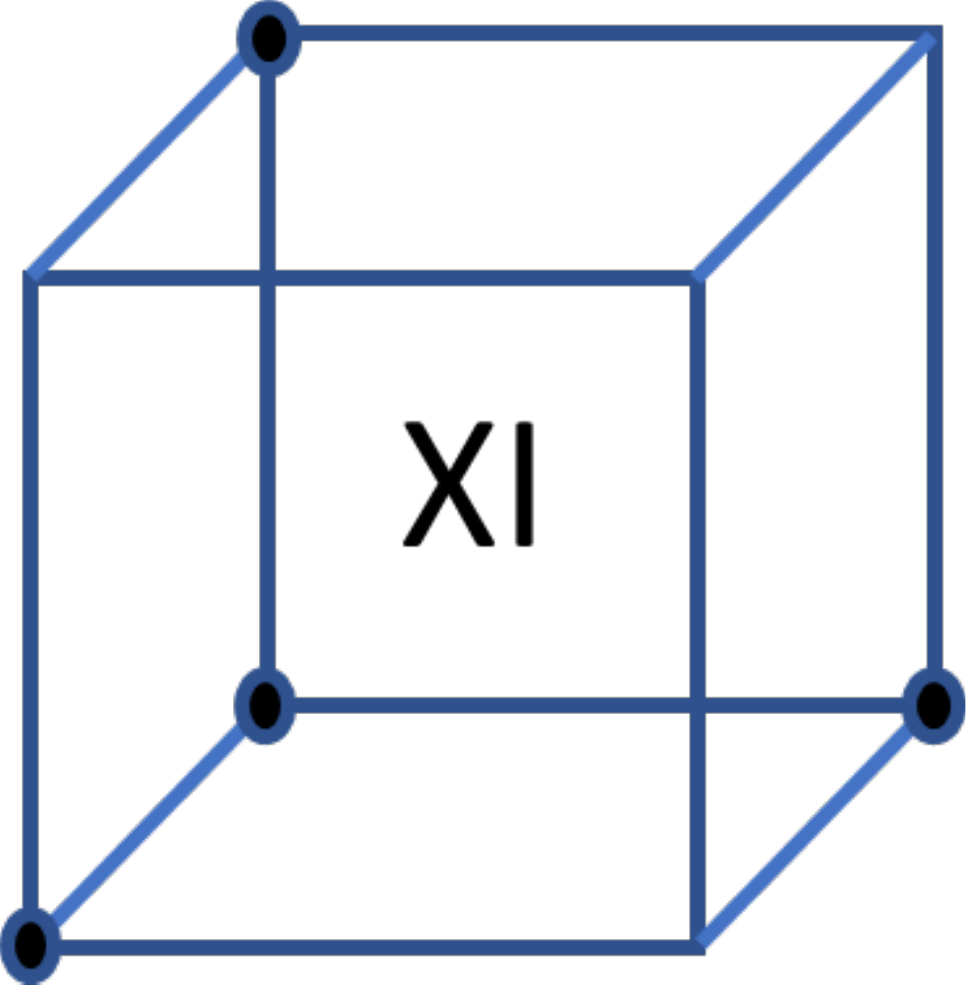}}
\sidesubfloat[]{\includegraphics[scale=0.28]{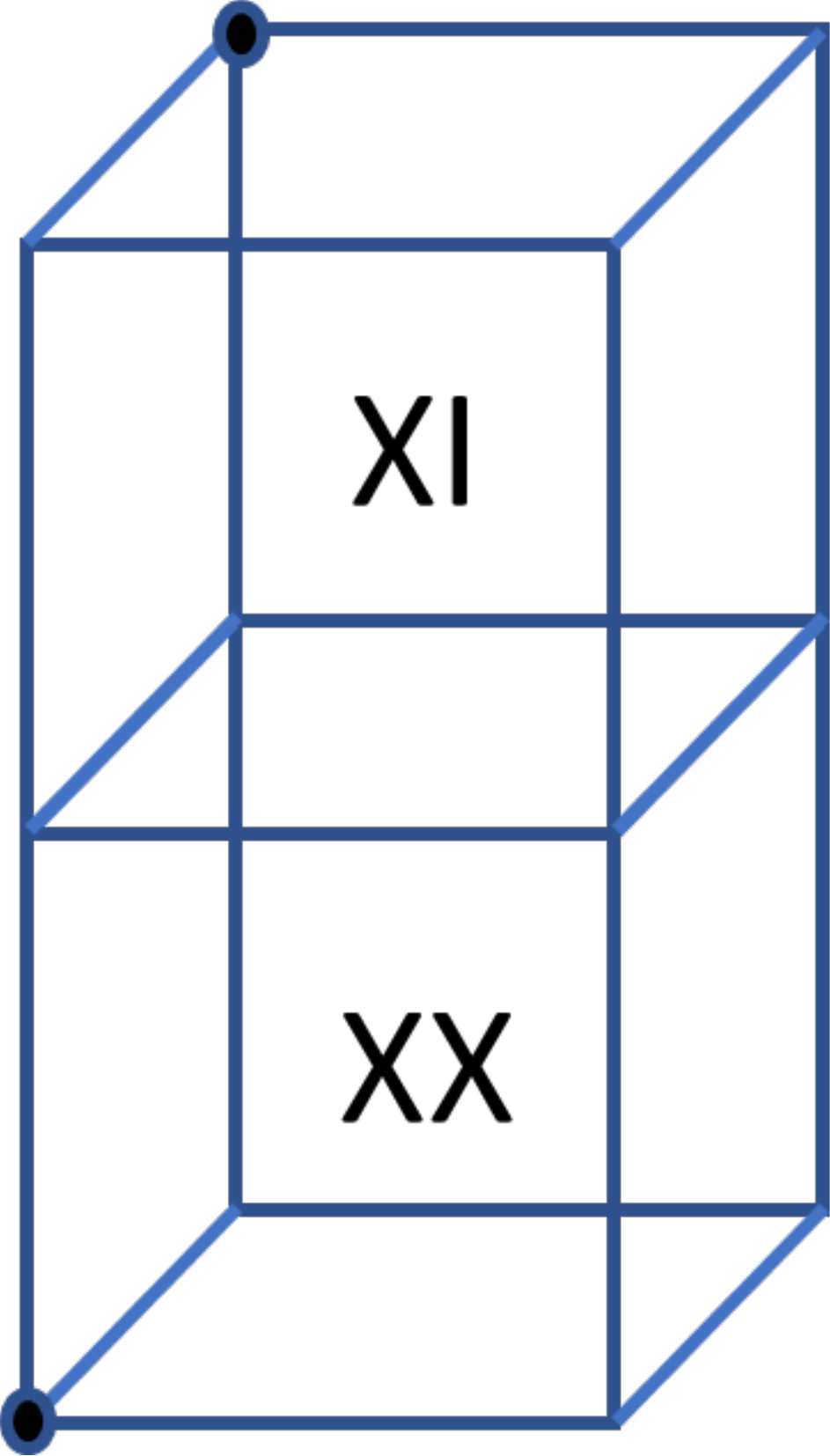}}
\caption{Excitation patterns and a string operator for Cubic code 6}
\label{CC6}
\end{figure}

\begin{table}[H]
\renewcommand{\arraystretch}{1.4}{
\begin{tabular}{c|cccccc c}
Model & Deformability & $n_{\text{rods}}^{3D}$ & $n_{\text{rods}}^{2D}\left(zx,xy,yz\right)$ &
$n_{\text{rods}}^{1D}\left(x,y,z\right)$ & $n_{m}\left(zx,xy,yz\right)$ & $d$ & Type\tabularnewline
\hline
CC6 & $\times$ & 0 & $\left(0,0,0\right)$ & $\left(0,0,0\right)$ & $\left(c,0,c\right)$ & 1 & fractal type-I\tabularnewline
\end{tabular}}
\caption{Operator data. For definitions,  see table~\ref{table_invariants}}
\label{CC1_data}
\end{table}

\subsubsection{Cubic code 9}
\label{CC9_data}
This is a lineon model whose stabilizer generators are given by
\begin{align}
\begin{array}{c}
\drawgenerator{XI}{II}{IX}{XI}{IX}{IX}{XI}{XX}
\quad
\drawgenerator{ZI}{ZI}{IZ}{ZZ}{IZ}{II}{ZI}{IZ}
\end{array}
\end{align}
The $X$-sector excitation patterns and a string operator $IX[11\bar{2}]_{(000)}XI[11\bar{2}]_{(001)}$ are shown in Fig.~\ref{CC5}.  
% stabilizer relations indicate that a single local stabilizer generator is a fracton and there exists a string operator, we know this model is type-I. 
We have found that cubic code 9 is equivalent to cubic code 5. Applying a modular transformation $z\rightarrow zx^{-1}$, shifting the first qubit on every vertex by one unit in the $\hat{x}$ direction and applying a CNOT operation between first and second qubits on every vertex maps cubic code 9 to cubic code 5.

\begin{figure}[H]
\vspace{6mm}
\centering
\sidesubfloat[]{\includegraphics[scale=0.28]{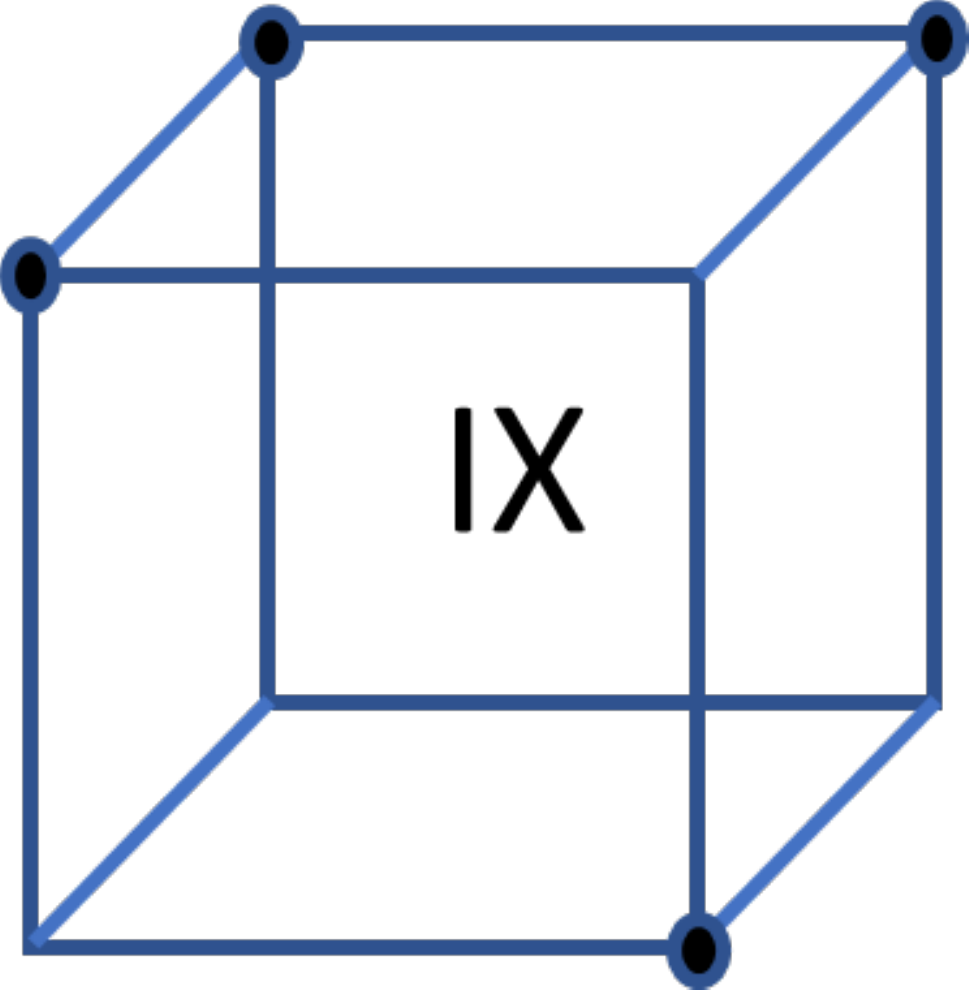}}
\sidesubfloat[]{\includegraphics[scale=0.28]{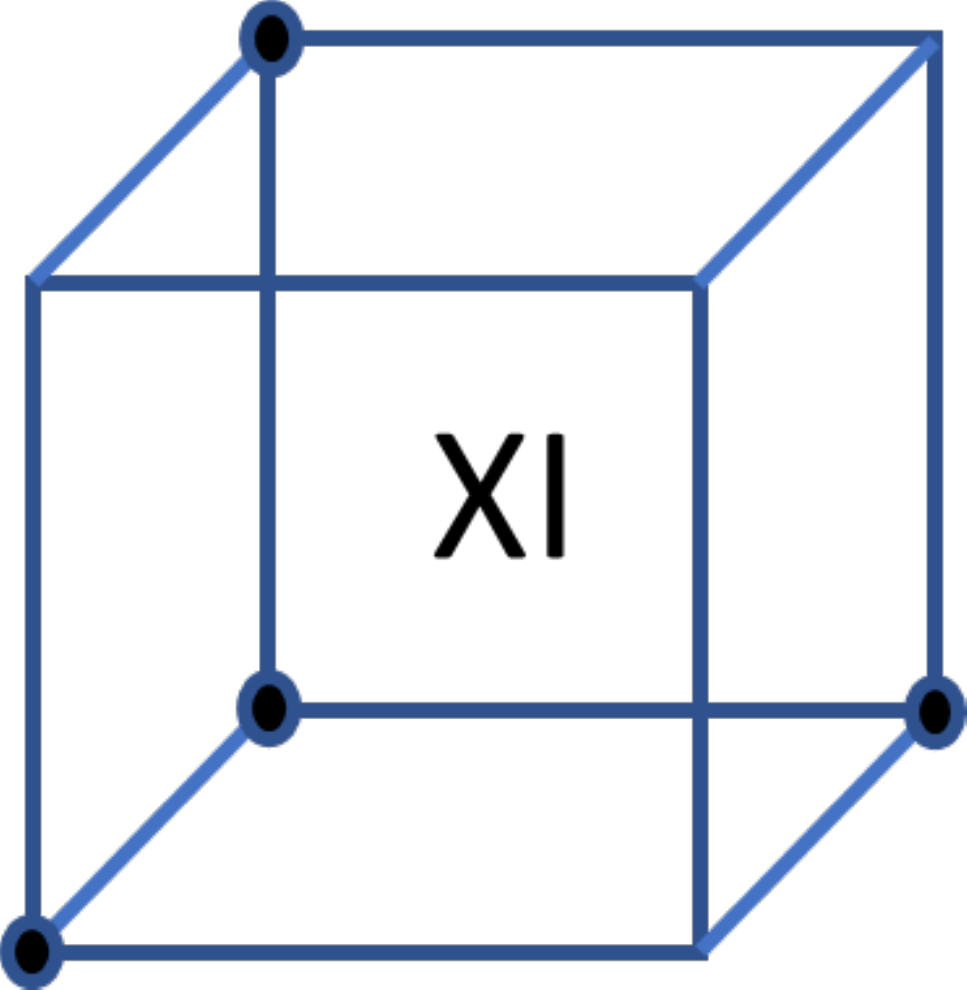}}
\sidesubfloat[]{\includegraphics[scale=0.28]{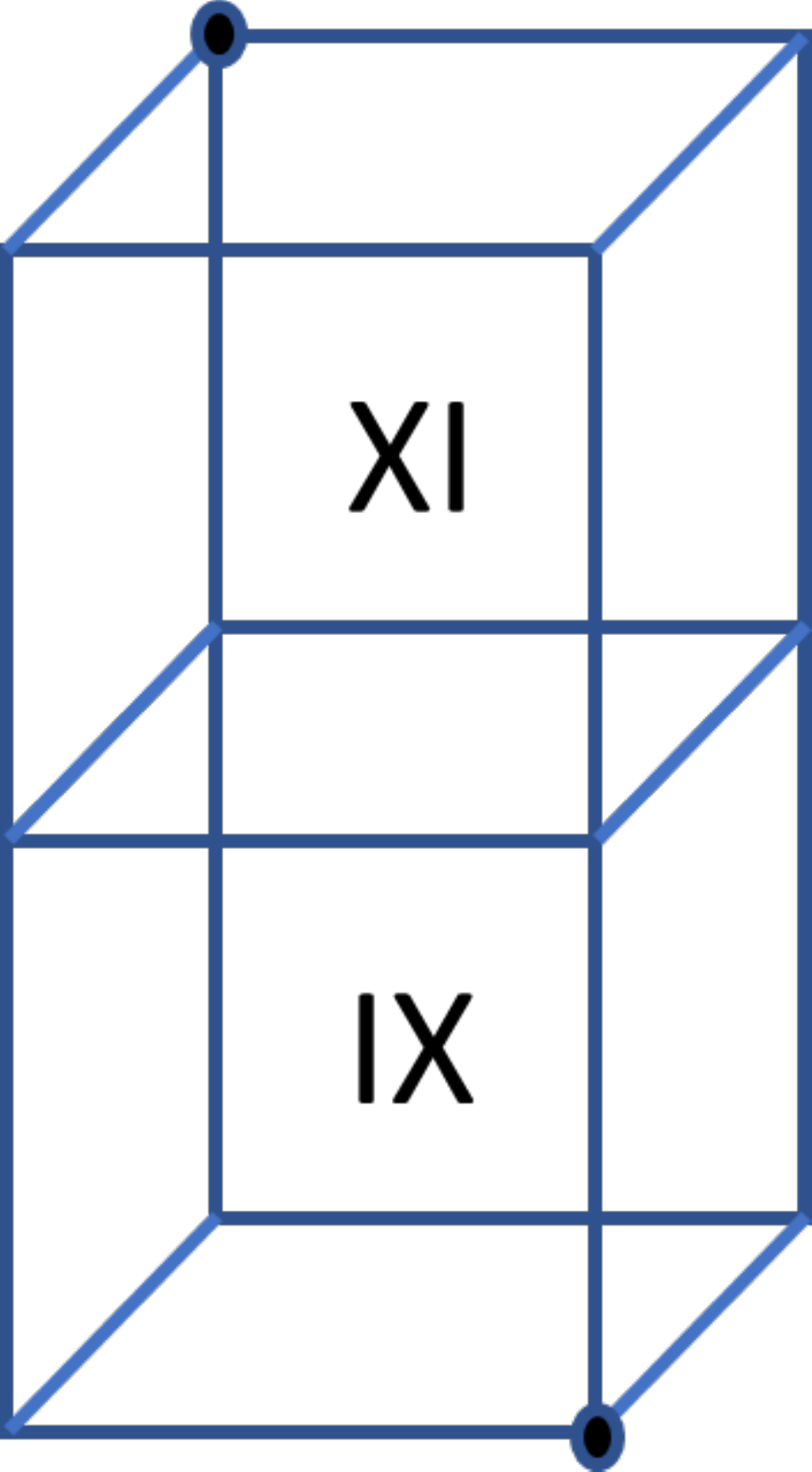}}
\caption{Excitation patterns and a string operator for Cubic code 9}
\label{CC6}
\end{figure}

\begin{table}[H]
\renewcommand{\arraystretch}{1.4}{
\begin{tabular}{c|cccccc c}
Model & Deformability & $n_{\text{rods}}^{3D}$ & $n_{\text{rods}}^{2D}\left(zx,xy,yz\right)$ &
$n_{\text{rods}}^{1D}\left(x,y,z\right)$ & $n_{m}\left(zx,xy,yz\right)$ & $d$ & Type\tabularnewline
\hline
CC9 & $\times$ & 0 & $\left(0,0,0\right)$ & $\left(0,0,0\right)$ & $\left(c,c,c\right)$ & 1 & fractal type-I\tabularnewline
\end{tabular}}
\caption{Operator data. For definitions, see table~\ref{table_invariants}}
% \label{CC9_data}
\end{table}

\subsubsection{Sierpinski fractal spin liquid}
This is a lineon model due to Chamon-Castelnovo~\cite{doi:10.1080/14786435.2011.609152} and Yoshida~\cite{yoshida2013exotic} whose stabilizer generators are given by
\begin{align}
\begin{array}{c}
\drawgenerator{IX}{XX}{}{XI}{}{}{}{IX}
\quad
\drawgenerator{}{}{}{ZI}{ZI}{ZZ}{}{IZ}
\end{array}
\, .
\end{align}

\begin{table}[H]
\renewcommand{\arraystretch}{1.4}{
\begin{tabular}{c|cccccc c}
Model & Deformability & $n_{\text{rods}}^{3D}$ & $n_{\text{rods}}^{2D}\left(zx,xy,yz\right)$ &
$n_{\text{rods}}^{1D}\left(x,y,z\right)$ & $n_{m}\left(zx,xy,yz\right)$ & $d$ & Type\tabularnewline
\hline
SFSL & $\checkmark$ & 0 & $\left(0,0,0\right)$ & $\left(0,0,n_1\right)$ & $\left(0,0,c\right)$ & 1 & fractal type-I
\end{tabular}}
\caption{Operator data. For definitions, see table~\ref{table_invariants}}
\end{table}

\subsubsection{Fibonacci fractal spin liquid }
This is a lineon model due to Yoshida~\cite{yoshida2013exotic} whose stabilizer generators are given by
\begin{align}
\begin{array}{c}
\xymatrix@!0{%
&& && IX \ar@{-}[dl]   \\
& && IX \ar@{-}[rr] &&     \\
 && IX \ar@{-}[ur] && XX \ar@{-}[ll]\ar@{-}[ur]   \\
& && \ar@{.}[uu]\ar@{.}[dl]\ar@{.}[rr] &&  \ar@{-}[uu] \\
&&  \ar@{-}[uu] && XI \ar@{-}[uu]\ar@{-}[ll]\ar@{-}[ur]
}
\quad
\xymatrix@!0{%
 && IZ \ar@{-}[rr]\ar@{-}[dl] &&    \\
&  &&  \ar@{-}[ur]\ar@{-}[ll]\ar@{-}[dd]   \\
&& ZZ \ar@{.}[uu]\ar@{.}[rr]\ar@{.}[dl] && ZI \ar@{-}[uu]\ar@{-}[dl]   \\
&  \ar@{-}[uu]\ar@{-}[rr] && ZI    \\
 && ZI \ar@{-}[ur]  
}
\end{array}
\, .
\end{align}

\subsection{Fractal type-I  models with planons}

In this section we collect fractal type-I models that support composite planons: cubic codes 11-17 from Ref.~\onlinecite{haah2011local}. 
All of these models include planons in a single stack of parallel planes, we indicate the orientation and braiding statistics of these planons for each model. 
The planons in these models were found via generalized Gauss's laws, each formed by the product of stabilizer generators over a membrane that leave a planon string operator along the boundary. These Gauss's laws are easy to identify by inspection of a single stabilizer generator: since the product of the operators on all corners of any cubic code generator is equal to the identity, a planar Gauss's law arises whenever the product of operators on the corners of a single square face of a generator equals the identity. 

\subsubsection{Cubic code 11}
This is a fracton model with stabilizer generators given by 
\begin{align}
\begin{array}{c}
\drawgenerator{IX}{II}{XI}{XX}{IX}{XX}{II}{XI}
\quad
\drawgenerator{ZI}{ZZ}{II}{IZ}{ZI}{II}{IZ}{ZZ}
\end{array}
\, .
\end{align}
This model supports particles of all sub-dimensional mobilities 0, 1 and 2. 
Planon operators~\cite{haah2011local} are given by $ZZ[\hat{z}]_{(000)}ZI[\hat{z}]_{(100)}$ and $XI[\hat{y}]_{(000)}IX[\hat{y}]_{(100)}$ and there is also a lineon operator $XX[\hat{x}]_{(000)}IX[\hat{x}]_{(100)}XI[\hat{x}]_{(200)}$. 
The planon operators lie in the $yz$ plane due to emergent Gauss's laws which arise since the product of operators at the corners of a $yz$ square in a single generator equals the identity. 
For example $XI$, $XX$, $IX$ and $II$ on the front $yz$ face multiply to give $II$. 
The string operator $XI[\hat{y}]_{(000)}IX[\hat{y}]_{(100)}$ appears along the top and bottom edge of a Gauss's law membrane and can be seen to arise from the product of corner operators in a single generator along the boundary of the membrane, e.g. the product of $XX$ and $IX$ give $XI$ along the lower edge of the back face and $XI$, $XX$ give $IX$ along the lower edge of the front face. 

\begin{figure}[h!]
\centering
\includegraphics[scale=1]{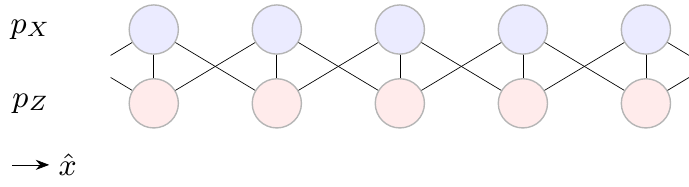}
 \caption{Anticommutation relations for planons parallel to the $yz$ plane in CC11. $p_X$ denotes composite planon excitations of $X$ stabilizers, and similarly for $Z$. An edge between a $p_X$ and $p_Z$ planon indicates that they have a braiding phase of $-1$.  }
\label{fig:CC11planons}
\end{figure}

\begin{table}[H]
\renewcommand{\arraystretch}{1.4}{
\begin{tabular}{c|cccccc c}
Model & Deformability & $n_{\text{rods}}^{3D}$ & $n_{\text{rods}}^{2D}\left(zx,xy,yz\right)$ &
$n_{\text{rods}}^{1D}\left(x,y,z\right)$ & $n_{m}\left(zx,xy,yz\right)$ & $d$ & Type\tabularnewline
\hline
CC11 & $\checkmark$ & 0 & $\left(0,0,n_1\right)$ & $\left(n_2,n_1,n_1\right)$ & $\left(c,c,\ell\right)$ & 0 & fractal type-I\tabularnewline
\end{tabular}}
\caption{Operator data. For definitions,  see table~\ref{table_invariants}}
\label{CC1_data}
\end{table}

\subsubsection{Cubic code 12}
This is a fracton model with stabilizer generators given by 
\begin{align}
\begin{array}{c}
\drawgenerator{IX}{II}{II}{XI}{IX}{XI}{XX}{XX}
\quad
\drawgenerator{ZI}{IZ}{ZZ}{ZZ}{ZI}{II}{II}{IZ}
\end{array}
\, .
\end{align}
This model has all sub-dimensional mobilities 0, 1 and 2. The planon  operators~\cite{haah2011local} are $IZ[\hat{z}]_{(000)}ZI[\hat{z}]_{(010)}$ and $XI[\hat{x}]_{(000)}XX[\hat{x}]_{(010)}$ due to Gauss's laws in the $xz$ planes.  

\begin{figure}[h!]
\centering
\includegraphics[scale=1]{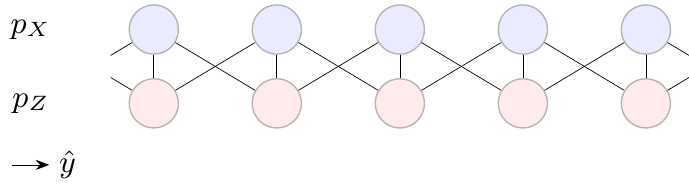}
 \caption{Anticommutation relations for planons parallel to the $xz$ plane in CC12. Where $p_X$ denotes composite planon excitations of $X$ stabilizers, and similarly for $Z$. An edge between a $p_X$ and $p_Z$ planon indicates that they have a braiding phase of $-1$.}
\label{fig:CC12planons}
\end{figure}
\begin{table}[H]
\renewcommand{\arraystretch}{1.4}{
\begin{tabular}{c|cccccc c}
Model & Deformability & $n_{\text{rods}}^{3D}$ & $n_{\text{rods}}^{2D}\left(zx,xy,yz\right)$ &
$n_{\text{rods}}^{1D}\left(x,y,z\right)$ & $n_{m}\left(zx,xy,yz\right)$ & $d$ & Type\tabularnewline
\hline
CC12 & $\checkmark$ & 0 & $\left(n_1,0,0\right)$ & $\left(n_1,0,n_1\right)$ & $\left(\ell,c,c\right)$ & 0 & fractal type-I\tabularnewline
\end{tabular}}
\caption{Operator data. For definitions, see table~\ref{table_invariants}}
\label{CC1_data}
\end{table}

\subsubsection{Cubic code 13}
This is a fracton model with stabilizer generators given by  
\begin{align}
\begin{array}{c}
\drawgenerator{XI}{II}{II}{XX}{IX}{XX}{XI}{IX}
\quad
\drawgenerator{ZI}{ZZ}{IZ}{ZI}{IZ}{II}{II}{ZZ}
\end{array}
\, .
\end{align}
This model has all sub-dimensional mobilities 0, 1 and 2. A planon operator~\cite{haah2011local} is given by $ZZ[\hat{z}]_{(000)}IZ[\hat{z}]_{(010)}$ due to Gauss's laws in the $xz$ planes.

\begin{figure}[!htb]
\centering
\includegraphics[scale=1]{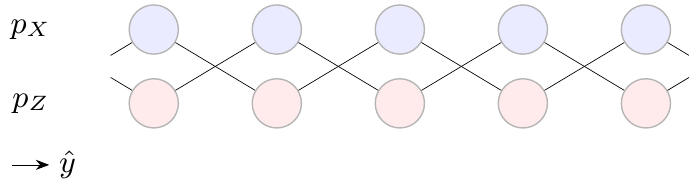}
 \caption{Anticommutation relations for planons parallel to the $xz$ plane in CC13. Where $p_X$ denotes composite planon excitations of $X$ stabilizers, and similarly for $Z$. An edge between a $p_X$ and $p_Z$ planon indicates that they have a braiding phase of $-1$.}
\label{fig:CC13planons}
\end{figure}

\begin{table}[!htb]
\renewcommand{\arraystretch}{1.4}{
\begin{tabular}{c|cccccc c}
Model & Deformability & $n_{\text{rods}}^{3D}$ & $n_{\text{rods}}^{2D}\left(zx,xy,yz\right)$ &
$n_{\text{rods}}^{1D}\left(x,y,z\right)$ & $n_{m}\left(zx,xy,yz\right)$ & $d$ & Type\tabularnewline
\hline
CC13 & $\checkmark$ & 0 & $\left(n_1,0,0\right)$ & $\left(n_1,0,n_1\right)$ & $\left(\ell,c,0\right)$ & 0 & fractal type-I\tabularnewline
\end{tabular}}
\caption{Operator data. For definitions, see table~\ref{table_invariants}}
\label{CC1_data}
\end{table}

\subsubsection{Cubic code 14}
This is a fracotn model with stabilizer generators given by 
\begin{align}
\begin{array}{c}
\drawgenerator{XI}{II}{XX}{XI}{IX}{XI}{XX}{XX}
\quad
\drawgenerator{ZI}{IZ}{ZZ}{ZZ}{IZ}{II}{ZZ}{IZ}
\end{array}
\, .
\end{align}
This model has has all sub-dimensional mobilities 0, 1 and 2. A planon operator~\cite{haah2011local} is given by $IX[\hat{z}]_{(000)}XI[\hat{z}]_{(010)}$ due to Gauss's laws in the $xz$ planes.

\begin{figure}[!htb]
\centering
\includegraphics[scale=1]{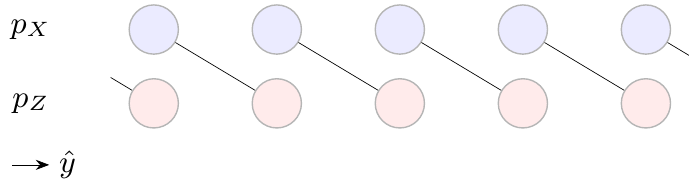}
 \caption{Anticommutation relations for planons parallel to the $xz$ plane in CC14. Where $p_X$ denotes composite planon excitations of $X$ stabilizers, and similarly for $Z$. An edge between a $p_X$ and $p_Z$ planon indicates that they have a braiding phase of $-1$.}
\label{fig:CC14planons}
\end{figure}

\begin{table}[!htb]
\renewcommand{\arraystretch}{1.4}{
\begin{tabular}{c|cccccc c}
Model & Deformability & $n_{\text{rods}}^{3D}$ & $n_{\text{rods}}^{2D}\left(zx,xy,yz\right)$ &
$n_{\text{rods}}^{1D}\left(x,y,z\right)$ & $n_{m}\left(zx,xy,yz\right)$ & $d$ & Type\tabularnewline
\hline
CC14 & $\checkmark$ & 0 & $\left(n_1,0,0\right)$ & $\left(n_1,0,n_1\right)$ & $\left(\ell,c,c\right)$ & 0 & fractal type-I\tabularnewline
\end{tabular}}
\caption{Operator data. For definitions, see table~\ref{table_invariants}}
\label{CC1_data}
\end{table}

\subsubsection{Cubic code 15}
This is a fracton model with stabilizer generators given by 
\begin{align}
\begin{array}{c}
\drawgenerator{XI}{XX}{II}{IX}{IX}{XI}{II}{XX}
\quad
\drawgenerator{ZI}{IZ}{II}{ZZ}{IZ}{ZZ}{II}{ZI}
\end{array}
\, .
\end{align}
This model has all sub-dimensional mobilities 0, 1 and 2. The planon operators~\cite{haah2011local} are $ZI[\hat{y}]_{(000)}ZZ[\hat{y}]_{(100)}$ and $IX[\hat{z}]_{(000)}XI[\hat{z}]_{(100)}$ due to Gauss's laws along $yz$ planes. 

\begin{figure}[!htb]
\centering
\includegraphics[scale=1]{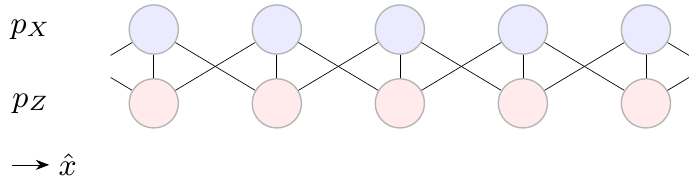}
\caption{Anticommutation relations for planons parallel to the $yz$ plane in CC15. Where $p_X$ denotes composite planon excitations of $X$ stabilizers, and similarly for $Z$. An edge between a $p_X$ and $p_Z$ planon indicates that they have a braiding phase of $-1$.}
\label{fig:CC15planons}
\end{figure}

\begin{table}[H]
\renewcommand{\arraystretch}{1.4}{
\begin{tabular}{c|cccccc c}
Model & Deformability & $n_{\text{rods}}^{3D}$ & $n_{\text{rods}}^{2D}\left(zx,xy,yz\right)$ &
$n_{\text{rods}}^{1D}\left(x,y,z\right)$ & $n_{m}\left(zx,xy,yz\right)$ & $d$ & Type\tabularnewline
\hline
CC15 & $\checkmark$ & 0 & $\left(0,0,n_1\right)$ & $\left(0,n_1,n_1\right)$ & $\left(c,c,\ell\right)$ & 0 & fractal type-I\tabularnewline
\end{tabular}}
\caption{Operator data. For definitions, see table~\ref{table_invariants}}
\label{CC1_data}
\end{table}

\subsubsection{Cubic code 16}
\label{CC16_data}
This is a fracton model that turns out to be equivalent to CC15. The stabilizer generators are given by 
\begin{align}
\begin{array}{c}
\drawgenerator{XI}{XX}{II}{XX}{IX}{IX}{II}{XI}
\quad
\drawgenerator{ZI}{ZI}{II}{IZ}{IZ}{ZZ}{II}{ZZ}
\end{array}
\, .
\end{align}
This model has all sub-dimensional mobilities 0, 1 and 2. Several string operators supported by the model~\cite{haah2011local} are $ZZ[101]_{(000)}IZ[\hat{101}]_{(100)}$ and $IX[110]_{(000)}XI[110]_{(100)}$.

Applying a modular transformation $x\mapsto x/y,\, y\mapsto y,\, z\mapsto z y$ and relabelling the axes we find
\begin{align}
\begin{array}{c}
\drawgenerator{IX}{IX}{XI}{II}{XI}{XX}{XX}{II}
\quad
\drawgenerator{IZ}{ZZ}{ZZ}{II}{ZI}{ZI}{IZ}{II}
\end{array}
\end{align}
Rotating $\pi /2$ anticlockwise around $\hat{y}$, and reflecting across the $x=y$ plane we find CC15. 

% \begin{figure}[h!]
% \centering
% \includegraphics[scale=1]{CC16Planons}
%  \caption{Anticommutation relations for planons in CC16 after modular transformation. }
% \label{fig:CC15planons}
% \end{figure}
\begin{table}[H]
\renewcommand{\arraystretch}{1.4}{
\begin{tabular}{c|cccccc c}
Model & Deformability & $n_{\text{rods}}^{3D}$ & $n_{\text{rods}}^{2D}\left(zx,xy,yz\right)$ &
$n_{\text{rods}}^{1D}\left(x,y,z\right)$ & $n_{m}\left(zx,xy,yz\right)$ & $d$ & Type\tabularnewline
\hline
CC16 & $\times$ & 0 & $\left(0,0,0\right)$ & $\left(0,0,0\right)$ & $\left(c,c,c\right)$ & 0 & fractal type-I\tabularnewline
\end{tabular}}
\caption{Operator data. For definitions,  see table~\ref{table_invariants}}
\label{CC1_data}
\end{table}

\subsubsection{Cubic code 17}
This is a fracton model with stabilizer generators given by 
\begin{align}
\begin{array}{c}
\drawgenerator{XI}{IX}{IX}{XX}{IX}{XX}{XI}{IX}
\quad
\drawgenerator{ZI}{ZZ}{IZ}{ZI}{IZ}{ZI}{ZI}{ZZ}
\end{array}\, .
\end{align}
This model has all sub-dimensional mobilities 0, 1 and 2. A planon operator~\cite{haah2011local} is given by $ZZ[\hat{x}]_{(000)}IZ[\hat{x}]_{(001)}$ due to Gauss's laws parallel to the $xz$ plane.

\begin{figure}[h!]
\centering
\includegraphics[scale=1]{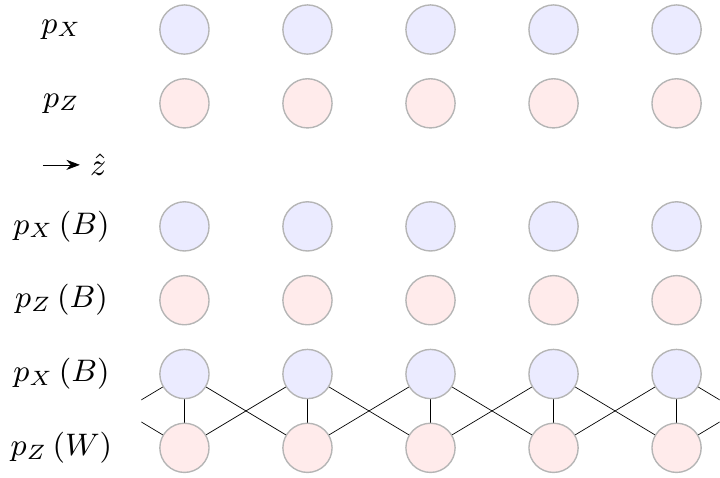}
 \caption{Anticommutation relations for planons parallel to the $xz$ plane in CC17. The planons can be decomposed into those hopping on a $(B/W)$ checkerboard coloring of the 2D $\hat{z}$-planes.  Where $p_X (B)$  denotes composite planon excitations of $X$ stabilizers on the black sublattice, and similarly for $Z$ and $(W)$. An edge between a $p_X$ and $p_Z$ planon indicates that they have a braiding phase of $-1$.}
\label{fig:CC17planons}
\end{figure}

\begin{table}[H]
\renewcommand{\arraystretch}{1.4}{
\begin{tabular}{c|cccccc c}
Model & Deformability & $n_{\text{rods}}^{3D}$ & $n_{\text{rods}}^{2D}\left(zx,xy,yz\right)$ &
$n_{\text{rods}}^{1D}\left(x,y,z\right)$ & $n_{m}\left(zx,xy,yz\right)$ & $d$ & Type\tabularnewline
\hline
CC17 & $\checkmark$ & 0 & $\left(0,n_1,0\right)$ & $\left(n_1,n_1,n_2\right)$ & $\left(c,\ell,c\right)$ & 0 & fractal type-I\tabularnewline
\end{tabular}}
\caption{Operator data. For definitions, see table~\ref{table_invariants}}
\label{CC1_data}
\end{table}

\subsection{ Foliated type-I  models}
\label{foliatedType1zoo}
In this section we collect foliated type-I models: stacks of 2D toric code in 3D, the twice foliated X-cube model~\cite{shirley2018Fractional}, the standard X-cube model~\cite{vijay2016fracton}, the Checkerboard model~\cite{new_TO_vijay}, Halasz and Hsieh's type-I model~\cite{hsieh_halasz_partons}, Chamon's model~\cite{chamon2005quantum}, the membrane coupled X-cube model~\cite{PhysRevB.95.245126}, and Halasz, Hsieh and Balents's models~\cite{HHB_models}. Only the first four models listed above have been rigorously shown to be foliated, we show that Halasz and Hsieh's type-I model is equivalent to two copies of the checkerboard model and hence also foliated. The other models are expected to be foliated but it remains an interesting open question to show this. 
We remark that closely related Majorana stabilizer models have appeared in the literature~\cite{new_TO_vijay}, and it has even been demonstrated that the Majorana checkerboard model is foliated and equivalent to the qubit X-cube model~\cite{Wang2019}. For simplicity we only consider qudit stabilizer models and so do not include any majorana stabilizer models here. 

\label{sec:Type1}
\subsubsection{Stack of 2D toric codes}
A stack of 2D toric codes forms a 3D model with planons obviously appearing in the planes of the 2D toric codes. The stabilizer generators are given by
\begin{align}
\begin{array}{c}
\drawgenerator{}{}{IX}{XX}{XI}{}{}{}
\quad
\drawgenerator{}{}{IZ}{}{ZI}{ZZ}{}{}
\end{array}
\, .
\end{align}  
We remark there are also $\Z_N$ versions of the 2D toric code and stacks made up from them. 
  
\begin{table}[H]
\renewcommand{\arraystretch}{1.4}{
\begin{tabular}{c|cccccc c}
Model & Deformability & $n_{\text{rods}}^{3D}$ & $n_{\text{rods}}^{2D}\left(zx,xy,yz\right)$ &
$n_{\text{rods}}^{1D}\left(x,y,z\right)$ & $n_{m}\left(zx,xy,yz\right)$ & $d$ & Type\tabularnewline
\hline
2DTC$_{xy}$ & $\checkmark$ & 0 & $\left(0,n_1,0\right)$ & $\left(n_1,n_1,0\right)$ & $\left(0,\ell,0\right)$ & 2 & foliated type-I \tabularnewline
 
2DTC$_{yz}$ & $\checkmark$ & 0 & $\left(0,0,n_1\right)$ & $\left(0,n_1,n_1\right)$ & $\left(\ell,0,0\right)$ & 2 & foliated type-I \tabularnewline
 
2DTC$_{xz}$ & $\checkmark$ & 0 & $\left(n_1,0,0\right)$ & $\left(n_1,0,n_1\right)$ & $\left(0,0,\ell\right)$ & 2 & foliated type-I \tabularnewline
 \end{tabular}}
\caption{Operator data. For definitions, see table~\ref{table_invariants}}
\label{CC1_data}
\end{table}

\subsubsection{Twice foliated X-cube model}
This is a lineon model whose stabilizer generators are given by 
\begin{align}
\begin{array}{c}
\drawgenerator{}{}{IX}{IX}{IX}{XX}{}{XI}
\quad
\drawgenerator{ZI}{ZZ}{}{IZ}{}{}{ZI}{ZI}
\end{array}
\, .
\end{align}  
Single stabilizer excitations are lineons along the $\hat{z}$ direction, and a composite formed by a pair of like excitations shifted along the $\hat{x}$ or $\hat{y}$ direction is a planon in the orthogonal plane respectively. 
This model was found with a foliated construction~\cite{shirley2018Fractional} and also via an anisotropic coupled layer construction~\cite{Fuji2019}.

\subsubsection{$X$-Cube model}
This is a fracton model with stabilizer generators given by
\begin{align}
\begin{array}{c}
\drawgenerator{IXI}{}{IXX}{IIX}{XIX}{XXX}{XII}{XXI}
\quad
\drawgenerator{}{}{}{}{IIZ}{}{IZZ}{IZI}
\quad
\drawgenerator{ZIZ}{}{IIZ}{}{}{}{}{ZII}
\quad
\drawgenerator{}{}{IZI}{ZZI}{ZII}{}{}{}
\end{array}
\, .
\end{align}
This model has particles of all sub-dimensional mobilities 0, 1 and 2. There are also $\Z_N$ versions of the $X$-Cube model~\cite{vijay2016fracton,vijay2017isotropic,PhysRevB.97.165106}, and other generalizations known as cage-nets~\cite{prem2018cage} and the related string-membrane-nets~\cite{Slagle2018foliated}. 

\begin{table}[H]
\renewcommand{\arraystretch}{1.4}{
\begin{tabular}{c|cccccc c}
Model & Deformability & $n_{\text{rods}}^{3D}$ & $n_{\text{rods}}^{2D}\left(zx,xy,yz\right)$ &
$n_{\text{rods}}^{1D}\left(x,y,z\right)$ & $n_{m}\left(zx,xy,yz\right)$ & $d$ & Type\tabularnewline
\hline
XC & $\checkmark$ & 0 & $\left(n_1,n_1,n_1\right)$ & $\left(n_2,n_2,n_2\right)$ & $\left(\ell,\ell,\ell\right)$ & 0 & foliated type-I\tabularnewline
\end{tabular}}
\caption{Operator data. For definitions, see table~\ref{table_invariants}}
\label{CC1_data}
\end{table}

\subsubsection{Checkerboard model}
This is a fracton model with stabilizer generators given by
\begin{align}
\begin{array}{c}
\drawgenerator{X}{X}{X}{X}{X}{X}{X}{X}
\quad
\drawgenerator{Z}{Z}{Z}{Z}{Z}{Z}{Z}{Z}
\end{array} \, ,
\end{align}
where the generators only appear on one colour of cells in a cubic lattice with bicoloured cells. The checkerboard model was shown to be local unitary equivalent to two copies of the $X$-cube model, and hence foliated, in Ref.~\onlinecite{shirley2018Foliated}. A twisted generalization of the checkerboard model was formulated in Ref.~\onlinecite{hao_twisted}. 

\begin{table}[H]
\renewcommand{\arraystretch}{1.4}{
\begin{tabular}{c|cccccc c}
Model & Deformability & $n_{\text{rods}}^{3D}$ & $n_{\text{rods}}^{2D}\left(zx,xy,yz\right)$ &
$n_{\text{rods}}^{1D}\left(x,y,z\right)$ & $n_{m}\left(zx,xy,yz\right)$ & $d$ & Type\tabularnewline
\hline
CB & $\checkmark$ & 0 & $\left(n_1,n_1,n_1\right)$ & $\left(n_2,n_2,n_2\right)$ & $\left(\ell,\ell,\ell\right)$ & 0 & foliated type-I\tabularnewline
\end{tabular}}
\caption{Operator data. For definitions, see table~\ref{table_invariants}}
\label{CC1_data}
\end{table}

\subsubsection{HH  Type-I model}
This is a fracton model that is equivalent to two decoupled copies of the checkerboard model after bicoloring the sites and applying swap to qubits on sites of one color. The stabilizer generators are given by
\begin{align}
\begin{array}{c}
\drawgenerator{IX}{XI}{XI}{IX}{XI}{IX}{IX}{XI}
\quad
\drawgenerator{IZ}{ZI}{ZI}{IZ}{ZI}{IZ}{IZ}{ZI}
\end{array}
\, .
\end{align}

\begin{table}[H]
\renewcommand{\arraystretch}{1.4}{
\begin{tabular}{c|cccccc c}
Model & Deformability & $n_{\text{rods}}^{3D}$ & $n_{\text{rods}}^{2D}\left(zx,xy,yz\right)$ &
$n_{\text{rods}}^{1D}\left(x,y,z\right)$ & $n_{m}\left(zx,xy,yz\right)$ & $d$ & Type\tabularnewline
\hline
HH-I & $\checkmark$ & 0 & $\left(n_1,n_1,n_1\right)$ & $\left(n_2,n_2,n_2\right)$ & $\left(\ell,\ell,\ell\right)$ & 0 & foliated type-I\tabularnewline
\end{tabular}}
\caption{Operator data. For definitions, see table~\ref{table_invariants}}
\label{CC1_data}
\end{table}

\subsubsection{Chamon's model}
This is the first topological fracton model to appear in the literature~\cite{chamon2005quantum}. It is expected to admit a foliation structure but this has not been shown. The stabilizer generators are given by
\begin{align}
\begin{array}{c}
\drawgenerator{}{}{IXII}{}{IIYI}{IXYZ}{}{IIIZ}
\quad
\drawgenerator{}{XIII}{}{}{IIZI}{}{XIZY}{IIIY}
\quad
\drawgenerator{YZIX}{YIII}{IZII}{}{}{}{}{IIIX}
\quad
\drawgenerator{}{ZIII}{IYII}{ZYXI}{IIXI}{}{}{}
\end{array}\, ,
\end{align}
which can also be written in terms of a single stabilizer generator by making a different choice of lattice vectors 
\begin{align}
\begin{array}{c}
\drawgenerator{Z}{I}{X}{Y}{Z}{I}{X}{Y}
\end{array}\, .
\end{align}
Some lineon string operators are given by
$X[0\bar{1}1]_{(000)}$, $Y[1\bar{1}0]_{(000)}$ and $Z[010]_{(000)}$. 
Planon string operators for excitations with mobility in the $\hat{i}$-plane appear at the boundary of a product of stabilizers over a square in the $\hat{i}$ plane of the dual lattice. The lineon and planon operators are shown in Fig.~\ref{chamon_exc}. 

\begin{figure}
    \centering
\sidesubfloat[]{\includegraphics[scale=0.22]{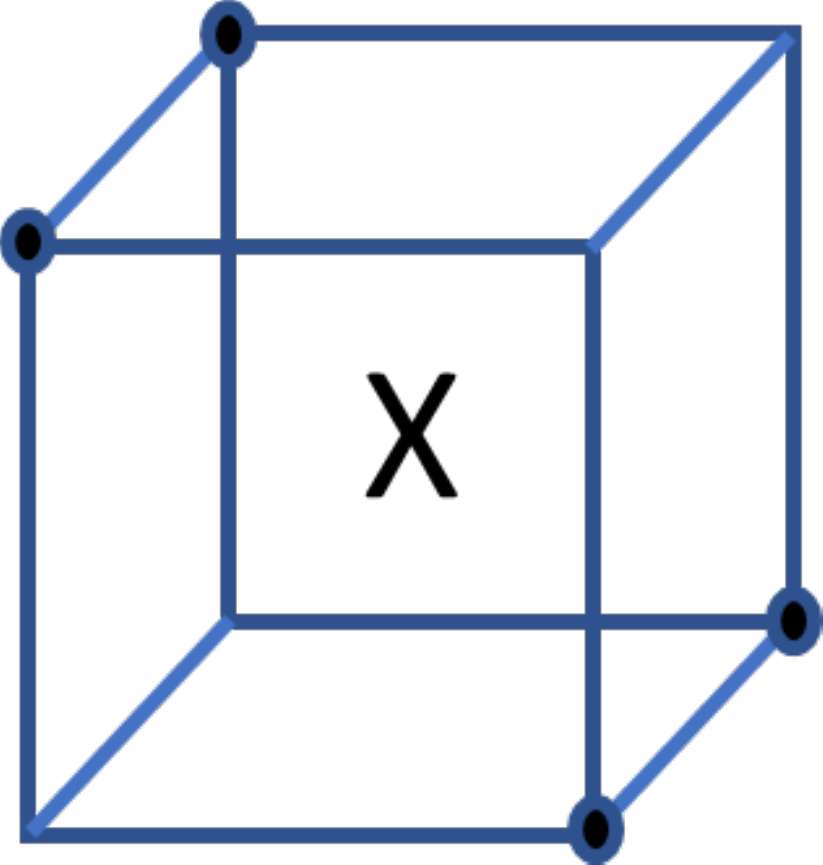}}
\sidesubfloat[]{\includegraphics[scale=0.22]{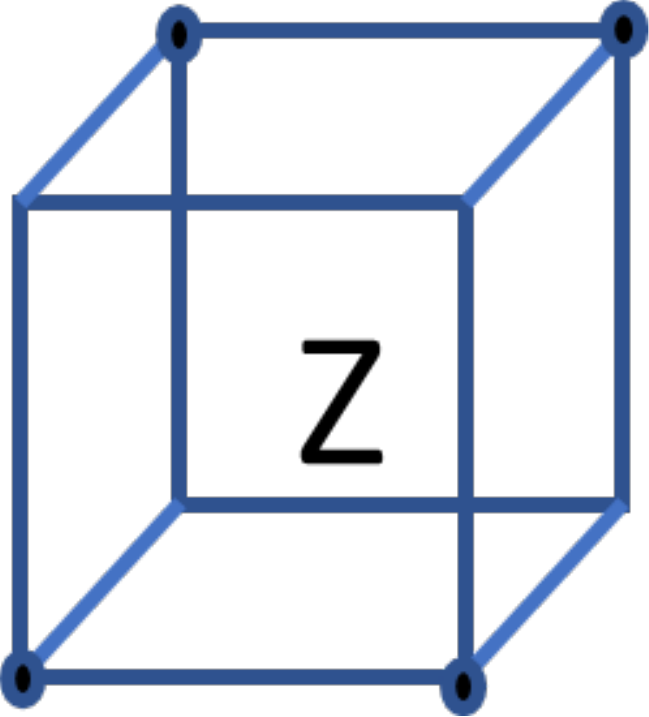}}    
\sidesubfloat[]{\includegraphics[scale=0.22]{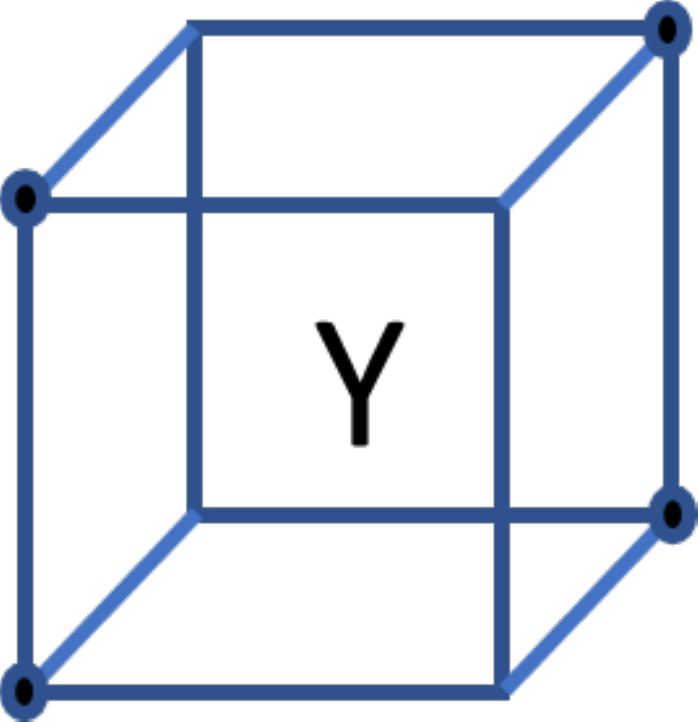}}\\    
\sidesubfloat[]{\includegraphics[scale=0.22]{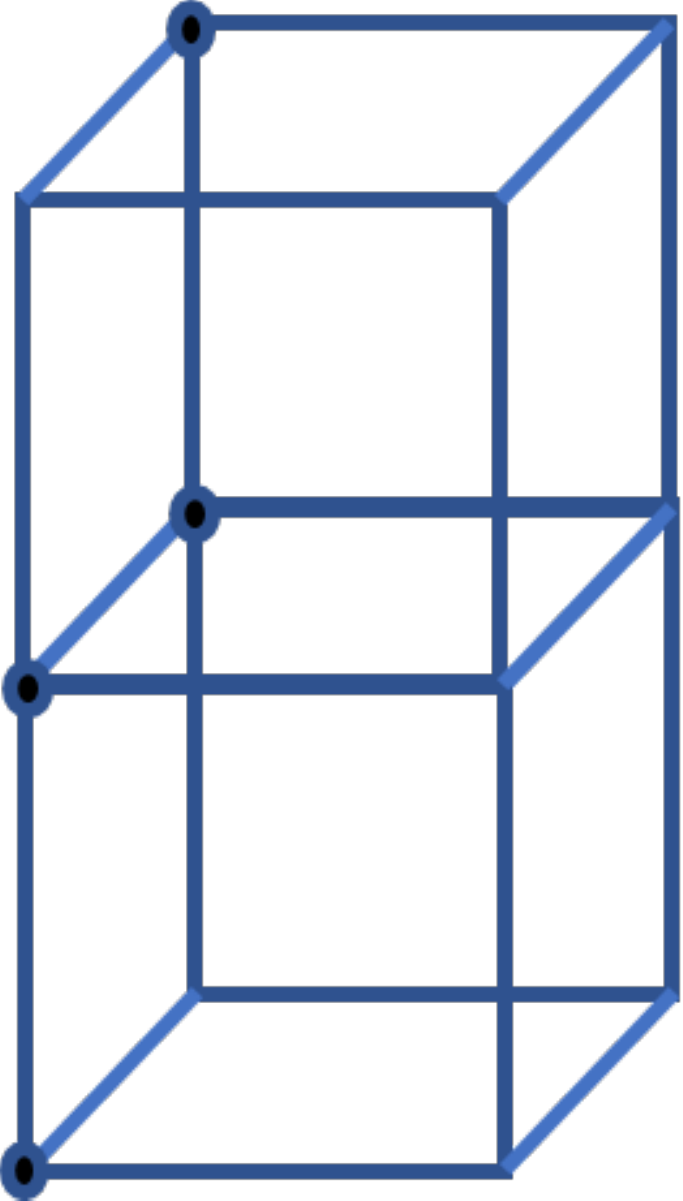}} 
\sidesubfloat[]{\includegraphics[scale=0.22]{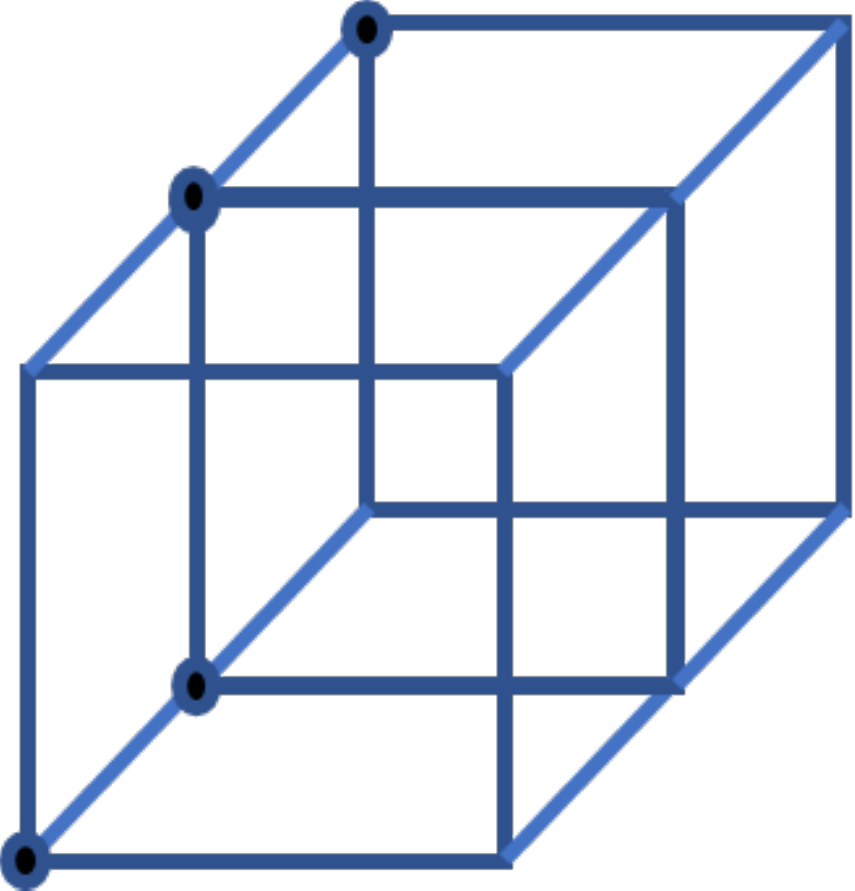}} 
\sidesubfloat[]{\includegraphics[scale=0.22]{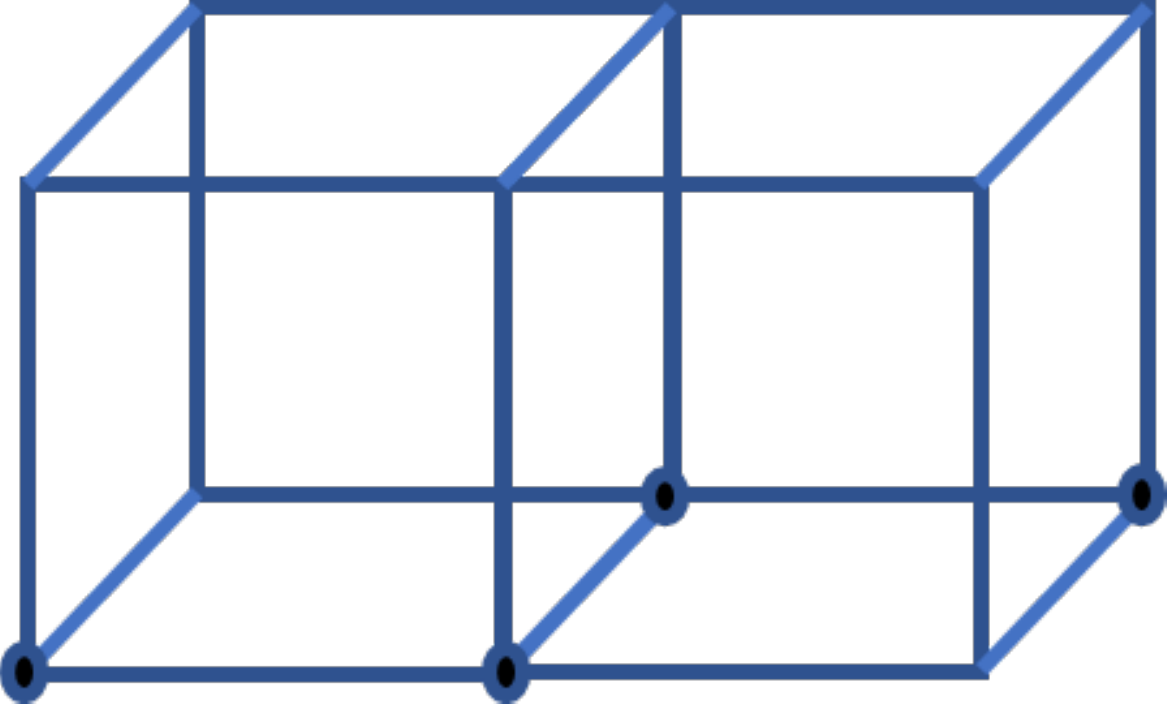}} 
\caption{Lineon and planon operators in Chamon's model. a) The excitation pattern for an $X$ operator, two of the excitations shown on the $yz$ plane form a lineon that moves along the $\hat{x}$ direction. b) The excitation pattern for a $Z$ operator, two excitations along $xz$ plane form a lineon that moves along the $\hat{y}$ direction c) The excitation pattern for a $Y$ operator, two excitations along the $xy$ plane form a lineon that moves along the $\hat{z}$ direction. (d)-(f) Excitations that form composite planons in the $xy$, $yz$ plane and $xz$ plane respectively.} 
    \label{chamon_exc}
\end{figure}

\begin{table}[H]
\renewcommand{\arraystretch}{1.4}{
\begin{tabular}{c|cccccc c}
Model & Deformability & $n_{\text{rods}}^{3D}$ & $n_{\text{rods}}^{2D}\left(zx,xy,yz\right)$ &
$n_{\text{rods}}^{1D}\left(x,y,z\right)$ & $n_{m}\left(zx,xy,yz\right)$ & $d$ & Type\tabularnewline
\hline
Chm & $\checkmark$ & $0$ & $\left(n_1,n_1,n_1\right)$ & $\left(n_2,n_2,n_2\right)$ & $\left(\ell,\ell,\ell\right)$ & 0 & foliated type-I \tabularnewline
\end{tabular}}
\caption{Operator data. For definitions, see table~\ref{table_invariants}}
\label{CC1_data}
\end{table}

\subsubsection{CSS model derived from Chamon's model}
This is a CSS model with the same charge module/structure as Chamons model~\cite{vijay2016fracton}. The stabilizer generators are given by
\begin{align}
\begin{array}{c}
\drawgenerator{IX}{II}{XI}{XX}{IX}{II}{XI}{XX}
\quad 
\drawgenerator{ZI}{II}{IZ}{ZZ}{ZI}{II}{IZ}{ZZ}
\end{array}\, .
\end{align}
It is expected to have the same foliation structure, if any, as Chamon's model. 

\subsubsection{Membrane coupled $X$-Cube model}

This is a model derived from coupling four copies of $X$-cube~\cite{PhysRevB.95.245126}. The four copies of $X$-Cube are arranged on the edges of four cubic sublattices of the face centered cubic lattice, such that an edge from three distinct cubic sublattices intersect on every edge and cube. We define a new lattice where these intersections of edges are the sites, the stabilizers are given by 
\begin{align}
\begin{array}{c}
\xymatrix@!0{%
 && III \ar@{-}[rr]\ar@{-}[dl]  && IXI &&  \\
& IIX && XII  \ar@{-}[ll]\ar@{-}[dd]\ar@{-}[ur] &&   \\
 && III \ar@{.}[uu]\ar@{.}[rr]\ar@{.}[dl] && IIX \ar@{-}[uu]\ar@{-}[rr] && XII   \\
& IXI \ar@{-}[uu]\ar@{-}[dl] && III  \ar@{-}[ll]\ar@{-}[dd]\ar@{-}[ur] && IXI  \ar@{-}[ll]\ar@{-}[dd]\ar@{-}[ur] &&   \\
XII && IIX \ar@{-}[ll]\ar@{-}[dd]\ar@{-}[ur] && III \ar@{.}[uu]\ar@{.}[rr]\ar@{.}[dl] && III \ar@{-}[uu]\ar@{-}[dl]   \\
& III \ar@{.}[uu]\ar@{.}[rr]\ar@{.}[dl] && XII \ar@{-}[rr]\ar@{-}[dl] && IIX \\
 III \ar@{-}[rr]\ar@{-}[uu] && IXI && 
 }
 \quad
 \drawgenerator{ZII}{}{IIZ}{}{ZII}{}{IIZ}{}
 \quad 
 \drawgenerator{}{}{IIZ}{IZI}{}{}{IIZ}{IZI}
 \end{array}\, ,
 \end{align}
 which is equivalent to four copies of $X$-Cube, where translations act by permuting the copies. The $X$-Cubes are then coupled with an on-site $XXX$ field, leading to the stabilizer generators
\begin{align}
\begin{array}{c}
\xymatrix@!0{%
 && II \ar@{-}[rr]\ar@{-}[dl]  && XI &&  \\
& IX && XX  \ar@{-}[ll]\ar@{-}[dd]\ar@{-}[ur] &&   \\
 && II \ar@{.}[uu]\ar@{.}[rr]\ar@{.}[dl] && IX \ar@{-}[uu]\ar@{-}[rr] && XX   \\
& XI \ar@{-}[uu]\ar@{-}[dl] && II  \ar@{-}[ll]\ar@{-}[dd]\ar@{-}[ur] && XI  \ar@{-}[ll]\ar@{-}[dd]\ar@{-}[ur] &&   \\
XX && IX \ar@{-}[ll]\ar@{-}[dd]\ar@{-}[ur] && II \ar@{.}[uu]\ar@{.}[rr]\ar@{.}[dl] && II \ar@{-}[uu]\ar@{-}[dl]   \\
& II \ar@{.}[uu]\ar@{.}[rr]\ar@{.}[dl] && XX \ar@{-}[rr]\ar@{-}[dl] && IX \\
 II \ar@{-}[rr]\ar@{-}[uu] && XI && 
 }
%\xymatrix@!0{%
%  && 002 \ar@{.}[rr]\ar@{.}[dl] && 012 \ar@{.}[rr] && 022  \\
% & 102\ar@{.}[dl] && 112 && 122  \\
% 202 && 001 \ar@{.}[uu]\ar@{.}[rr]\ar@{.}[dl] && 011 \ar@{.}[uu]\ar@{.}[rr] && 021 \ar@{.}[uu]  \\
% & 101 \ar@{.}[uu]\ar@{.}[dl] && 111 && 121 &&   \\
% 201 \ar@{.}[uu] && 000 \ar@{.}[uu]\ar@{.}[rr]\ar@{.}[dl] && 010 \ar@{.}[uu]\ar@{.}[rr]\ar@{.}[dl] && 020 \ar@{.}[uu]\ar@{.}[dl]  \\
% & 100 \ar@{.}[uu]\ar@{.}[rr]\ar@{.}[dl] && 110 \ar@{.}[rr]\ar@{.}[dl] && 120 \ar@{.}[dl]  \\
% 200 \ar@{.}[uu]\ar@{.}[rr] && 210 \ar@{.}[rr] && 220
%}
\quad
\xymatrix@!0{%
 && II \ar@{-}[rr]\ar@{-}[dl]  && IZ &&  \\
& ZI && ZZ  \ar@{-}[ll]\ar@{-}[dd]\ar@{-}[ur] &&   \\
 && II \ar@{.}[uu]\ar@{.}[rr]\ar@{.}[dl] && ZI \ar@{-}[uu]\ar@{-}[rr] && ZZ   \\
& IZ \ar@{-}[uu]\ar@{-}[dl] && II  \ar@{-}[ll]\ar@{-}[dd]\ar@{-}[ur] && IZ  \ar@{-}[ll]\ar@{-}[dd]\ar@{-}[ur] &&   \\
ZZ && ZI \ar@{-}[ll]\ar@{-}[dd]\ar@{-}[ur] && II \ar@{.}[uu]\ar@{.}[rr]\ar@{.}[dl] && II \ar@{-}[uu]\ar@{-}[dl]   \\
& II \ar@{.}[uu]\ar@{.}[rr]\ar@{.}[dl] && ZZ \ar@{-}[rr]\ar@{-}[dl] && ZI \\
 II \ar@{-}[rr]\ar@{-}[uu] && IZ && 
 }
\end{array}\, ,
\end{align}
within the subspace where $XXX=1$, to leading order in perturbation theory.  
This model is expected, but not shown, to be foliated.

\subsubsection{HHB  model A }
The following model is also expected to be foliated from the analysis in Ref.~\onlinecite{HHB_models} but it has not been shown. 
\begin{align}
\begin{array}{c}
\xymatrix@!0{%
 && IX \ar@{-}[rr]\ar@{-}[dl] && IX   \\
& IX && XX \ar@{-}[ur]\ar@{-}[ll]\ar@{-}[dd]   \\
&& IX \ar@{.}[uu]\ar@{.}[rr]\ar@{.}[dl] && XX \ar@{-}[uu]\ar@{-}[dl]   \\
& XX \ar@{-}[uu]\ar@{-}[rr] && IX \ar@{-}[rr] && XI    \\
 && XI \ar@{-}[ur]   \\
& && XI\ar@{-}[uu] 
}
\quad
\xymatrix@!0{%
&  && IZ \ar@{-}[dd]   \\
&& && IZ \ar@{-}[dl]   \\
& IZ \ar@{-}[rr] && ZI \ar@{-}[rr] && ZZ    \\
 && ZZ \ar@{-}[ur] && ZI \ar@{-}[ll]\ar@{-}[ur]   \\
& && ZZ\ar@{.}[uu]\ar@{.}[dl]\ar@{.}[rr] && ZI \ar@{-}[uu] \\
&& ZI \ar@{-}[uu] && ZI \ar@{-}[uu]\ar@{-}[ll]\ar@{-}[ur]
}
\end{array}
\, .
\end{align}

\subsubsection{HHB  model B}
Similarly for the following model, also from Ref.~\onlinecite{HHB_models}
\begin{align}
\begin{array}{c}
\xymatrix@!0{%
 && Z \ar@{-}[rr]\ar@{-}[dl] && Y   \\
& Y && Z \ar@{-}[ur]\ar@{-}[ll]\ar@{-}[dd]   \\
&& Z \ar@{.}[uu]\ar@{.}[rr]\ar@{.}[dl] && Y \ar@{-}[uu]\ar@{-}[dl]\ar@{-}[rr] && X   \\
& X \ar@{-}[uu]\ar@{-}[rr] && Y  && Z \ar@{-}[ll]\ar@{-}[ur] \ar@{-}[dd]  \\
 && && Z \ar@{.}[uu]\ar@{.}[rr]\ar@{.}[dl] && Y \ar@{-}[uu]\ar@{-}[dl]  \\
& && Y \ar@{-}[uu] \ar@{-}[rr]  && Z
}
\end{array}\, .
\end{align}

\subsection{TQFT models} 
\label{tqftzoo}

In this section we summarize the TQFT stabilizer models in 3D: toric code with bosonic~\cite{3D_TC_model_1} or fermionic point particle~\cite{Levin_wen_fermion} and the 3-fermion Walker-Wang model~\cite{walker2012} which is subtly nontrivial. 
 
\subsubsection{3D toric code with bosonic charge}
The stabilizer generators of the 3D toric code with bosonic charges are given by
\begin{align}
\begin{array}{c}
\drawgenerator{IXI}{XXX}{}{IIX}{}{}{XII}{}
\quad
\drawgenerator{}{}{}{}{IIZ}{IZZ}{}{IZI}
\quad
\drawgenerator{}{}{IIZ}{}{}{ZIZ}{}{ZII}
\quad
\drawgenerator{}{}{IZI}{}{ZII}{ZZI}{}{}
\end{array}\, .
\end{align}
We remark there are other models due to Bombin~\cite{Bombin2007} and Kim~\cite{Kim2010}, equivalent to copies of 3D toric code, based on lattices with special colorability properties. 
The 3D toric code can be generalized to $\Z_N$ and further other groups and twisted by 4-cocycles via Dijkgraaf-Witten models~\cite{DijkgraafWitten}. 

\begin{table}[H]
\renewcommand{\arraystretch}{1.4}{
\begin{tabular}{c|cccccc c}
Model & Deformability & $n_{\text{rods}}^{3D}$ & $n_{\text{rods}}^{2D}\left(zx,xy,yz\right)$ &
$n_{\text{rods}}^{1D}\left(x,y,z\right)$ & $n_{m}\left(zx,xy,yz\right)$ & $d$ & Type\tabularnewline
\hline
3DTC & $\checkmark$ & 1 & $\left(1,1,1\right)$ & $\left(1,1,1\right)$ & $\left(0,0,0\right)$ & 3 & TQFT\tabularnewline
\end{tabular}}
\caption{Operator data. For definitions, see table~\ref{table_invariants}}
\label{CC1_data}
\end{table}

\subsubsection{3D toric code with fermionic charge}
There are also 3D toric codes with fermionic charges. A particular Hamiltonian realization of this is due to Levin and Wen~\cite{Levin_wen_fermion} and its stabilizer generators are given by
%
% \begin{align}
% \begin{array}{c}
% \xymatrix@!0{%
% & XY \ar@{-}[rr] && IZ   \\
%     \\
% & ZI \ar@{-}[uu]\ar@{-}[rr] && YX \ar@{-}[uu] 
% }
% \quad
% \xymatrix@!0{%
% & ZY \ar@{-}[dl]   \\
%  IX  \\
% & XI \ar@{-}[uu] \ar@{-}[dl]  \\
%  YZ \ar@{-}[uu]
% }
% \quad
% \xymatrix@!0{%
% & YI \ar@{-}[rr]\ar@{-}[dl] && ZX \ar@{-}[dl]   \\
%  XZ \ar@{-}[rr] && IY 
% }
% \end{array}
% \end{align}
%
\begin{align}
\begin{array}{c}
\drawgenerator{}{}{}{}{YX}{ZI}{IZ}{XY}
\quad
\drawgenerator{IX}{}{YZ}{}{}{XI}{}{ZY}
\quad
\drawgenerator{}{}{XZ}{IY}{ZX}{YI}{}{}
\end{array}\, .
\end{align}
There is also a version due to Walker and Wang~\cite{walker2012} with the following stabilizer generators
\begin{align}
\begin{array}{c}
\drawgenerator{IXI}{XXX}{}{IIX}{}{}{XII}{}
\quad
\drawgenerator{IZI}{XII}{IZZ}{IIZ}{}{XII}{}{}
\quad
\drawgenerator{}{IXI}{}{IIZ}{ZIZ}{IXI}{ZII}{}
\quad
\drawgenerator{IZI}{IIX}{}{}{}{IIX}{ZII}{ZZI}
\end{array}\, .
\end{align}
These Hamiltonians realize the simplest case of discrete gauge theory with fermions. More general discrete gauge theories with fermions have been argued to cover the bulk topological excitations of any 3D topological order in a qudit commuting projector model~\cite{lan2017classification,Lan2019}. 
In Ref.~\onlinecite{Levin_wen_fermion} it was explained how to verify if a string operator moves a particle with fermionic or bosonic self-statistics. This can be implemented numerically to distinguish whether a 3D point particle is fermionic or bosonic. 
More generally the test can be applied to check whether any string operator moves a bosonic or fermionic point particle. This may be relevant for a more fine grained sorting of fractonic topologial stabilizer models, particularly nonCSS models, in the future.

\subsubsection{3-fermion Walker-Wang model}
The 3-fermion Walker-Wang model~\cite{walker2012} has 3 edges per vertex and 2 qubits per edge, labeled by 1 and 2. The stabilizer generators after coarse-graining the 6 qubits adjacent to each vertex onto a single site, are given by
\begin{align}
\begin{array}{c}
\drawgenerator{IX_1I}{X_1X_1X_1}{}{IIX_1}{}{}{X_1II}{}
\quad
\drawgenerator{IZ_1I}{X_1 II}{IZ_1Z_1}{IIZ_1}{}{X_1 X_2 II}{}{}
\quad
\drawgenerator{}{IX_1I}{}{IIZ_1}{Z_1 IZ_1}{IX_1 X_2 I}{Z_1 II}{}
\quad
\drawgenerator{IZ_1 I}{IIX_1 }{}{}{}{IIX_1 X_2}{Z_1 II}{Z_1 Z_1 I}
\\
\drawgenerator{IX_2I}{X_2X_2 X_2 }{}{IIX_2 }{}{}{X_2 II}{}
\quad
\drawgenerator{IZ_2 I}{X_1 X_2 II}{IZ_2 Z_2 }{IIZ_2 }{}{X_2 II}{}{}
\quad
\drawgenerator{}{IX_1 X_2 I}{}{IIZ_2 }{Z_2 IZ_2 }{IX_2 I}{Z_2 II}{}
\quad
\drawgenerator{IZ_2 I}{IIX_1 X_2}{}{}{}{IIX_2}{Z_2 II}{Z_2 Z_2 I}
\end{array}\, .
\end{align}
Recently this model has been shown to be trivialized by a locality preserving Clifford unitary~\cite{Haah2018}, even though a local unitary circuit cannot trivialize the model without also creating a 2D commuting projector Hamiltonian for a chiral 3-fermion topological order. 
It remains an interesting open question to define a local unitary circuit invariant that detects this model, or equivalently the locality preserving Clifford unitary that creates it. Since there are no nontrivial topological sectors in the bulk methods such as ours do not apply. 

\end{widetext}
\end{document}